\documentclass[a4paper]{article}
\usepackage{graphicx}
\usepackage{amsmath}
\usepackage{amsthm}

\usepackage{geometry}
\usepackage{amsfonts}
\usepackage{array}
\usepackage{amssymb}
\usepackage{upgreek}
\usepackage{mathrsfs}

\usepackage{enumitem}
\usepackage{cancel}

\usepackage{pgf}
\usepackage{tikz}

\usepackage{pgfplots}
\pgfplotsset{compat=newest}
\usetikzlibrary{arrows}
\usetikzlibrary{matrix}
\usepackage{color}
\usepackage{subcaption}

\definecolor{beaublue}{rgb}{0.74, 0.83, 0.9}
\definecolor{emerald}{rgb}{0.31, 0.78, 0.47}

\usepackage{hyperref}

\usepackage{authblk}

\newtheorem{corollary}{Corollary}[section]
\newtheorem{proposition}{Proposition}[section]
\newtheorem{definition}{Definition}[section]
\newtheorem{lemma}{Lemma}[section]
\newtheorem{example}{Example}
\newtheorem{remark}{Remark}
\newtheorem{theorem}{Theorem}[section]
\numberwithin{equation}{section}

\makeatletter
\def\@maketitle{%
  \newpage
  \null
  \vskip 2em%
  \begin{center}%
  \let \footnote \thanks
    {\Large\bfseries \@title \par}%
    \vskip 1.5em%
    {\normalsize
      \lineskip .5em%
      \begin{tabular}[t]{c}%
        \@author
      \end{tabular}\par}%
    \vskip 1em%
    {\normalsize \@date}%
  \end{center}%
  \par
  \vskip 1.5em}
\makeatother
\title{\textbf{L-extensions and L-boundary of conformal spacetimes}}

\author{A. BAUTISTA \thanks{E--mail: \texttt{alfredo.bautista@uam.es}}}
\affil{Depto. de Matem\'aticas, Univ. Carlos III de Madrid, \protect\\ Avda. de la Universidad 30, 28911 Legan\'es, Madrid, Spain , and \protect\\ Depto. de An\'alisis Econ\'omico: Econom\'{\i}a Cuantitativa, Univ. Aut\'onoma de Madrid \protect\\ C/ Francisco Tom\'as y Valiente 5, 28049 Madrid, Spain.} 

\author{A. IBORT \thanks{E--mail: \texttt{albertoi@math.uc3m.es}}}
\affil{Depto. de Matem\'aticas, Univ. Carlos III de Madrid \protect\\ Avda. de la Universidad 30, 28911 Legan\'es, Madrid, Spain, and \protect\\ ICMAT, Instituto de Ciencias Matem\'{a}ticas (CSIC-UAM-UC3M-UCM) \protect\\ C/ Nicol\'as Cabrera, 13-15, 28049, Madrid, Spain. }

\author{J. LAFUENTE \thanks{E--mail: \texttt{jlafuente@mat.ucm.es}}}
\affil{Depto. de Geometr\'{\i}a y Topolog\'{\i}a, and Instituto de Matem\'{a}ticas Interdisciplinar (IMI), Univ. Complutense de Madrid \protect\\ Avda. Complutense s/n, 28040 Madrid, Spain.} 

\date{October, 2018}

\begin{document}

\maketitle


\begin{abstract}   The notion of L-boundary, a new causal boundary proposed by R. Low based on constructing a `sky at infinity' for any light ray, is discussed in detail.  The analysis of the notion of L-boundary will be done in the 3-dimensional situation for the ease of presentation.    The proposed notion of causal boundary is intrinsically conformal and, as it will be proved in the paper, under natural conditions provides a natural extension $\overline{M}$ of the given spacetime $M$ with smooth boundary $\partial M = \overline{M} \backslash M$.  The extensions $\overline{M}$ of any conformal manifold $M$ constructed in this way are characterised exclusively in terms of local properties at the boundary points.  Such extensions are called L-extensions and it is proved that, if they exist, they are essentially unique.  Finally it is shown that in the 3-dimensional case, any L-extension is equivalent to the canonical extension obtained by using the L-boundary of the manifold.

\end{abstract}



\section{Introduction}

In his seminal paper about the conformal treatment of infinity \cite{Pe11}, R. Penrose argued that in order to deal with the properties of the fields at null infinity we should seek for an extension of the conformal structure of spacetime across null infinity.  In this way, in specific examples, a geometric notion of asymptotic flatness at null infinity was provided that allows to be represented by a conformal boundary, a hyper-surface in the extended spacetime.   Then the asymptotic behavior of fields that satisfy conformally covariant equations is greatly simplified if the spacetime admits a conformal extension of sufficient smoothness, that is, asymptotically simple.  Since then a large body of research has been devoted to exploit these ideas that allowed to understand better the asymptotic behaviour of solutions of Einstein's field equations.   A large body of literature has been devoted to study for various explicit  solutions whether they admit the required conformal extension (see for instance \cite{Pe84}, \cite{Ad09} and references therein).  It must be pointed out though that it is not understood yet which Cauchy data which are asymptotically flat in the sense of the standard Cauchy problem at space-like infinity evolve into asymptotically simple solutions (see for instance \cite{Fr04}).

Soon after the introduction of conformal boundaries to understand the physical properties of fields at infinity, Geroch, Kronheimer and Penrose again, in the striking paper \cite{Ge72} introduced the notion of \textit{ideal points} as a way to deal with singularities in models of space--time constructed according to the general theory of relativity, thus as \textit{the singularities themselves cannot be regarded as actually belonging to the manifold, we are led to consider methods of constructing additional ideal points which, when adjoined to $M$, result in a unified structure $\overline{M}$ incorporating `singular' as well as `non-singular' points} \cite{Ge72}.    The boundary constructed according with the ideas expressed in the previous work, the GKR-boundary, causal boundary, or just the c-boundary for short, is intrinsic to the given spacetime and conformally invariant and, as a difference with the conformal one, the c-boundary only takes into account for its construction time-like curves and directions.   

However the notion of the c-boundary itself was not free from  difficulties and controversies.  We will just mention the so called identification problem between future and past preboundary points that affects the selection of a natural topology for it (see for instance \cite{Sa09} and references therein for a detailed review on the subject).   In spite of all this, the incorporation of ideas by Marolf and Ross \cite{Ma03} allowed Flores, Herrera and S\'anchez \cite{Fl11} to put all ingredients together to prove that there is an (essentially unique) choice for the c-boundary which is consistent with the conformal boundary in the natural cases, giving a definitive support to both boundaries and closing in this way a long standing debate.

Thus one of the main achievements of the aforementioned work is to provide general conditions that guarantee that the (accessible) conformal boundary agrees with the c-boundary and the obtention of computable conditions for the conformal boundary points which are $C^1$ \cite[Sects. 4.2, 4.3]{Fl11}.    The notion of conformal envelopments $i \colon M \to M_0$ allows them to discuss the possible chronological and causal relations definable in the conformal boundary $\partial_iM$ associated to it.    In particular, the conditions found by Flores \textit{et al} that guarantee that the accessible part of the conformal boundary will coincide with the c-boundary (regular accessibility) will allow to prove that the absence of time-like points at the boundary implies the equivalence of the boundaries and is equivalent to the global hyperbolicity of the spacetime, hence for globally hyperbolic spacetimes with $C^1$ conformal boundary the conformal and causal boundaries are equivalent however, as the examples discussed in \cite{Ha17} show there are globally hyperbolic spaces without a conformal boundary.
Thus some additional conditions must be imposed on the conformal envelopments, which must be constructed on each instance,  in order to obtain a satisfactory conformal boundary and the smoothability of the boundary plays an important role in this situation.

In this work a new approach to the construction of a smooth conformal/causal boundary for a strongly causal spacetime $M$ is considered.  This new boundary, called in what follows L-boundary  because of its proponent R. Low and because it relies on an imaginative use of light rays, is profoundly inspired by Penrose's twistor program.  Actually, the main idea comes from considering the space of light rays instead of the space of events in spacetime as the main object in the analysis of causality.   In fact, spacetime events can be identified with the congruence of light rays arriving to it, the so called \textit{sky} of the event and, in the case that there is a one-to-one correspondence between events and skies (such spacetimes are said to be sky separating), we may try to study the causality (and other physical aspects) of the theory by studying instead the space of light rays $\mathcal{N}$ and the space of skies $\Sigma$ lying on it.    A number of conjectures regarding this program were raised by R. Low and others (see for instance the results on the beautiful conjectures on the relation between causality and contact and symplectic geometry in \cite{Ch08}, \cite{Ch10}).  The consistency of the program was proved when various reconstruction Theorems were obtained showing that under natural and rather mild conditions the topological, differential and causal properties of the original spacetime are fully characterized in terms of appropriate topological, differential and geometrical structures on the space of light rays and a family of skies \cite{Ba14}, \cite{Ba15}. 

Thus the development of a topological characterization of causality relations in the space of light rays started by R. Low in \cite{Lo88} (see also \cite{Lo89}, \cite{Lo90b}, \cite{Lo90}, \cite{Lo94}) led the author to sketch a new definition of a causal boundary for a strongly causal spacetime  by considering the problem of attaching a future endpoint to a null geodesic $\gamma$ in the space of light rays $\mathcal{N}$ of the given spacetime \cite{Lo06}. The main idea is to treat all null geodesics which `focus' at the same point at infinity as the `sky' of the (common) future endpoint of these null geodesics.   In the recent paper \cite{Ba17} a precise definition of the L-boundary  was presented for 3-dimensional spaces-times as well as some preliminary properties and examples, among them some results discussing its relation to the c-boundary discussed before.   The obtained results were encouraging enough to continue the study of this new notion of causal boundary because, if it exists, it has a bundle of interesting properties that complement in a natural way the discussion above on the relation between the conformal and the c-boundary.   To begin with the L-boundary  is intrinsically conformal and is formulated entirely in terms of the space of light rays $\mathcal{N}$ of the spacetime $M$.  Secondly, and this is one of the main results of the present paper, under natural conditions if it exists is smooth, providing a natural differentiable framework to the construction of conformal envelopments needed for the setting of conformal boundaries discussed before.    Moreover the construction of the L-boundary  is explicit and provides a beautiful bridge between the structures present in the space of events $M$ and the space of light rays $\mathcal{N}$, opening the road to a new understanding of the relations between topology and causality outlined above.    Hence in the present paper the differentiable foundations for the construction of the L-boundary  are laid and a number of results are obtained, among them the existence, under suitable natural conditions, of a class of extensions of the spacetime, called in the paper L--extensions, that are the natural candidates for conformal envelopments.

The paper will be organised as follows.  Sect. \ref{sec:Preliminaries} will be devoted to succinctly review the main notions and notations regarding the space of light rays of a spacetime and their skies.   In Sect. \ref{sec:lboundary} the idea behind the notion of the L-boundary sketched before will be revised in depth and its relation to the blowing up and down techniques in algebraic and symplectic geometry will be discussed, in particular the blow up space $\widetilde{N}$ of a spacetime $M$ will be defined. Moreover the main conditions satisfied for the spacetimes considered in this paper will be clearly established and a preliminar notion of L-spacetimes will be stated.  The example of a globally hyperbolic block in 3-dimensional Minkowski space will be thoroughly worked out as a sort of roadmap that could help the reader with the more technical aspects of the theory developed in subsequent sections.

Section \ref{sec:charts} will be devoted to introduce various local descriptions for the ambient space $\mathbb{P}(\mathcal{H})$, the projectivisation of the contact structure on the space of light rays, that will be used widely in the rest of the paper.    In Sect. \ref{sec:proj-param} it will be shown that there exists a projective parametrisation for light rays that will be used sparingly in the constructions to follow.    In Sect. \ref{sec:canon-ext} the canonical extension of 3-dimensional L-manifolds will be constructed as the quotient of the natural closure $\widetilde{\mathcal{N}}$ of the extended space of light rays with respect to the natural distributions $\oplus$ and $\ominus$ defined by the tangent spaces to skies at infinity.  This section is the most technical part of the paper and where the main Theorem will be proved, Thm. \ref{main_thm}, that shows that under the natural conditions stated in Sect. \ref{sec:lboundary}, the canonical distribution $\widetilde{\mathcal{D}}$ on the blow up space $\widetilde{\mathcal{N}}$, defined by the tangent spaces to the skies,  is extended smoothly to its boundary.  This result, together with a compactness assumption on the leaves of the boundary distribution, will allow to show that the extended spacetime $\overline{M}$ obtained glueing the L-boundary to $M$ is a smooth manifold with boundary.  Finally, in Sect. \ref{sec:lextensions} the general notion of L-extension will be introduced and their main properties discussed, in particular it will be shown that L-extensions are essentially unique and that the canonical extension defined by the L-boundary of a L-manifold is a L-extension.  Some examples showing that the conditions used in characterising L-extensions cannot be relaxed will also be discussed. 


\section{The space of light rays $\mathcal{N}$ of a spacetime $M$ and other background notions and  notations}\label{sec:Preliminaries}

We will summarise first some basic facts on the space of light rays of a spacetime of dimension $m>2$ in order to introduce the objects subject of this work (see  \cite{Ba14}, \cite{Ba15} and \cite{Ba17} and references therein for more details). 

Consider a time-oriented $m$-dimensional conformal Lorentz manifold $(M,\mathcal{C})$,  that is, a time-oriented $m$--dimensional Hausdorff smooth manifold $M$ equipped with a conformal class $\mathcal{C} = \{ \mathbf{g} = \lambda \mathbf{g}_0 \mid \lambda > 0 \}$ of Lorentz metrics $\mathbf{g}$. 
We define the space of light rays $\mathcal{N}$ corresponding to $(M,\mathcal{C})$ by 
\[
\mathcal{N}=\{\gamma\left(I\right)\subset M \mid \, \gamma:I\rightarrow M \text{ is a maximal null geodesic } \} \, ,
\]
that is, as the set of all images of maximal null geodesics. As a consequence of \cite[Lem.~2.7]{Mi08} and \cite[Lem.~2.1]{Ku88}, we get that any null geodesic $\gamma=\gamma\left(t\right)$ for the metric $\mathbf{g}_0\in\mathcal{C}$ is a null pregeodesic for any other metric $\mathbf{g}\in \mathcal{C}$ and this implies that $\mathcal{N}$ does not depend on any particular metric in $\mathcal{C}$ or, in other words, it depends just on the conformal structure of the Lorenztian manifold $M$.
Each one  of these images is called a \emph{light ray} and, from the definition of $\mathcal{N}$, one can interpret a light ray as an unparametrised null geodesic. 
If there is no risk of confusion we will use the same greek letter, usually $\gamma$, to denote both an element of $\mathcal{N}$, $\gamma\in \mathcal{N}$, and the image in $M$ of the corresponding maximal null geodesic, $\gamma\subset M$. 

Actually, for $M$ strongly causal, the hausdorffness of $\mathcal{N}$ is equivalent (see \cite[Sec.~3]{Lo90b}) to the null pseudo--convexity of $M$, that is, for any compact $K\subset M$ there exist a compact $K'\subset M$ such that any segment of light ray with endpoints in $K$ is contained in $K'$.

It is possible to equip $\mathcal{N}$ with suitable topological and differentiable structures, by using coordinate charts of subbundles of the tangent bundle $TM$, as done in \cite[Sec.~2.3]{Ba14}.
Indeed, if we fix an auxiliary metric $\mathbf{g}\in \mathcal{C}$, then we can define the subbundle of future light-like vectors on $M$ by $\mathbb{N}^{+}=\{v \in TM:\mathbf{g}\left( v ,v \right) =0,v \neq 0,v
\,\,\mathrm{future}\}\subset TM$. 
Its fibre at $p\in M$ will be denoted by $\mathbb{N}^{+}_p$. 
We will denote by $\mathbb{N}^{+}\left(W\right)$ the restriction of $\mathbb{N}^{+}$ to some given set $W\subset M$, and by 
\begin{equation}\label{PNW}
\mathbb{PN}\left(W\right)=\{ \left[u\right]:  v\in\left[u\right] \Leftrightarrow \exists \lambda>0 : v=\lambda u\in \mathbb{N}^{+}\left(W\right) \} \, ,
\end{equation}
the bundle of lines in $\mathbb{N}^{+}\left(W\right)$.

Observe that two different proportional vectors $v_1 , v_2 \in \mathbb{N}^{+}_p$ define different null geodesics with the same image in $M$, therefore both $v_1$ and $v_2$ define the same light ray $\gamma\in \mathcal{N}$. 
We will denote by $\gamma_{[v]} \in \mathcal{N}$ the light ray corresponding to the image of a null geodesic $\gamma:I\rightarrow M$ such that $\gamma'\left(0\right)=v\in \mathbb{N}^{+}_{\gamma\left(0\right)}$.
Because we are assuming that $M$ is strongly causal, for any $p\in M$ there exists a globally hyperbolic, causally convex and convex normal neighbourhood $V\subset M$ with a differentiable and spacelike Cauchy surface $C\subset V$ such that any causal curve entering in $V$, intersects $C$ in a singleton (see for instance \cite{Mi08}). In particular, the intersection of any light ray passing through $V$ with $C$ is exactly one point.
Then any $\left[v\right]\in \mathbb{PN}\left(C\right)$ defines, unambiguously, a light ray passing through $V$. 
We can choose the restriction 
\[
\Omega\left(C\right)=\left\{ v\in \mathbb{N}^{+}\left(C\right):\mathbf{g}\left(v,T\right)=-1 \right\}
\]
as a model for $\mathbb{PN}\left(C\right)$, where $T\in\mathfrak{X}\left(M\right)$ is a fixed global time-like vector field. 
The submanifold $\Omega\left(C\right)$ allows to define coordinates in $\mathcal{N}_{V}=\{ \gamma\in \mathcal{N}: \gamma\cap V\neq \varnothing \} \subset \mathcal{N}$.  
Calling $\Omega\left(V\right)$ the corresponding restriction of $\mathbb{N}^{+}\left(V\right)$ analogous to $\Omega\left(C\right)$, then we have the following diagram

\begin{equation}\label{diagram-charts}
\begin{tikzpicture} [every node/.style={midway}]
\matrix[column sep={8em,between origins},
        row sep={3em}] at (0,0)
{ \node(PN1)   {$\Omega\left( V\right)$}  ; & \node(N) {$\mathcal{N}_{V}$}; \\
  \node(PN2) {$\Omega\left( C\right)$};   & \node(PNC) {$\mathbb{PN}\left(C\right)$};                \\};
\draw[->] (PN1) -- (N) node[anchor=south]  {$\boldsymbol{\upgamma}$};
\draw[->] (PN2) -- (N) node[anchor=north]  {$\xi$};
\draw[<-right hook] (PN1)   -- (PN2) node[anchor=east] {inc};
\draw[->] (PN2) -- (PNC) node[anchor=north]  {$\pi^{\mathbb{N}}_{\mathbb{PN}}$};
\draw[->] (PNC) -- (N) node[anchor=west]  {$\mu$};
\end{tikzpicture}
\end{equation}
where $\boldsymbol{\upgamma}:\Omega\left( V\right)\rightarrow \mathcal{N}_{V}$ is a submersion defined by $\boldsymbol{\upgamma}\left(u\right)=\gamma_{u}$ and the map  $\xi=\left. \boldsymbol{\upgamma}\right|_{\Omega\left( C\right)}$ is a diffeomorphism obtained by the restriction of $\boldsymbol{\upgamma}$ to the hypersurface $\Omega\left( C\right)\subset \Omega\left( V\right)$, as seen in \cite[Sec.~2.3]{Ba14}.
Notice that $\Omega\left( C\right)$ is a section of the bundle $\pi^{\mathbb{N}}_{\mathbb{PN}} : \mathbb{N}^{+}\left( C\right) \rightarrow \mathbb{PN}\left(C\right)$\footnote{We will preferably denote by $\pi^{A}_{B}$ the canonical projection of the bundle $\pi^A_B:A\rightarrow B$.} and the map $\mu : \mathbb{PN}\left(C\right) \rightarrow \mathcal{N}_{V}$ is a diffeomorphism \cite[Sec.~2.3]{Ba14}.
We will also use the notation $\gamma_{\left[v\right]} = \mu\left(\left[v\right]\right) \in \mathcal{N}$.

For any $x\in M$, the set 
\begin{equation}
S\left( x\right) = \{\gamma \in \mathcal{N}: x \in \gamma \subset M\} 
\end{equation}
will be called \emph{the sky of $x$} and consist of all light rays passing through $x$.
Since each $\left[v\right]\in \mathbb{PN}_x$ defines a light ray $\gamma_{\left[v\right]}\in \mathcal{N}$, then $S\left(x\right)$ is diffeomorphic to the standard sphere $\mathbb{S}^{m-2}$.

We will say that $M$ is \emph{sky-separating} if the \emph{sky map} $S$ from $M$ to the set of skies $\Sigma =\{X\subset \mathcal{N}:X=S\left( x\right) \,\,  \mathrm{ for \, \, some } \,\, x\in M\}$ mapping any $x\in M$ into its sky $S(x)$, is injective, that is, $S\left(x\right)=S\left(y\right)$ implies that $x=y$.
If we assume that $M$ is sky-separating, then it is possible to define the \emph{reconstructive or regular topology} \cite[Def.~1]{Ba14}, \cite[Def.~13]{Ba15}, in the set of all skies in such a way that the map $S \colon M\rightarrow \Sigma $ is a diffeomorphism \cite[Cor.~17]{Ba15}, when the differentiable structure of the \emph{space of skies} $\Sigma$  is compatible with said topology.
 
A light ray is an unparametrised curve, but if we fix an auxiliary metric $\mathbf{g}\in\mathcal{C}$, given a local Cauchy hypersurface $C$ as before, the geodesic parameter of a light ray $\gamma\in \mathcal{N}$ in the corresponding open set in $\mathcal{N}$ defined by $C$, is determined by the initial values $\gamma\left(0\right)\in C$ and $\gamma'\left(0\right)\in \Omega\left(C\right)$. 
Recall that a differentiable curve $\Gamma:\left(-\epsilon,\epsilon\right)\rightarrow \mathcal{N}$ such that $\Gamma\left(0\right)=\gamma$ defines the vector $\Gamma'\left(0\right)\in T_{\gamma}\mathcal{N}$. 
Then, assumed $\mathbf{g}\in\mathcal{C}$ is fixed, the curve $\Gamma$ corresponds to a variation $\mathbf{f}:\left(-\epsilon,\epsilon\right)\times I\rightarrow M$ of null geodesics in $M$, and $\Gamma'\left(0\right)$ can be defined by the Jacobi field defined by $\mathbf{f}$ for $s=0$, that is 
\[
J\left(\tau\right)=\left.\frac{\partial \mathbf{f}\left(s,\tau\right)}{\partial s}\right|_{s=0} \in T_{\gamma\left(\tau\right)}M   .
\]
Recall that a Jacobi field $J$ along a geodesic $\gamma$ is a vector field along $\gamma$ satisfying the differential equation 
\begin{equation}\label{eq-Jacobi-fields}
J'' = R\left(J,\gamma'\right)\gamma'
\end{equation}
such that the \textit{prime} symbol ($'$) in $J$ denotes the covariant derivative along $\gamma$ and $R$ is the curvature tensor \cite[Def.~8.2]{On83}.
When all geodesics $\gamma_s=\mathbf{f}\left(s,\cdot\right)$ of the variation $\mathbf{f}$ are such that the value of $\mathbf{g}\left(\gamma'_s , \gamma'_s\right)$ is independent of $s$ then the corresponding Jacobi field $J$ along $\gamma=\gamma_0$ satisfies the property 
\[
\mathbf{g}\left(J\left(\tau\right),\gamma'\left(\tau\right)\right)= \mathrm{constant}
\]
for all $\tau\in I$ \cite[Lem.~2.1]{Ch18}. 
We will work with variations such that $\mathbf{g}\left(\gamma'_s , \gamma'_s\right)=0$ for all $s$. 

If we choose another metric $\overline{\mathbf{g}}\in \mathcal{C}$, the light ray $\gamma\in \mathcal{N}$ can be parametrized as a null geodesic related to $\overline{\mathbf{g}}$ by $\overline{\gamma}=\overline{\gamma}\left(\overline{\tau}\right)$ where $\tau=h\left(\overline{\tau}\right)$ is the corresponding change of parameter such that $\overline{\gamma}\left(\overline{\tau}\right)=\gamma\left(h\left(\overline{\tau}\right)\right)$.
The same curve $\Gamma\subset \mathcal{N}$ defines a variation of null geodesics $\overline{\mathbf{f}}$  in $\left(M,\overline{\mathbf{g}}\right)$ such that its Jacobi field along $\overline{\gamma}=\gamma=\Gamma\left(0\right)\in \mathcal{N}$ verifies 
\[
\overline{J}\left(\overline{\tau}\right) = J\left(h\left(\overline{\tau}\right)\right) \ \left(\mathrm{mod}\ \gamma'\left(h\left(\overline{\tau}\right)\right) \right)
\]
for all $\overline{\tau}$, then the vector $\xi=\Gamma'\left(0\right)\in T_{\gamma}\mathcal{N}$ can be identified with an equivalence class of Jacobi fields on $\gamma$ given by 
\begin{equation}\label{Jmodgamma}
\langle J\rangle =J \ (\mathrm{mod} \ \gamma ^{\prime })  .
\end{equation}

The space of light rays $\mathcal{N}$ has a relevant canonical contact structure, that is, a maximal non-integrable\footnote{Let us recall that non-integrable means that given any local 1-form $\alpha$ such that locally $\mathcal{H} =\ker \alpha$, then $d\alpha$ is non-degenerate when restricted to $\mathcal{H}$.} hyperplane distribution $\mathcal{H}\subset T\mathcal{N}$, which can be described by Jacobi fields as in \cite{Lo98}, \cite{Lo06} by
\begin{equation}  \label{contact}
\mathcal{H}_{\gamma}=\lbrace \langle J\rangle\in T_{\gamma}\mathcal{N}:\mathbf{g}\left(J,\gamma ^{\prime }\right)=0 \rbrace
\end{equation}
where $\mathbf{g}\in \mathcal{C}$.   Notice that the contact structure $\mathcal{H}$ depends only on the conformal structure $\mathcal{C}$.

Given $x\in M$ then, for any parametrized $\gamma\in X=S\left(x\right)\in \Sigma$ such that $\gamma \left( s_{0}\right) =x$, we have 
\begin{equation}\label{tangent_sky-1}
T_{\gamma }X=\{\langle J\rangle\in T_{\gamma }\mathcal{N}:J\left( s_{0}\right) =0\left(
\mathrm{{mod}\gamma ^{\prime }}\right) \}  
\end{equation}%
hence, for $\langle J\rangle\in T_{\gamma }X$, we have that $\mathbf{g}\left( J,\gamma ^{\prime
}\right)=0 $ because $J\left( s_{0}\right) =0\left( \mathrm{{mod}%
\gamma ^{\prime }}\right) $ and this implies that $T_{\gamma }X$ is a $(m-2)$--dimensional subspace inside of the $(2m-4)$--dimensional vector space $\mathcal{H}_{\gamma }\subset T_{\gamma}\mathcal{N}$. 
Therefore any sky $X$ is a Legendrian submanifold of the contact structure $\mathcal{H}$ on $\mathcal{N}$.

Observe that, if $\gamma\in \mathcal{N}$ such that $x,y\in \gamma\subset M$ satisfy $0\neq \langle J\rangle \in T_{\gamma}S\left(x\right) \cap T_{\gamma}S\left(y\right)$, then $x,y$ are conjugate points of the Jacobi field $J$ along $\gamma$.
Then we will say that $M$ is \emph{light non--conjugate} if for any $x,y\in M$ such that $\gamma\in S\left(x\right) \cap S\left(y\right) \subset \mathcal{N}$ then $T_{\gamma}S\left(x\right) \cap T_{\gamma}S\left(y\right) = \left\{ 0 \right\}$. 
Along any null geodesic $\gamma$ in $\left(M,\mathbf{g}\right)$, we can always find two non--conjugate points $p,q\in\gamma$ and this means that $T_{\gamma}S\left(p\right) \cap T_{\gamma}S\left(q\right)=\{0\}$.
Since $T_{\gamma}S\left(x\right) \subset \mathcal{H}_{\gamma}$ for all $x\in \gamma$ and $\mathrm{dim}\left(T_{\gamma}S\left(p\right) \oplus T_{\gamma}S\left(q\right)\right)= \mathrm{dim}\left(\mathcal{H}_{\gamma}\right)$ then we have that
\begin{equation}\label{contact-struct}
\mathcal{H}_{\gamma} = T_{\gamma}S\left(p\right) \oplus T_{\gamma}S\left(q\right)
\end{equation} 
for any pair of non--conjugate points $p,q\in\gamma$.   Notice that because of (\ref{contact-struct}) the contact structure is spanned by tangent spaces to skies, hence again, it shows that $\mathcal{H}$ is independent from the metric in $\mathcal{C}$ and is a conformal invariant. 


\section{The L-boundary}\label{sec:lboundary}

As it was discussed in the introduction, the objective of this work is to construct a new causal boundary that complements the main features of the causal c-boundary and the conformal boundary.   Such notion of causal boundary, introduced by R. Low in  \cite{Lo06}, was started to be developed in the recent article \cite{Ba17} where some examples and preliminary results were exhibited.   The main characteristic of the L-boundary is that its construction relies only on the conformal structure of the original spacetime and, following Penrose's insight as developed by R. Low, it uses the geometrical structure of the space of light rays.  


\subsection{Blowing up and down a spacetime}\label{sec:blowing}

In this section we will offer another approach to the construction of the L-boundary, closely related the well-known `blowing up' and `blowing down'  techniques in algebraic and symplectic geometry, that could help to visualize the constructions leading to the notion of the L-boundary.  

The main idea behind the construction of points in the L-boundary can be summarised as follows.  If we identify the events $x$ of the given spacetime $M$ with their corresponding skies $X = S(x)$\footnote{Actually it can be shown that under certain natural conditions both spaces, the spacetime $M$ and the space of skies $\Sigma$, are diffeomorphic and the conformal structure of the spacetime can be recovered from the structure of the pair $\mathcal{N}$ and $\Sigma$ \cite{Ba14},\cite{Ba15}.}, then the L-boundary is obtained by adding `skies at infinity'.   

As we were indicating above, the way to do that is reminiscent of the well-known blowing up and down techniques in algebraic and symplectic geometry (see for instance \cite[Ch. 7.1]{Mc98}).     Blowing up the origin in $\mathbb{C}^n$ consists in removing the point $\mathbf{0}$ and replacing it by the lines passing through it, that is, we replace $\mathbb{C}^n$ by the manifold $\widetilde{\mathbb{C}}_0^n = \{ (z,l) \in \mathbb{C}^n \times \mathbb{CP}^{n-1} :  z \in l \in \mathbb{CP}^{n-1} \}$. 

As a set, the blown up space $\widetilde{\mathbb{C}}_0^n$ is just $\mathbb{C}^n \backslash \{ \mathbf{0}\} \sqcup \mathbb{CP}^{n-1}$ with $\mathbb{CP}^{n-1}$ being the space of complex tangent lines through the origin.  Notice that $\mathbb{CP}^{n-1} \subset \widetilde{\mathbb{C}}_0^n$ as the subset $\{ (\mathbf{0},l) :  l \in \mathbb{CP}^{n-1}\}$ (called the exceptional divisor).  The projection onto the first factor induces a diffeomorphism between $\widetilde{\mathbb{C}}_0^n \backslash \mathbb{CP}^{n-1}$ and $\mathbb{C}^n \backslash \{ \mathbf{0} \}$ and collapses (blows down)  the exceptional divisor $\mathbb{CP}^{n-1}$ onto the origin $\mathbf{0} \in \mathbf{C}^n$.  The blown up space $\widetilde{\mathbb{C}}_0^n$  can also be visualized as the tautological line bundle $L$ over $\mathbb{CP}^{n-1}$ where the fibre over $l \in \mathbb{CP}^{n-1}$ is the line $l$ itself considered as a one-dimensional space in $\mathbb{C}^n$.   Notice that the zero section of the bundle can be identified with the exceptional divisor $\mathbb{CP}^{n-1}$.   We may recover the original space by identifying back all the lines $l \in \mathbb{CP}^{n-1}$ among themselves or, if you wish, collapsing the base space of the bundle to one point.     

Actually, there is no reason to restrict ourselves to $\mathbb{C}^n$ or to blow up and down just a single point $\mathbf{0}$.  We may consider a time-oriented strongly causal pseudoconvex Lorentzian manifold $M$ whose  space of all future-oriented unparametrized causal geodesics is a smooth Hausdorff manifold \cite{Lo90b}.    For the purposes of this work we will only need the subspace of future oriented maximal null geodesics $\mathcal{N}$ as it is the one that captures the conformal properties of the original spacetime.  Thus, we may define the blown up space $\widetilde{M}$ of the conformal spacetime $(M, \mathcal{C})$ as the smooth manifold:
$$
\widetilde{M} = \{(x ,\gamma) \in M \times \mathcal{N} \colon x \in \gamma \in X = S(x) \subset \mathcal{N}\} \, ,
$$
with $S(x)$ the congruence of null geodesics passing through $x$, that is, the sky at $x$.   It is clear that the natural fibration $\pi_M \colon \widetilde{M} \to M$ defines a fibre bundle structure over $M$ (the charts defined in Sect. \ref{sec:Preliminaries} would provide the desired local trivializations of $\pi_M$) with standard fibre $\mathbb{S}^{m-2}$ ($m = \dim M$).   Hence we would construct the boundary of $M$ determining if there exists a natural boundary $\partial\widetilde{M}$ and an extension  $\tilde{\pi}_M$ of the projection $\pi_M$ to it.  Then `blowing down' the fibres of $\tilde{\pi}_M$ we will obtain the boundary of $M$ we are looking for.     

However the previous description of the blown up space $\widetilde{M}$ is not appropriate because there is not an obvious embedding into a larger space $\widetilde{M} \subset \Omega$ that would help to identify the boundary we are after.    As it will be clear in what follows, it is much more convenient to consider the projection $\pi_{\mathcal{N}} \colon \widetilde{M} \to \mathcal{N}$ induced by the projection onto the second factor of the product manifold $M \times \mathcal{N}$.   The fibre of $\pi_{\mathcal{N}}$ at $\gamma$ is the graph of the geodesic $\gamma$ in $M$, but notice that each point $x \in \gamma$, determines the sky $S(x)$, and $\gamma \in S(x)$ for all $x\in \gamma$, thus as the point $x$ moves along $\gamma$ the sky $S(x)$ changes with $x$. Thus we would like to follow the sky $S(x)$ as $x$ moves to the `end' of $\gamma$.   Now it is clear that instead of considering the sky $S(x)$ itself is more convenient to look at its tangent space at $\gamma$. The tangent space $T_\gamma S(x)$ to the sky $S(x)$ at $\gamma$ will be a $(m-2)$-plane $T_\gamma S(x)$ in $T_\gamma \mathcal{N}$.  Even more, because of (\ref{tangent_sky-1}), $T_\gamma S(x) \subset \mathcal{H}_\gamma$, and the tangent space to the sky $S(x)$ at $\gamma$ lies in the contact hyperplane $\mathcal{H}_\gamma$.     Thus we may consider the space $\widetilde{M}$ as sitting inside the Grasmannian $\mathrm{Gr}^{m-2}(\mathcal{H})$ of $(m-2)$-planes on the contact distribution $\mathcal{H}$ by means of the canonical embbeding:
$\iota \colon \widetilde{M} \to \mathrm{Gr}^{m-2}(\mathcal{H})$, $\iota (x,\gamma) = T_\gamma S(x) \in \mathrm{Gr}^{m-2}(\mathcal{H}_\gamma)$.   For the purposes of this work we will denote the range of the embbeding $\iota$ as $\widetilde{\mathcal{N}}$, thus 
$$
\widetilde{\mathcal{N}} = \{ T_\gamma S(x) \in \mathrm{Gr}^{m-2}(\mathcal{H}_\gamma) \colon x \in \gamma \in \mathcal{\mathcal{N}} \} \, ,
$$ 
and we will call it in what follows the blow up of the spacetime $M$.   Then the projection $\pi_\mathcal{N} \colon \widetilde{M} \to \mathcal{N}$ becomes the restriction of the canonical projection $\pi \colon \mathrm{Gr}^{m-2}(\mathcal{H}) \to \mathcal{N}$ (that maps every $(m-2)$-plane $W_\gamma \mapsto \gamma$) that will be denoted again, with a slight abuse of notation, by $\pi_{\mathcal{N}}$.  The fibre at $\gamma$ will be denoted for short by $\tilde{\gamma}$, that is $\pi_\mathcal{N}^{-1}(\gamma) = \tilde{\gamma} = \{T_\gamma S(x) \mid x \in \gamma \}$ .

Choosing a parametrization $\gamma(t)$ of the geodesic $\gamma$, there is a natural way of looking for the boundary of $\widetilde{\mathcal{N}}$ in $\mathrm{Gr}^{m-2}(\mathcal{H})$ by looking at the trace left by the skies at infinity.  That is, if $\gamma:\left(a,b\right)\rightarrow M$ is a future-oriented inextendible null geodesic, we can define the curve $\widetilde{\gamma}:\left(a,b\right)\rightarrow \mathrm{Gr}^{m-2}\left(\mathcal{H}_{\gamma}\right)$ defined by $\widetilde{\gamma}\left(t\right)= T_{\gamma}S\left(\gamma\left(t\right) \right) $.  If the limit points 
\begin{equation}\label{boundary-field}
\begin{tabular}{l}
$\ominus_{\gamma} = \lim_{s\mapsto a^{+}}\widetilde{\gamma}\left(s\right)\in \mathrm{Gr}^{m-2}\left(\mathcal{H}_{\gamma}\right)$ , \vspace{3mm} \\
$\oplus_{\gamma} = \lim_{s\mapsto b^{-}}\widetilde{\gamma}\left(s\right)\in \mathrm{Gr}^{m-2}\left(\mathcal{H}_{\gamma}\right)$ ,
\end{tabular}
\end{equation}
exist, they will represent the tangent space to the `sky at infinity', hence we define the boundary $\partial^+\widetilde{\mathcal{N}} = \{ \oplus_\gamma \mid \gamma \in \mathcal{N}\}$ and $\partial^-\widetilde{\mathcal{N}}$  is defined similarly (see Fig. \ref{lboundary} for a pictorical representation of such construction).    

Notice that given a point $x\in M$, its sky $X = S(x)$ determines a submanifold $\widetilde{X}= \{T_\gamma X \mid \gamma \in X \} \subset \widetilde{\mathcal{N}}$, which corresponds to the fibres of the canonical projection map $\pi_M \colon \widetilde{\mathcal{N}} \to M$ given by $\pi_M (T_\gamma S(x)) = x$  (see diagram \ref{infinitesimal_sky} below).   The fibres $\widetilde{X}$ of $\pi_M$ define a canonical distribution $\mathcal{D}^\sim$ on $\widetilde{\mathcal{N}}$.   It will be the main contribution of this paper (see Thm. \ref{main_thm}) to show that the distribution  $\mathcal{D}^\sim$ can be extended smoothly to the boundary of the blown up space $\widetilde{\mathcal{N}}$.  Hence, blowing down the integral leaves of total distribution on the closure $\overline{\widetilde{\mathcal{N}}}$, that is, considering the quotient space defined by its leaves, will provide the extension $\overline{M}$ and the L-boundary we are looking for.

\begin{figure}[h]
\centering
\begin{tikzpicture} [every node/.style={midway}]
\matrix[column sep={7em,between origins},
        row sep={3em}] at (0,0)
{ & \node(Gr) {$\mathrm{Gr}^{m-2}(\mathcal{H})$}; &  \\
\node(PN1)   {$\bigcup_{\gamma \in \mathcal{N}} \tilde{\gamma} = \widetilde{\mathcal{N}}^{2m-2}$}; & \node(cong) {$\cong$}; & \node(N) {$\widetilde{\Sigma}^{2m-2} =  \bigcup_{x \in M} \tilde{X}$}; \\
\node(PN2) {$\mathcal{N}^{2m-3}$};  &  & \node(PNC) {$M^m \cong \Sigma^m$};         \\};
\draw[<-] (PN2)   -- (PN1) node[anchor=east] {$\pi_\mathcal{N}$};
\draw[->] (N) -- (PNC) node[anchor=west]  {$\pi_M$};
\draw[<-left hook] (Gr) -- (N);  \draw[<-right hook] (Gr) -- (PN1);
\end{tikzpicture}
\caption{ Diagram summarizing the spaces related to the blow up space $\widetilde{\mathcal{N}}$.  The dimensions of the various spaces are indicated as superindices.  The space of skies $\Sigma$ is identified with $M$ and $\widetilde{\mathcal{N}}$ sits inside the contact Grassmannian $\mathrm{Gr}^{m-2}(\mathcal{H})$.}
\label{infinitesimal_sky}
\end{figure}


\subsection{L-spacetimes}

Notice that in the particular simple instance of 3-dimensional spacetimes, skies are circles and their tangent spaces are lines in the bidimensional contact plane $\mathcal{H}_\gamma$, thus the Grasmannian $\mathrm{Gr}^{m-2}(\mathcal{H}_\gamma)$ becomes $\mathbb{P}(\mathcal{H}_\gamma)$, the space of lines in the contact plane $\mathcal{H}_\gamma$ which is diffeomorphic to $\mathbb{S}^1$.  We will take advantage of this and present the construction of the L-boundary on 3-dimensional spacetimes. This choice will provide significant technical simplifications making the main ideas involved on such construction more transparent.

Thus, if the null geodesics $\gamma$ define a nice boundary $\partial^+\widetilde{\mathcal{N}}$ at $+\infty$ in $\mathbb{P}(\mathcal{H})$, and the tangent spaces to the skies at infinity $\oplus_\gamma = \lim_{s\to \infty}  T_\gamma S(\gamma(s))$ define an integrable distribution, its integral leaves will define the actual skies at infinity, hence the corresponding quotient space will define a boundary for $M$, its points being the `events' defined by the skies at infinity (see in Fig. \ref{lboundary} a pictorical representation of the correspondence between the skies at infinity and the points of the L-boundary).   

\begin{figure}[h]
  \centering
    \includegraphics[scale=0.5]{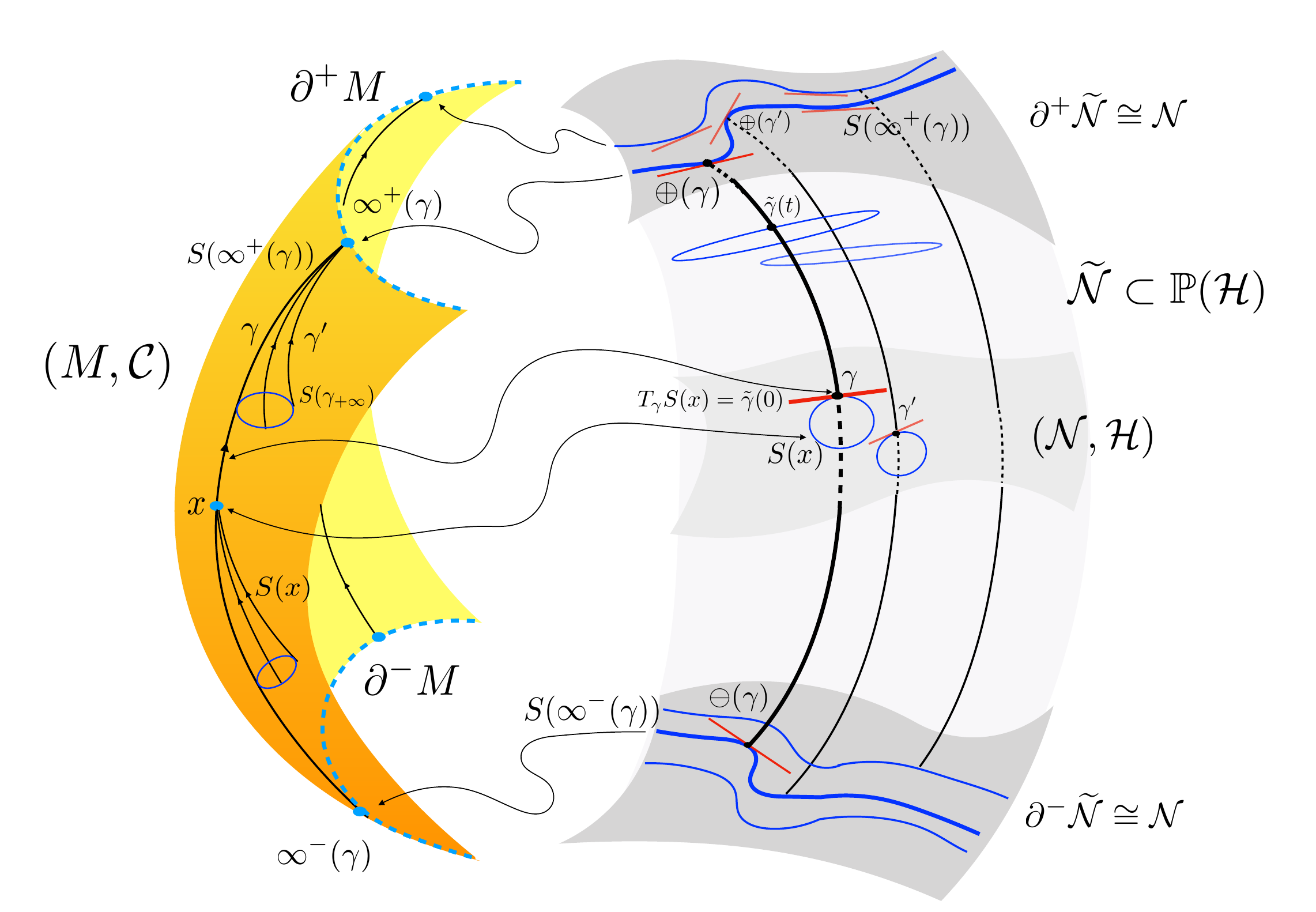}
  \caption{The construction of the L-boundary.  On the left, the spacetime of events $M$; on the right, the space of light rays $\mathcal{N}$.  A light ray $\gamma \subset M$ on the left corresponds to a point $\gamma \in \mathcal{N}$ on the right, and a sky $S(x)$ (a congruence of geodesics) on the left, corresponds to a circle on the right.   The `skies at infinity', $S(\infty^+(\gamma))$, are integral curves of the distributions $\oplus$ and $\ominus$, represented by blue lines on the boundaries $\partial^\pm \widetilde{\mathcal{N}} \cong \mathcal{N}$, or congruences of geodesics that approach a point $\infty^+(\gamma)$ at infinity.  The `skies at infinity' are glued to the spacetime $M$ forming its L-boundary.}
  \label{lboundary}
\end{figure}

Of course, in general, it is not true that the tangent spaces to the skies at infinity define a nice boundary, or that the integral curves of the distributions $\oplus$, $\ominus$ define a smooth manifold, however, as we will see in this article, in 3-dimensional spacetimes, under reasonable technical conditions, the new conformal boundary defined in this way exists and exhibits a number of interesting features: it is a smooth boundary and sets the ground for an extension of the conformal structure of the original spacetime.     At the same time we must point out that the restriction to 3-dimensional spacetimes is not a fundamental one.   As it was commented before, in the 3-dimensional case, the geometry of the Grassmannian $\mathrm{Gr}^{m-2}(\mathcal{H})$ is particularly nice and allows to simplify the technical conditions needed for the various constructions (and helps to visualise better the objects entering the full picture), thus restricting to such case is just for the purposes of brevity and clarity in the exposition, leaving a general discussion for upcoming works (see Sect. \ref{sec:discussion}).

\bigskip

From now on we will assume that the conformal Lorentz manifold $\left(M,\mathcal{C}\right)$ verifies:
\begin{enumerate}\label{conditions}
\item $\dim M = 3$. 
\item $\left(M,\mathcal{C}\right)$ is time oriented, strongly causal, null--pseudo convex, light non--conjugate and sky-separating.
\item The distributions $\oplus, \ominus : \mathcal{N} \rightarrow \mathbb{P}\left(\mathcal{H}\right)$ defined by $\oplus_{\gamma}=\lim_{s\mapsto b^{-}}T_{\gamma}S\left(\gamma\left(s\right)\right)$ and $\ominus_{\gamma}=\lim_{s\mapsto a^{+}}T_{\gamma}S\left(\gamma\left(s\right)\right)$ are differentiable and regular and such that $\oplus_{\gamma}\neq\ominus_{\gamma}$ for any maximally and future--directed parametrized light ray $\gamma:\left(a,b\right)\rightarrow M$.
\end{enumerate}

As it will be shown along the paper, see Thm. \ref{main_thm}, the set of conditions above will be sufficient for a spacetime to possess a well defined $L$-boundary (perhaps not smooth), then we may introduce the following definition. 

\begin{definition}\label{Lspace} A smooth manifold $M$ satisfying conditions 2 and 3 above will be said to be an L--spacetime.
\end{definition}

Notice that, condition 2 above summarises general conditions under which the topological and differentiable structures of the spaces $M$ and $\mathcal{N}$ are good enough (see \cite{Ba14}, \cite{Ba15} for a detailed discussion of them). In particular if we consider $\dim M = 3$, then by \cite[Lem.~2.5]{Ba17} $M$ is light non--conjugate if and only if for any $X\neq Y \in \Sigma$ with $\gamma\in X\cap Y$ then $T_{\gamma}X \neq T_{\gamma}Y$.
Notice that any null geodesic has no conjugate points contained in any given normal neighbourhood. This implies that the curve $\widetilde{\gamma}\left(s\right)=T_{\gamma}S\left(\gamma\left(s\right)\right)\in \mathbb{P}\left(\mathcal{H}_{\gamma}\right)\simeq \mathbb{S}^{1}$ is locally injective. Moreover, if we assume $M$ to be light non--conjugate, then $\widetilde{\gamma}$ must be injective and therefore the continuity of $\widetilde{\gamma}$ would imply that the limits $\oplus_{\gamma}$ and $\ominus_{\gamma}$ do exist in $\mathbb{P}\left(\mathcal{H}_{\gamma}\right)$, that is they define $1$--dimensional subspaces of $\mathcal{H}_{\gamma}$.
 
Condition 3 refers only  to the minimum regularity conditions that the distributions $\oplus$, $\ominus$ must satisfy in order to guarantee that the construction of the L-boundary makes sense.  However, see the discussion after Corollary \ref{Corollary-ext}, they will not guarantee that the constructed boundary is smooth. 


\subsection{A simple example: Minkowski spacetime}\label{sec:block}

We will illustrate the previous ideas using a simple family of $L$-spacetimes.  We will consider a globally hyperbolic block embedded in 3-dimensional Minkowski spacetime $\mathbb{M}^3$, identified with $\mathbb{R}^3 = \{(t,x,y) \mid t,x,y \in \mathbb{R} \}$ with the standard Lorentz metric $\mathbf{g}_3 = - dt^2 + dx^2 + dy^2$ and the corresponding conformal structure.  Thus let $a,b$ be two real numbers such that $a < 0$ and $1 < b$,  we will  consider the spacetime $\mathbb{M}^3(a,b) = \{ (t,x,y) \in \mathbb{R}^3 \mid a < t < b \}$ equipped with the restriction of the Lorentz metric $\mathbf{g}_3$ (see Fig. \ref{M3ab} (a)).   Since $\mathbb{M}^3(a,b) \subset \mathbb{M}^3$ is open and they share the Cauchy surface $C = \{ t = 0\} \cong \mathbb{R}^2$.  Their corresponding space of light rays can be identified with $C \times \mathbb{S}^1$, actually the local maps $\mathbb{PN}(C) \to \mathcal{N}_V$ defined in Sect. \ref{sec:Preliminaries} are now globally defined and  $\mathbb{PN}(C) \cong C \times \mathbb{S}^1\cong \mathcal{N}$.    Using the previous identification, we will introduce explicit coordinates $(x,y,\theta)$ on $\mathcal{N}$ as follows.  Given $(x,y) \in C$, $-\pi < \theta \leq \pi$,  we will denote by $\gamma_{(x,y,\theta)}\in \mathcal{N}$ the light ray determined by the null geodesic:
\begin{equation}\label{gammaxyphi}
\gamma_{(x,y,\theta)}(t) = (t, x+t\cos \theta, y + t\sin \theta ) \, , \quad a < t < b \, .
\end{equation}
Notice that $\gamma_{(x,y,\theta)}$ cuts the Cauchy hypersurface $C$ at $(x,y)$ and its projection onto $C$ determines the angle $\theta$ with respect to the $x$-axis.

Thus, given $(x_0,y_0) \in C$, $-\pi < \theta_0 \leq \pi$ and $t_0 \in (a,b)$, we denote by $p_0 \in \mathbb{M}^3(a,b)$, the event:
$$
p_0 = \gamma_{(x_0,y_0,\theta_0)}(t_0) =  (t_0, x_0+t_0\cos \theta_0, y_0 + t_0\sin \theta_0 ) \, . 
$$

Denoting by $\gamma$ the geodesic $\gamma_{(x_0,y_0,\theta_0)}$ we will denote by $X(\gamma, t, s)$, with $a< t < b$, $-\pi < s \leq \pi$,  the family of null geodesics passing through $\gamma(t)$, that is the sky $S(\gamma(t))$ at $\gamma(t)$.  In particular $X(\gamma, t_0, s)(\tau) = p_0 + \tau(1, \cos s , \sin s)$.   The intersection of the geodesic $X(\gamma, t_0, s)$ with $C$ happens when $\tau = -t_0$, and then, 
$$
X(\gamma, t_0, s) (-t_0) = (0,x_0 + t_0( \cos \theta_0 - \cos s) ,y_0 + t_0(\sin \theta_0 -\sin s)) \in C \, .
$$
Then, using back the notation for geodesics in Eq. (\ref{gammaxyphi}), we get for the light rays in $S(p_0)$:
$$
X(\gamma, t_0, s) = \gamma_{(x_0 + t_0( \cos \theta_0 - \cos s) ,y_0 + t_0(\sin \theta_0 -\sin s), s)} = \gamma_{(x(s), y(s), s)}\, .
$$
Hence the sky at $p_0$ is described as the circle centered at $p_0$ of radius $t_0$,  $(x_0 + t_0( \cos \theta_0 - \cos s) ,y_0 + t_0(\sin \theta_0 -\sin s), s)$, $-\pi < s \leq \pi$ in $C\times \mathbb{S}^1$.

\begin{figure}[h]
\centering
    \includegraphics[scale=0.6]{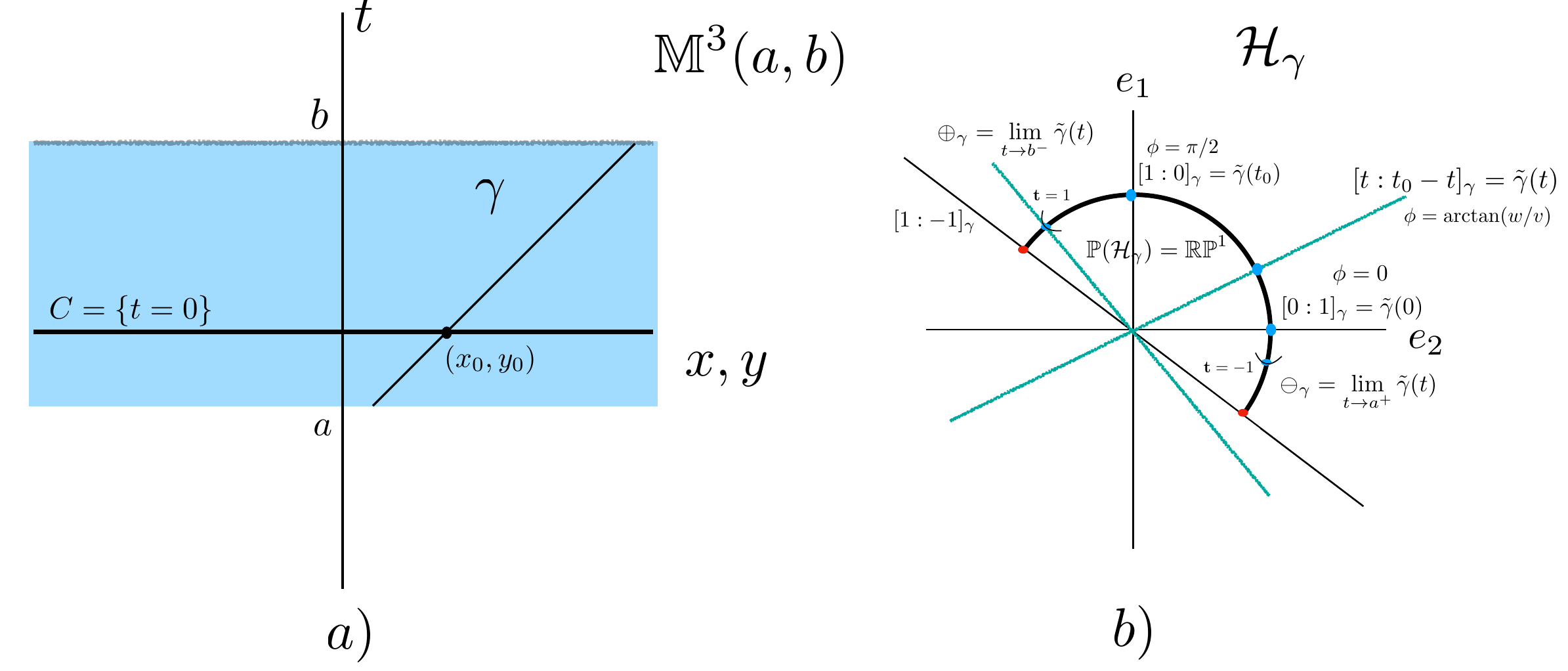}
\caption{ a) A bidimensional projection of the globally hyperbolic block $\mathbb{M}^3(a,b)$. The Cauchy hypersurface $C = \{ t = 0\}$ a light ray $\gamma$ are drawn.  The topological boundary $\{t = b\}$ is shaded.  b) The contact plane $\mathcal{H}_\gamma$ at the light ray $\gamma$.   Lines described by homogeneous coordinates $[w:v]_\gamma$ are drawn in blue.   The projective circle $\mathbb{P}(\mathcal{H}_\gamma)$ is highligthed (the red dots should be identified).  The distributions $\oplus_\gamma$, $\ominus_\gamma$ are marked with blue dots and the polar coordinate $\phi$ is indicated.} 
\label{M3ab}
\end{figure}

The space of light rays $\mathcal{N}$ is a smooth Hausdorff 3-manifold (clearly $\mathbb{M}^3$, hence $\mathbb{M}^3(a,b)$, is a L-spacetime) and its tangent bundle $T\mathcal{N}$ is 6-dimensional.   The global identification $\mathcal{N} \cong C \times \mathbb{S}^1$, provides distinguished basis $\partial/\partial x \mid_\gamma$, $\partial/\partial y \mid_\gamma$, and $\partial/\partial \theta\mid_\gamma$, for $T_\gamma \mathcal{N}$, that is:
$$
T_\gamma \mathcal{N} = \mathrm{span} \left\{ \left.\frac{\partial}{\partial x}\right|_\gamma, \left.\frac{\partial}{\partial y}\right|_\gamma, \left.\frac{\partial}{\partial \theta}\right|_\gamma \right\} \, .
$$
Thus the tangent space $T_\gamma S(p_0)$ at $\gamma$ of the sky at $p_0$  is  obtained by taking the derivative at $s= \theta_0$ of the curve $(x(s), y(s), s) = (x_0 + t_0( \cos \theta_0 - \cos s), y_0 + t_0(\sin \theta_0 -\sin s), s)$, thus,
\begin{equation}\label{tangent_sky}
T_\gamma S(p_0) = \mathrm{span} \left\{  t_0 \sin \theta_0\left. \frac{\partial}{\partial x}\right|_\gamma - t_0 \cos \theta_0 \left.\frac{\partial}{\partial y}\right|_\gamma +  \left.\frac{\partial}{\partial \theta}\right|_\gamma \right\}   \, .
\end{equation}
If instead of $p_0$ we choose another event in the geodesic $\gamma = \gamma_{(x_0,y_0,\theta_0)}$, for instance the intersection $\{ (0, x_0,y_0) \}= \gamma \cap C$ with the Cauchy hypersurface $C$, then the sky passing through it has the form $\gamma_{(x_0 ,y_0, s)}$ and its tangent space:
$$
T_\gamma S(0, x_0,y_0) = \mathrm{span} \left\{  \left.\frac{\partial}{\partial \theta}\right|_\gamma \right\} \, .
$$
Then, because of (\ref{contact-struct}), we obtain the 2-dimensional contact hyperplane $\mathcal{H}_\gamma \subset T_\gamma \mathcal{N}$ at $\gamma$, 
\begin{equation}\label{Hskies}
\mathcal{H}_\gamma =  T_\gamma S(p_0)  \oplus T_\gamma S(0, x_0,y_0)  = \mathrm{span} \left\{ t_0 \sin \theta_0\left. \frac{\partial}{\partial x}\right|_\gamma - t_0 \cos \theta_0 \left.\frac{\partial}{\partial y}\right|_\gamma +  \left.\frac{\partial}{\partial \theta}\right|_\gamma, \left.\frac{\partial}{\partial \theta}\right|_\gamma \right\} \, .
\end{equation}
Notice that the 1-form $\alpha = \cos \theta \,  dx + \sin \theta \, dy$ on $\mathcal{N}$ defines a contact 1-form with $\ker \alpha_\gamma = \mathcal{H}_\gamma$.  Then $d\alpha = (\sin \theta\,  dx - \cos \theta \, dy) \wedge d\theta$ restricted to $\mathcal{H}_\gamma$ is nondegenerate.  In fact a simple calculation shows that $\ker d\alpha$ is spanned by the vector field $R = \cos \theta\,  \partial / \partial x + \sin \theta \, \partial /\partial y$, the Reeb field of the contact 1-form $\alpha$, which is transversal to $\mathcal{H}$.

We are ready now to describe the blow up space $\widetilde{\mathcal{N}}$ of $M$.  For that we will observe first that the bundle of lines on the contact planes over the space of light rays, $\mathbb{P}(\mathcal{H}) \to \mathcal{N}$, is a bundle of circles over $\mathcal{N}$, that is, the fibre over $\gamma \in \mathcal{N}$ is the real projective space $\mathbb{P}(\mathcal{H}_\gamma)$, i.e., the space of lines passing through $\mathbf{0}$ in $\mathcal{H}_\gamma$ (see Fig. \ref{M3ab} (b)).    Such space is easily described by using homogeneous coordinates.  Given $w, v$ two real numbers, we define the line $[w:v]_\gamma \subset \mathcal{H}_\gamma$, as:
$$
[w:v]_\gamma = \mathrm{span} \{ w e_1 + v e_2 \} \, ,
$$
with $e_1, e_2$ a given linear basis of $\mathcal{H}_\gamma$.   Using the natural basis provided by the description of $\mathcal{H}_\gamma$ given in (\ref{Hskies}), that is $e_1 =  t_0 \sin \theta_0\left. \frac{\partial}{\partial x}\right|_\gamma - t_0 \cos \theta_0 \left.\frac{\partial}{\partial y}\right|_\gamma+  \left.\frac{\partial}{\partial \theta}\right|_\gamma$ and  $e_2 = \left.\frac{\partial}{\partial \theta}\right|_\gamma$, we get:
$$
[w:v]_\gamma = \mathrm{span} \left\{ w t_0 \sin \theta_0\left. \frac{\partial}{\partial x}\right|_\gamma - w t_0\cos \theta_0 \left.\frac{\partial}{\partial y}\right|_\gamma + (w + v) \left.\frac{\partial}{\partial \theta}\right|_\gamma \right\} \subset \mathcal{H}_\gamma \, .
$$
Hence, because of (\ref{tangent_sky}) the line defined by the tangent space to the sky $S(p_0)$ at $\gamma$ 
has homogeneous coordinates  $w = 1$ and $v = 0$, that is $T_\gamma S(p_0) = [1:0]_\gamma \in \mathbb{P}(\mathcal{H}_\gamma)$.  Similarly $T_\gamma S(0,x_0,y_0) = [0:1]_\gamma$ (see Fig. \ref{M3ab}).
We may also consider the polar coordinate $\phi = \arctan (w/v)$, then $\phi (T_\gamma S(p_0)) = \pi/2$ and $\phi (T_\gamma S(0,x_0,y_0) ) = 0$.  Notice that $(x,y,\theta, \phi)$ provide a coordinate chart for $\mathbb{P}(\mathcal{H})$.

Thus given a geodesic $\gamma = \gamma_{(x_0,y_0,\theta_0)}$, the curve $\tilde{\gamma}(t)$ on $\mathbb{P}(\mathcal{H}_\gamma)$ defined by the tangent spaces to the skies at $\gamma(t)$, is given by:
\begin{equation}\label{gammatildet}
 \tilde{\gamma}(t) = T_\gamma S(\gamma(t)) = [t:t_0-t]_\gamma \in \mathbb{P}(\mathcal{H}_\gamma) \, ,
\end{equation}
i.e., $\phi (\tilde{\gamma}(t)) = \arctan (t/t_0 -t)$.
Thus the blow up space $\widetilde{\mathcal{N}}$ consists of all lines $\tilde{\gamma}(t)$, $t\in (a,b)$, in $\mathbb{P}(\mathcal{H}_\gamma)$.  Thus we may write:
$$
\widetilde{\mathcal{N}} = \{ (\gamma, [t:t_0-t]_\gamma) :  \gamma\in \mathcal{N}, a< t < b \} \subset \mathbb{P}(\mathcal{H}) \, . 
$$ 
To calculate the distribution $\oplus$ (similarly for $\ominus$) we consider the limit $t \to b^-$ of (\ref{gammatildet}), that is:
$$
\oplus_\gamma = \lim_{t \to b^-}  \tilde{\gamma}(t) = [b:t_0-b]_\gamma \in \mathbb{P}(\mathcal{H}_\gamma)\, ,
$$
that is $\phi(\oplus_\gamma) = \arctan(b/t_0-b)$.
Thus the line $\oplus_\gamma$ is the line spanned by the tangent vector $b\sin \theta_0\left. \frac{\partial}{\partial x}\right|_\gamma - b \cos \theta_0 \left.\frac{\partial}{\partial y}\right|_\gamma + \left.\frac{\partial}{\partial \theta}\right|_\gamma$ and the boundary $\partial^+\widetilde{\mathcal{N}}$ is the graph of the map $\mathcal{N} \to \mathbb{P}(\mathcal{H})$, given by $\gamma \mapsto (\gamma, [b:t_0-b]_\gamma)$ or, in local coordinates $(x,y,\theta) \mapsto (x,y,\theta, \phi =  \arctan(b/t_0-b))$ ($t_0$ if fixed for all $\gamma$).

We obtain the sky at infinity of $\gamma$ as the orbit in $\mathcal{N}$ of the distribution $\oplus$ passing through $\gamma$, that is, we look for curves $c(s) = (x(s), y(s), \theta(s))$ such that:
\begin{equation}\label{IVPab}
\frac{dx}{ds}  = b\sin \theta \, , \quad \frac{dy}{ds}  = -b\cos \theta \, , \quad \frac{d\theta}{ds}  = 1  \, ,
\end{equation}
with initial value $c(0) = (x_0,y_0, \theta_0)$, that is,
\[
c(s) = (x_0 + b\cos\theta_0- b\cos(s + \theta_0), y_0 + b\sin\theta_0- b\sin(s + \theta_0), s + \theta_0)  .
\]

Then the light ray corresponding to $c(s)$ is given by (see Eq. (\ref{gammaxyphi})):
\begin{align*}
\gamma_s(t) &= (t, x(s) + t \cos \theta(s), y(s) + t\sin \theta(s)) = \\
&=(t, x_0 + b\cos\theta_0 + (t-b) \cos(s + \theta_0),  y_0 + b\sin\theta_0 +(t-b) \sin(s + \theta_0)) \, .
\end{align*}

Thus the orbit $X_\gamma$ passing through $\gamma$ consists of the family of light rays $\gamma_s$ and all of them satisfy that $\lim_{t\to b^-} \gamma_s(t) = (b,x_0+ b\cos\theta_0,y_0+ b\sin\theta_0)$ for all $s$, that is, they are exactly the sky (in $\mathbb{M}^3$) of the event $(b,x_0+ b\cos\theta_0,y_0+ b\sin\theta_0)$ at the boundary of $\mathbb{M}^3(a,b)$.   In other words, the future L-boundary of $\mathbb{M}^3(a,b)$ which is defined as the quotient of $\partial^+\widetilde{\mathcal{N}} \cong \mathcal{N}$ with respect to the integral curves of the distribution $\oplus$, is bidimensional and its leaves can be naturally identified with the points $(b,x_0,y_0)$.  Notice that the leaves of the distribution $\oplus$ (respec. $\ominus$) are compact (actually diffeomorphic to $\mathbb{S}^1$) and the distribution defined on $\overline{\widetilde{\mathcal{N}}}$ is regular.


\section{Coordinate charts in $\mathbb{P}\left(\mathcal{H}\right)$}\label{sec:charts}

We will construct coordinate charts in $\mathbb{P}\left(\mathcal{H}\right)$ from the natural atlas in $T\mathcal{N}$ (see \cite{Ba14}).   We will succinctly review the construction of such atlas adapted to the present situation.

Fix $\mathbf{g}\in \mathcal{C}$ an auxiliary metric. 
Since $M$ is assumed to be strongly causal, by \cite[Rem.~2.15]{Mi08}, there is a topological basis of globally hyperbolic, causally convex, normal  open neighbourhoods for any $q\in M$. 
Choose $U\subset M$ one of such neighbourhoods.
Since $U$ is globally hyperbolic, by \cite[Thm.~3.78]{Mi08}, there exists a smooth spacelike Cauchy surface $C\subset U$. Then, any light ray passing through $U$ intersects $C$ in a singleton $\{ q \} = \gamma\cap C$.

Let $\left\{E_{1},E_{2} ,E_{3}\right\} $ be an orthonormal frame in $U$ such that $E_{1} $ is future--oriented timelike and $\left\{E_{2},E_{3}\right\}$ are spacelike and such that $E_{2}\left(p\right)$ and $E_{3}\left(p\right)$ are tangent to $C$ for all $p\in C\subset U$. 
Consider $c\equiv\left(x,y\right)$ a coordinate system for $C$. 
Given $\theta \in \left(-\pi,\pi\right]$, the light ray passing through $p\in C$ can be defined by the null geodesic $\gamma$ such that $\gamma\left(0\right) = p$ and $\gamma'\left(0\right) = E_1\left(p\right)+\cos \theta E_2\left(p\right) + \sin \theta E_3\left(p\right)$. 
Then, denoting by $\mathcal{N}_{U}\subset \mathcal{N}$ the open set of all light rays intersecting $U$, that is 
\[
\mathcal{N}_{U}=\{ \gamma\in \mathcal{N}: \gamma\cap U\neq \varnothing \}
\]
we can define local coordinates in $\mathcal{N}_{U}$ by
\begin{equation}\label{carta-N}
\psi:\mathcal{N}_{U}\rightarrow \mathbb{R}^{3};\hspace{1cm}\psi=\left(x,y,\theta\right)
\end{equation}
Notice that if $M$ is globally hyperbolic, as in the example of the Minkowski block $\mathbb{M}^3(a,b)$ discussed in Sect. \ref{sec:block}, there exists a globally defined Cauchy surface $C$ and the space of light rays $\mathcal{N}$ can be identified with $C \times \mathbb{S}^1$.  Then any local chart $(x,y)$ in $C$ will define a local chart $(x,y,\theta)$ of $\mathcal{N}$.

Recall (see Eq. (\ref{Jmodgamma})) that a tangent vector $\langle J\rangle\in T_{\gamma}\mathcal{N}_{U}$ can be identified with a class of Jacobi fields along $\gamma$ modulo $\gamma'$. 
This Jacobi field can be uniquely determined by its initial vectors $J(0)$ and $J'(0)$ at $p = \gamma\left(0\right)$. Taking the equivalence $\mathrm{mod}\gamma'$ into account we may choose representatives of these initial vectors in the subspace $T_p C$, that is we may choose the initial vectors:
\begin{equation}\label{initial-vectors-J}
\left\{
\begin{tabular}{l}
$J\left(0\right)=w^2 E_2\left(p\right) +w^3 E_3\left(p\right)$ \\
$J'\left(0\right)=v^2 E_2\left(p\right) +v^3 E_3\left(p\right)$
\end{tabular}
\right.
\end{equation} 
and since $\mathbf{g}\left(\gamma', J'\right)=0$, then $v^2 \cos \theta + v^3 \sin \theta=0$ and therefore $v^2, v^3$ are determined one from another.
Without any lack of generality we consider $\cos\theta \neq 0$ then we choose $v = v^3$, $w^2$ and $w^3$ as local coordinates in $T\mathcal{N}_{U}$.
So, a chart in $T\mathcal{N}_{U}$ can be defined by
\[
\overline{\psi}:T\mathcal{N}_{U}\rightarrow \mathbb{R}^{6};\hspace{1cm}\overline{\psi}=\left(x,y,\theta; w^2,w^3,v\right)  .
\]

Now, we will define coordinates in $\mathcal{H}_{U}=\mathcal{H}\cap T\mathcal{N}_{U} = \bigcup_{\gamma\in \mathcal{N}_{U}}\mathcal{H}_{\gamma}$ from the chart $\overline{\psi}$. 
Any $\langle J\rangle\in \mathcal{H}_{\gamma}$ verifies $\mathbf{g}\left(\gamma', J\right)=0$ and therefore $w^2 \cos \theta + w^3 \sin \theta=0$. Then, if $\cos \theta \neq 0$, we have $w^2=-\tan \theta \cdot w^3$ and $w=w^3$ can be considered as a coordinate for $\mathcal{H}_{U}$, then 
\begin{equation}\label{carta-H}
\varphi:\mathcal{H}_{U}\rightarrow \mathbb{R}^{5} \, ;\qquad \varphi=\left(x,y,\theta,w,v\right)
\end{equation}
is a coordinate chart. 
Trivially, the initial vectors $J\left(0\right)$ and $J'\left(0\right)$ are related by 
\begin{equation}\label{relation-J-Jprima}
v\left(\langle J\rangle\right)J\left(0\right) = w\left(\langle J\rangle\right)J'\left(0\right)   .
\end{equation}
Moreover, observe that both $J\left(0\right)$ and $J'\left(0\right)$ have been chosen in the vector subspace 
\begin{equation}\label{subspace-init-vect}
T_p C \cap \left\{\gamma'\left(0\right)\right\}^{\perp}\subset T_p M
\end{equation}
where $\left\{\gamma'\left(0\right)\right\}^{\perp} = \left\{u\in T_p M: \mathbf{g}\left( \gamma'\left(0\right), u \right) = 0 \right\}$ is the subspace orthogonal to $\gamma'\left(0\right)$.

It is also easy to build coordinates in $\mathbb{P}\left(\mathcal{H}\right)$ if we adapt them from the chart $\varphi$. 
If we consider $\langle J\rangle\in \mathcal{H}_{\gamma}$ and $\langle \overline{J}\rangle=\langle \lambda J\rangle$ for some $\lambda\in \mathbb{R}$, then we trivially have
\begin{equation}\label{coordinates-H-PH}
\left\{ 
\begin{array}{l}
w\left(\langle \overline{J}\rangle\right) = \lambda w\left(\langle J\rangle\right)  \\
v\left(\langle \overline{J}\rangle\right) = \lambda v\left(\langle J\rangle\right) 
\end{array}
\right.
\end{equation}
then the homogeneous coordinate $\phi=\left[w:v\right]$ (or equivalently the polar coordinate $\phi=\arctan (w/v)$, see Fig. \ref{M3ab}(b) again) verifies 
\[
\phi\left(\langle \overline{J}\rangle\right)=\left[w\left(\langle \overline{J}\rangle\right):v\left(\langle \overline{J}\rangle\right)\right]=
\left[w\left(\langle J\rangle\right):v\left(\langle J\rangle\right)\right]=\phi\left(\langle J\rangle\right)
\]
and it determines the line $\left[J\right]=\mathrm{span}\left\{\langle J\rangle\right\}\in \mathbb{P}\left(\mathcal{H}_{\gamma}\right)$.
Therefore, a coordinate chart in $\mathbb{P}\left(\mathcal{H}\right)$ can be obtained by 
\begin{equation}\label{carta-PH}
\widetilde{\varphi}:\mathbb{P}\left(\mathcal{H}_{U}\right)\rightarrow \mathbb{R}^{4};\hspace{1cm}\widetilde{\varphi}=\left(x,y,\theta,\phi\right)  .
\end{equation}


\section{A projective parameter for light rays}\label{sec:proj-param}

Under the stated hypotheses, in this section we will show the existence of a maximal parameter $\mathbf{t}\in\left(-1,1\right)$ such that the map defined by 
\[
\begin{tabular}{rcl}
$\mathcal{N}_{U}\times \left(-1,1\right)$ & $\rightarrow$ & $M$ \\
$\left(\gamma,\mathbf{t}\right)$ & $\mapsto$ & $\gamma\left(\mathbf{t}\right)$
\end{tabular}
\]
is differentiable. 
The properties of this particular parameter will permit us to extend the conformal manifold $M$ in such a way that the extension of any light ray will be transversal to the boundary $\partial M$.  

We need some previous Lemmas.

\begin{lemma}\label{Lemma-diffeo-sigma}
Let $\pi^{\mathbb{PN}}_{M}:\mathbb{PN}\rightarrow M$ be the canonical projection. 
Then the map 
\begin{equation}\label{difeo-sigma}
\begin{tabular}{rrcl}
$\sigma:$ & $\mathbb{PN}$ & $\rightarrow$ & $\mathbb{P}\left(\mathcal{H}\right)$ \\
& $\left[u\right]$ & $\mapsto$ & $T_{\gamma_{\left[u\right]}}S\left(\pi^{\mathbb{PN}}_{M}\left(\left[u\right]\right)\right)$
\end{tabular}
\end{equation}
is differentiable.
\end{lemma}

\begin{proof}
Since $M$ is light non--conjugate, by \cite[Lem.~2.5]{Ba17}, then $\sigma$ is injective.

Now, let us show that $\sigma$ is differentiable. 
Fix some auxiliary metric $\mathbf{g}\in\mathcal{C}$ and consider a globally hyperbolic open set $U\subset M$ with a Cauchy surface $C\subset U$ . 
We can assume the existence of an orthogonal frame $\{ E_1, E_2, E_3 \}$ such that $E_2, E_3$ are spacelike and $E_1$ is timelike with respect to the metric $\mathbf{g}$. 
Also assume that $E_2\left(c\right),E_3\left(c\right)\in T_c C$ for all $c\in C$.

Let us define, as in (\ref{PNW}):
\[
\begin{tabular}{c}
$\mathbb{PN}\left(U\right)=\{ \left[u\right]\in \mathbb{PN}: \pi^{\mathbb{PN}}_{M}\left(\left[u\right]\right)\in U \}$ \\
$\mathbb{N}^{+}\left(U\right)=\{ u\in \mathbb{N}^{+}: \pi^{\mathbb{N}}_{M}\left(u\right)\in U \}$ 
\end{tabular}
\]
where $\pi^{\mathbb{PN}}_{M}$ and $\pi^{\mathbb{N}}_{M}$ are the corresponding canonical projection on $M$.
Clearly, we can identify diffeomorphically $\mathbb{PN}\left(U\right)$ with $\Omega\left(U\right)=\{ u\in \mathbb{N}^{+}\left(U\right): \mathbf{g}\left(u,E_1\right)=-1 \}$ and, with a slight abuse of notation, we will prove that $\sigma:\Omega\left(U\right)\rightarrow\mathbb{P}\left(\mathcal{H}\right)$ given by $\sigma\left(u\right)=T_{\gamma_{u}}S\left(\pi^{\mathbb{N}}_{M}\left(u\right)\right)$ is differentiable.

Define the angle of the projection of $u\in \Omega\left(U\right)$ on $\mathrm{span}\{E_2,E_3\}$ by 
\[
\theta\left(u\right)=\arctan \left(\frac{\mathbf{g}\left(u,E_3\right)}{\mathbf{g}\left(u,E_2\right)}\right)\in \mathbb{R}
\]
and the curve of vectors at $\pi^{\mathbb{N}}_{M}\left(u\right)\in U$ given by
\[
W\left(u,s\right)=E_1\left(\pi^{\mathbb{N}}_{M}\left(u\right)\right) + \cos\left( \theta\left(u\right) + s \right)E_2\left(\pi^{\mathbb{N}}_{M}\left(u\right)\right) + \sin\left( \theta\left(u\right) + s \right)E_3\left(\pi^{\mathbb{N}}_{M}\left(u\right)\right) \in T_{\pi^{\mathbb{N}}_{M}\left(u\right)}M   .
\]

Assuming the notation of diagram (\ref{diagram-charts}), now we can define
\[
\overline{W}\left(u,s\right)=\xi^{-1}\circ \boldsymbol{\upgamma}\left(W\left(u,s\right)\right)\in \Omega\left(C\right)
\]
where $\xi=\left.\boldsymbol{\upgamma}\right|_{\Omega\left(C\right)}$, and then we can build
\[
f\left(u,s,\tau\right)=\mathrm{exp}_{\pi^{\mathbb{N}}_{M}\circ \overline{W}\left(u,s\right)}\left(\tau\cdot \overline{W}\left(u,s\right) \right) \in M   .
\]
Notice that for every fixed $u\in\Omega\left(U\right)$ the map $f_{u}\left(s,\tau\right)=f\left(u,s,\tau\right)$ is a lightlike geodesic variation with initial values at $C$ running the sky $S\left(\pi^{\mathbb{N}}_{M}\left(u\right)\right)$ and then its Jacobi field along $\gamma_u$ is 
\[
J_u\left(\tau\right)=\frac{\partial f}{\partial s}\left(u,0,\tau\right)
\]
that, by construction, it satisfies 
\begin{equation}\label{prop-sigma-1}
\langle J_u\rangle \in T_{\gamma_u}S\left(\pi^{\mathbb{N}}_{M}\left(u\right) \right) = \sigma\left(u\right)
\end{equation}
and its initial vectors for $\tau =0$ are given by 
\[
\left\{ 
\begin{tabular}{l}
$J_u\left(0\right)=\frac{\partial f}{\partial s}\left(u,0,0\right)$ \\
$J'_u\left(0\right)=\frac{D}{d \tau}\frac{\partial f}{\partial s}\left(u,0,0\right)$
\end{tabular}
\right.
\]
depending differentially on $u\in\Omega\left(U\right)$, so the map 
\begin{equation}\label{map-sigma-H}
\begin{tabular}{rcl}
$\Omega\left(C\right)$ & $\rightarrow$ & $\mathcal{H}$ \\
$u$ & $\mapsto$ & $\langle J_u\rangle$
\end{tabular}
\end{equation}
is differentiable.
Since $0 \neq \langle J_u\rangle \in \mathcal{H}$ for all $u$ then we can pass to the quotient obtaining that the map 
\[
u\mapsto \left[ J_u\right]= T_{\gamma_u}S\left(\pi^{\mathbb{N}}_{M}\left(u\right) \right) \in \mathbb{P}\left(\mathcal{H}\right)
\]
is differentiable. 
Then, in virtue of (\ref{prop-sigma-1}), $\sigma\left(u\right)=\left[ J_u\right]= T_{\gamma_u}S\left(\pi^{\mathbb{N}}_{M}\left(u\right) \right)$ is differentiable.   
\end{proof}

Recall that the distributions $\oplus$ and $\ominus$ defined in section~\ref{sec:Preliminaries} assign to each  $\gamma\in \mathcal{N}$, if they exist, the endpoints of the curve 
\begin{equation}\label{def-gammatilde}
\widetilde{\gamma}\left(\tau\right) = T_{\gamma}S\left(\gamma\left(\tau\right)\right) \in \mathbb{P}\left(\mathcal{H}_{\gamma}\right)
\end{equation}
where $\gamma=\gamma\left(\tau\right)$ is a parametrization of the light ray $\gamma$.  
Observe that it is possible to define the curve $\widetilde{\gamma}$ by 
\[
\widetilde{\gamma}\left(\tau\right)=\sigma\left( \left[\gamma'\left(\tau\right) \right] \right)
\]
that, because of Lemma \ref{Lemma-diffeo-sigma}, it is a differentiable curve. 

\begin{lemma}\label{Lemma-gamma-tilde-regular}
Given a maximal null geodesic $\gamma:I\rightarrow M$, then the curve $\widetilde{\gamma}\left(\tau\right)= T_{\gamma}S\left(  \gamma\left(\tau\right)\right)$ is regular for all $\tau\in I$.
\end{lemma}

\begin{proof}
Let us assume that $\gamma:I\rightarrow M$ is a null geodesic related to the metric $\mathbf{g}\in\mathcal{C}$.
Moreover, without any lack of generalization, we assume that $0\in I$, then it is sufficient to prove that $\widetilde{\gamma}'\left(0\right)\neq 0$. 
If $p=\gamma\left(0\right)\in M$, we can consider an orthonormal frame $\{ E_1\left(p\right), E_2\left(p\right), E_3\left(p\right) \}\subset T_p M$ taken from the frame used to define the charts of section~\ref{sec:charts} such that $p\in C\subset U$.
Recall that $E_1$ is timelike and $E_2, E_3$ are spacelike then we can write $\gamma'\left(0\right) = E_1\left(p\right)+\cos \theta ~E_2\left(p\right) + \sin \theta ~E_3\left(p\right)$. 
Let $\{ \mathbf{E}_1, \mathbf{E}_2, \mathbf{E}_3 \}$ be the basis of parallel vector fields, transported from  $\{ E_1\left(p\right), E_2\left(p\right), E_3\left(p\right) \}$, along $\gamma$. 
Then we have that 
\[
\gamma'\left(\tau\right) = \mathbf{E}_1\left(\gamma\left(\tau\right)\right)+\cos \theta ~\mathbf{E}_2\left(\gamma\left(\tau\right)\right) + \sin \theta ~\mathbf{E}_3\left(\gamma\left(\tau\right)\right)  . 
\]

Let us define $J_t$ as a Jacobi field along $\gamma$ such that $J_{t}\left(t\right)= 0 ~\left(\mathrm{mod}~\gamma'\left(t\right)\right)$, then $\langle J_t\rangle \in \widetilde{\gamma}\left(t\right)$ and so, by equations (\ref{contact}) and (\ref{relation-J-Jprima}), its initial vectors can be chosen as 
\begin{equation}\label{eq-gammatilde-initial}
\left\{ 
\begin{tabular}{l}
$J_t\left(0\right)= u_1\left(t\right)\left(-\sin \theta ~E_2\left(p\right) + \cos \theta ~E_3\left(p\right)\right)$ \\
$J'_t\left(0\right)= u_2\left(t\right)\left(-\sin \theta ~E_2\left(p\right) + \cos \theta ~E_3\left(p\right)\right)$
\end{tabular}
\right.
\end{equation}
Observe that the Jacobi field $J_t$ is just the image of $\gamma'\left(t\right)$ under the differentiable map (\ref{map-sigma-H}), that is $J_t \equiv J_{\gamma'\left(t\right)}$.
This implies that the functions $u_1\left(t\right)$ and $u_2\left(t\right)$ are differentiable.
Moreover, since $J_t \neq 0$ and $J_0\left(0\right)=0$, then $J'_0\left(0\right)\neq 0$ and so $u_1\left(0\right)=0$ and $u_2\left(0\right)\neq 0$.

The Jacobi field $J_t$ can be written by
\begin{equation}\label{J_t-gammatilde}
J_{t}\left(\tau\right) = \sum_{i=1}^{3}\eta_{i}\left(t,\tau\right) \mathbf{E}_i\left(\gamma\left(\tau\right)\right)  .
\end{equation}
If we substitute the expression (\ref{J_t-gammatilde}) in the differential equation (\ref{eq-Jacobi-fields}) then we have
\[
\sum_{i=1}^{3}\frac{d^2 \eta_{i}}{d\tau^2} \mathbf{E}_i = \sum_{i=1}^{3}\eta_{i}  R\left(\mathbf{E}_i, \gamma'\right)\gamma'
\]
and hence we obtain a system of three linear ordinary differential equations given by 
\[
\frac{d^2 \eta_{j}}{d\tau^2} = \sum_{i=1}^{3}\eta_{i} \mathbf{g}\left(\mathbf{E}_j , \mathbf{E}_j  \right) \mathbf{g}\left(R\left(\mathbf{E}_i, \gamma'\right)\gamma' , \mathbf{E}_j\right) \qquad  \text{ for } j=1,2,3  .
\]
whose solutions depend smooth and linearly on the initial values, then 
\[
\eta_{j}\left(t,\tau\right)=\sum_{i=1}^{2} a_{ji}\left(\tau\right)u_i\left(t\right) \qquad  \text{ for } j=1,2,3 .
\]
where $A\left(\tau\right)=\left(a_{ji}\left(\tau\right)\right)$ is a differentiable matrix.
Therefore
\begin{equation}\label{eq-gammatilde-tau}
\left\{ 
\begin{array}{l}
J_t\left(\tau\right)= \sum_{j=1}^{3}\sum_{i=1}^{2} a_{ji}\left(\tau\right)u_i\left(t\right)\mathbf{E}_j\left(\gamma\left(\tau\right)\right) \\
J'_t\left(\tau\right)= \sum_{j=1}^{3}\sum_{i=1}^{2} a'_{ji}\left(\tau\right)u_i\left(t\right)\mathbf{E}_j\left(\gamma\left(\tau\right)\right)
\end{array}
\right.  .
\end{equation}
If we evaluate (\ref{eq-gammatilde-tau}) at $\tau=0$ and compare it with (\ref{eq-gammatilde-initial}), we obtain the values 
\begin{equation}\label{eq-system-A-0}
A\left(0\right)=\left( a_{ji}\left(0\right) \right) = \left(\begin{array}{cr} 0 & 0 \\ -\sin\theta & 0 \\ \cos\theta & 0 \end{array} \right) ; \qquad  A'\left(0\right)=\left( a'_{ji}\left(0\right) \right) = \left(\begin{array}{cr} 0 & 0 \\  0 & -\sin\theta  \\ 0 & \cos\theta  \end{array} \right)  .
\end{equation}

By the condition $J_{t}\left(t\right)= 0 ~\left(\mathrm{mod}~\gamma'\left(t\right)\right)$, we get the system
\begin{equation}\label{eq-system-a-t}
\left\{
\begin{array}{l} a_{11}\left(t\right)u_1\left(t\right)+a_{12}\left(t\right)u_2\left(t\right)=\lambda \\ a_{21}\left(t\right)u_1\left(t\right)+a_{22}\left(t\right)u_2\left(t\right)=\lambda \cos\theta\\ a_{31}\left(t\right)u_1\left(t\right)+a_{32}\left(t\right)u_2\left(t\right)=\lambda \sin\theta
\end{array}
\right.  
\end{equation}
and calling $\left(B_1,B_2\right)=\left(a_{21}\sin\theta - a_{31}\cos\theta, \quad a_{22}\sin\theta - a_{32}\cos\theta\right)$, from the second and third equation of the system~(\ref{eq-system-a-t}), we obtain
\[
B_1\left(t\right)u_1\left(t\right)+B_2\left(t\right)u_2\left(t\right)=0\, , \qquad  \mathrm{ for \, 	\, all \,\, } t\, , 
\]
and by the values in (\ref{eq-system-A-0}), we can have that $\left(B_1\left(0\right),B_2\left(0\right)\right)= \left(-1,0\right)$ and $\left(B'_1\left(0\right),B'_2\left(0\right)\right)= \left(0,-1\right)$. 
Since $u_2\left(0\right)\neq 0$ and  $B_1\left(0\right)\neq 0$, there exists $\epsilon>0$ such that $u_2\left(t\right)\neq 0$ and  $B_1\left(t\right)\neq 0$ for all $t\in \left(-\epsilon,\epsilon\right)$, so we have
\[
u_1\left(t\right)= -\frac{B_2\left(t\right)}{B_1\left(t\right)}u_2\left(t\right) \text{ for all } t\in \left(-\epsilon,\epsilon\right)
\] 
and the curve $t\mapsto \langle J_t\rangle\in \mathcal{H}_{\gamma}$ is written in the coordinates (\ref{carta-H}) by 
\[
\varphi\left(J_t\right)=\left(x_0,y_0, \theta, -\frac{B_2\left(t\right)}{B_1\left(t\right)}u_2\left(t\right), u_2\left(t\right)\right) 
\]
whence the coordinates (\ref{carta-PH}) of $\widetilde{\gamma}\left(t\right)$ for $t\in \left(-\epsilon,\epsilon\right)$ are 
\[
\widetilde{\varphi}\left(\widetilde{\gamma}\left(t\right)\right)=\widetilde{\varphi}\left(\left[J_t\right]\right)=\left(x_0,y_0, \theta, \left[ -\frac{B_2\left(t\right)}{B_1\left(t\right)}u_2\left(t\right): u_2\left(t\right)\right]\right)= \left(x_0,y_0, \theta, \left[ -\frac{B_2\left(t\right)}{B_1\left(t\right)}: 1\right]\right)
\]
and because 
\[
\left. \left( -\frac{B_2\left(t\right)}{B_1\left(t\right) }\right)\right|_{t = 0} '= \frac{B_2\left(0\right)B'_1\left(0\right)-B'_2\left(0\right)B_1\left(0\right)}{B^2_1\left(0\right)}=-1\neq 0 
\]
then $\widetilde{\gamma}'\left(0\right)\neq 0$ as we claimed.
\end{proof}


\begin{proposition}\label{prop-difeo-sigma}
The map $\sigma:\mathbb{PN}\rightarrow\mathbb{P}\left(\mathcal{H}\right)$ defined in Lemma \ref{Lemma-diffeo-sigma}
is a diffeomorphism onto its image.
\end{proposition}

\begin{proof}
Using the same notation as in Lemma \ref{Lemma-diffeo-sigma}, we will show the statement for the map $\sigma:\Omega\left(U\right) \rightarrow \mathbb{P}\left(\mathcal{H}\right)$ given by $\sigma\left(u\right)=T_{\gamma_{u}}S\left(\pi^{\mathbb{N}}_{M}\left(u\right)\right)$.

Fix some $u\in \Omega\left(U\right)$.
With no lack of generality we can assume that $p=\pi^{\mathbb{N}}_{M}\left(u\right)\in C$, because in other case, since the neighbourhood $U$ is globally hyperbolic it is possible to choose another Cauchy surface containing $p$.

By Lemma \ref{Lemma-diffeo-sigma}, $\sigma$ is a differentiable and injective map. 
If we consider the restriction of $\sigma$ to $\Omega\left(C\right)$, then $\pi^{\mathbb{P}\left(\mathcal{H}\right)}_{\mathcal{N}}\circ \left.\sigma\right|_{\Omega\left(C\right)}=\xi$  where $\xi=\left.\boldsymbol{\upgamma}\right|_{\Omega\left(C\right)}:\Omega\left(C\right)\rightarrow \mathcal{N}_{U}$ is the diffeomorphism of diagram~(\ref{diagram-charts}). 
So, consider the differential 
\[
\left(d\pi^{\mathbb{P}\left(\mathcal{H}\right)}_{\mathcal{N}}\right)_{\sigma\left(u\right)}\circ \left(d\left.\sigma\right|_{\Omega\left(C\right)}\right)_{u}=\left(d\xi\right)_{u}  .
\]
Notice that $\left(d\pi^{\mathbb{P}\left(\mathcal{H}\right)}_{\mathcal{N}}\right)_{\sigma\left(u\right)}$ is surjective of rank equal to 3 and $\left(d\xi\right)_{u}$ is an isomorphism of rank equal to 3, then the rank of $\left(d\left.\sigma\right|_{\Omega\left(C\right)}\right)_{u}$ must be 3.

Now, we will study $d\sigma_u$ for a vector in $T_u \Omega\left(U\right)$ transversal to $T_u \Omega\left(C\right)$.
Take the null geodesic $\gamma=\gamma\left(t\right)$ such that $\gamma'\left(0\right)=u$, then the curve $c\left(t\right)=\gamma'\left(t\right)\in \Omega\left(U\right)$ is regular and transversal to $\Omega\left(C\right)$ at $u$. 
Observe that 
\[
\sigma\left(c\left(t\right)\right)= T_{\gamma}S\left(\gamma\left(t\right)\right) = \widetilde{\gamma}\left(t\right)
\]
and by Lemma~\ref{Lemma-gamma-tilde-regular}, we have 
\[
d\sigma_u\left(c'\left(t\right)\right)=  \widetilde{\gamma}'\left(t\right) \neq 0
\] 
and this show that $d\sigma_u$ is an isomorphism.
Then, the Inverse function Theorem assures that $\sigma$ is a local diffeomorphism for any $u\in \mathbb{PN}$ and, due to its injectivity, then $\sigma$ is a diffeomorphism onto its image. 
  \end{proof}

Since $\dim \mathbb{PN} = 4$ and $\dim \mathbb{P}(\mathcal{H}) = 4$, then the previous Proposition has the following consequence.

\begin{corollary}\label{corol-Ntilde}
The map $\sigma$ induces a differentiable structure on $\widetilde{\mathcal{N}}=\mathrm{Im}\left(\sigma \right)$ such that $\widetilde{\mathcal{N}} \subset \mathbb{P}(\mathcal{H})$ is an open submanifold.
\end{corollary}

Now, we will show the existence of a common inextensible future projective parameter $\mathbf{t}$ for all $\gamma\in \mathcal{N}_{U}$ such that the map $\left(\gamma,\mathbf{t}\right)\mapsto \gamma\left(\mathbf{t}\right)\in M$ is differentiable.
We will need the following proposition.

\begin{proposition}\label{prop-varepsilon}
For any $\gamma_0\in\mathcal{N}$ there exist $\mathcal{N}_U\subset \mathcal{N}$, an interval $\left(a,b\right)\subset \mathbb{R}$ and a diffeomorphism $\varepsilon:\mathcal{N}_U \times \mathbb{R}\rightarrow \mathbb{P}\left(\mathcal{H}_U\right)-\widetilde{\infty}$ such that $\widetilde{\infty}$ is a section of the bundle $\mathbb{P}\left(\mathcal{H}_U\right)\rightarrow \mathcal{N}_U$ where $\widetilde{\infty}\cap\overline{\widetilde{\mathcal{N}}}=\varnothing$ and the restriction $\varepsilon:\mathcal{N}_U \times \left(a,b\right)\rightarrow \mathbb{P}\left(\mathcal{H}_U\right)\cap\widetilde{\mathcal{N}}$ is the diffeomorphism defined by $\varepsilon\left(\gamma,s\right)=\widetilde{\gamma}\left(s\right)$.
\end{proposition}

\begin{proof}

Let us assume the notation of section \ref{sec:charts} and fix $\gamma_0\in \mathcal{N}$. 
By hypotheses, $\oplus, \ominus : \mathcal{N} \rightarrow \mathbb{P}\left(\mathcal{H}\right)$ are differentiable and regular distributions and therefore there exist an open $\mathcal{N}_{U}\subset \mathcal{N}$ neighbourhood of $\gamma_0$ and functions $\phi_{\oplus}:\mathcal{N}_{U}\rightarrow \mathbb{R}$ and $\phi_{\ominus}:\mathcal{N}_{U}\rightarrow \mathbb{R}$ such that 
$\phi_{\oplus}\left(\gamma\right)=\phi\left( \oplus_{\gamma} \right)$ and $\phi_{\ominus}\left(\gamma\right)=\phi\left( \ominus_{\gamma} \right)$ (see \cite[Prop.~2.7]{Ba15}).

Let us consider the coordinated chart $\left(\mathcal{N}_{U}, \psi=\left(x,y,\theta\right)\right)$ at $\gamma_0$ as in equation (\ref{carta-N}), and such that $\oplus_{\gamma} \neq \ominus_{\gamma}$ for all $\gamma\in \mathcal{N}_{U}$.
In this coordinate system, we have that 
\[
\left\{
\begin{tabular}{l}
$\phi\left( \oplus_{\gamma} \right)\equiv \phi_{\oplus}\left(\gamma\right)$ \\
$\phi\left( \ominus_{\gamma} \right) \equiv \phi_{\ominus}\left(\gamma\right)$
\end{tabular}
\right.
\]
and observe that if $P\in\mathbb{P}\left( \mathcal{H}_{U} \right)$ is a line of Jacobi fields on some $\gamma\in\mathcal{N}_{U}$ that annihilate at $\gamma\cap C$, then $\sigma^{-1}\left(P\right)\in \mathbb{PN}\left(C\right)$ and hence, there exists a differentiable function $\phi_0:\mathcal{N}_{U}:\rightarrow \mathbb{R}$ such that $\phi_0 = \phi \circ \sigma \circ \mu^{-1}$, that is $\phi_{0}\left(\gamma\right)=\phi\left(T_{\gamma}S\left(\gamma \cap C\right)\right)$,  where $\phi$ is the coordinate in $\widetilde{\mathcal{N}}_{U}=\{ P\in \widetilde{\mathcal{N}}: \pi^{\mathbb{P}\left(\mathcal{H}\right)}_{\mathcal{N}}\left(P\right)\in \mathcal{N}_{U} \}$ of (\ref{carta-PH}), the map $\sigma$ is the diffeomorphism (\ref{difeo-sigma}), $\widetilde{\mathcal{N}}$ is the image of $\sigma$ according Corollary  \ref{corol-Ntilde}, $\mu:\mathbb{PN}\left(C\right)\rightarrow \mathcal{N}_{U}$ is the diffeomorphism of diagram (\ref{diagram-charts}) given by $\mu\left(\left[u\right]\right)=\gamma_{\left[u\right]}$ and $\pi^{\mathbb{P}\left(\mathcal{H}\right)}_{\mathcal{N}}:\mathbb{P}\left( \mathcal{H} \right)\rightarrow \mathcal{N}$ the canonical projection.

For any $\gamma\in \mathcal{N}_{U}$, by the assumption of $\oplus_{\gamma} \neq \ominus_{\gamma}$, we can consider the projective map $\mathbf{t}_{\gamma}:\mathbb{P}\left(\mathcal{H}_{\gamma}\right) \rightarrow \mathbb{R}\cup \{\infty\}$ such that 
\begin{equation}\label{projective-param}
\begin{tabular}{ccccc}
$\mathbf{t}_{\gamma}\left(\oplus_{\gamma}\right) = 1$, & & $\mathbf{t}_{\gamma}\left(\ominus_{\gamma}\right) = -1$, & & $\mathbf{t}_{\gamma}\left(\phi^{-1}\left(\phi_0\left(\gamma\right)\right)\right) = 0$
\end{tabular}
\end{equation}

Let us denote by $\widetilde{\infty}=\{ P\in \mathbb{P}\left( \mathcal{H}_{U} \right): \mathbf{t}_{\gamma}\left(P\right)=\infty \text{ for } P\in \mathbb{P}\left( \mathcal{H}_{\gamma} \right) \}$. 
So, the function $\mathbf{t}:\mathbb{P}\left( \mathcal{H}_{U} \right)-\widetilde{\infty}\rightarrow\mathbb{R}$ verifying (\ref{projective-param}) can be found to have the form 
\begin{equation}\label{projectivity}
 \mathbf{t}\left(P\right)=\frac{A\phi\left(P\right) + B}{C\phi\left(P\right) + D}\, ,
 \end{equation}
 where $A,B,C,D\in \mathbb{R}$ depends on $\gamma$ and it becomes
\[
\mathbf{t}\left(P\right)=\frac{\left(\phi_{\ominus}-\phi_{\oplus}\right)\left(\phi\left(P\right)-\phi_0\right)}{\left(2\phi_0 -\left(\phi_{\oplus}+\phi_{\ominus}\right)\right)\phi\left(P\right) +\left(2\phi_{\oplus}\phi_{\ominus}-\phi_0\left(\phi_{\oplus}+\phi_{\ominus}\right)\right)}
\] 
where for brevity, we have denoted $\phi_0=\phi_0\left(\pi^{\mathbb{P}\left(\mathcal{H}\right)}_{\mathcal{N}}\left(P\right)\right)$, $\phi_{\oplus}=\phi_{\oplus}\left( \pi^{\mathbb{P}\left(\mathcal{H}\right)}_{\mathcal{N}}\left(P\right)\right)$ and $\phi_{\ominus}=\phi_{\ominus}\left( \pi^{\mathbb{P}\left(\mathcal{H}\right)}_{\mathcal{N}}\left(P\right)\right)$.

Since, by hypothesis, $\partial^{+} \widetilde{\mathcal{N}}_{U}=\left\{ \oplus_{\gamma}: \gamma\in \mathcal{N}_{U} \right\}$ and $\partial^{-} \widetilde{\mathcal{N}}_{U}=\left\{ \ominus_{\gamma}: \gamma\in \mathcal{U} \right\}$ are differentiable hypersurfaces in $\mathbb{P}\left( \mathcal{H}_{U} \right)$ (diffeomorphic to $\mathcal{N}_{U}$), as well as $\sigma\circ \mu^{-1}\left(\mathcal{N}_{U}\right)$, then the functions $\phi_{\oplus}\circ \pi^{\mathbb{P}\left(\mathcal{H}\right)}_{\mathcal{N}}$ and $\phi_{\ominus}\circ \pi^{\mathbb{P}\left(\mathcal{H}\right)}_{\mathcal{N}}$  are differentiable, as well as $\phi_0 \circ \pi^{\mathbb{P}\left(\mathcal{H}\right)}_{\mathcal{N}}$, therefore $\mathbf{t}$ is a differentiable function. 

Since, 
\[
\frac{d\mathbf{t}}{d\phi}=\frac{2\left(\phi_{\ominus}-\phi_{\oplus}\right)\left(\phi_{\oplus}-\phi_{0}\right)\left(\phi_{\ominus}-\phi_{0}\right)}{\left[\left(2\phi_0 -\left(\phi_{\oplus}+\phi_{\ominus}\right)\right)\phi\left(P\right) +\left(2\phi_{\oplus}\phi_{\ominus}-\phi_0\left(\phi_{\oplus}+\phi_{\ominus}\right)\right)\right]^2}\neq 0
\]
we can replace the coordinate $\phi$ by $\mathbf{t}$ as a new coordinate, then $\widetilde{\psi}=\left(c,\theta,\mathbf{t}\right)$ becomes a new coordinate system. 

Observe that for any fixed $\gamma\in \mathcal{N}$ such that $\psi\left(\gamma\right)=\left(c,\theta\right)$, the curve parametrized by $\mathbf{t}=s$ such that, in the chart $\widetilde{\psi}$, is written by 
\[
\widetilde{\psi}\left(\widetilde{\gamma}\left(s\right)\right)= \left(c,\theta, s\right)
\]
is precisely $\widetilde{\gamma}\left(s\right)\in \mathbb{P}\left( \mathcal{H}_{U} \right)$ for $s\in\left(-1,1\right)$.

In fact, if we use the coordinates $\psi$ in $\mathcal{N}$ of equation (\ref{carta-N}) and $\widetilde{\psi}=\left(c,\theta,\mathbf{t}\right)$ in $\mathbb{P}\left(\mathcal{H}\right)$, then the map 
\begin{equation}\label{eq-varepsilon}
\begin{tabular}{rrcl}
$\varepsilon:$ & $\mathcal{N}_{U}\times \mathbb{R}$ & $\rightarrow$ & $\mathbb{P}\left(\mathcal{H}_{U}\right)-\widetilde{\infty}$ \\
& $\left(\gamma,\mathbf{t}\right)$ & $\mapsto$ & $\widetilde{\gamma}\left(\mathbf{t}\right)$
\end{tabular}
\end{equation}
can be expressed in coordinates by 
\[
\left(\left(x,y,\theta\right),\mathbf{t}\right) \longmapsto \left(x,y,\theta,\mathbf{t}\right)
\] 
hence, trivially it is a diffeomorphism such that the restriction $\left.\varepsilon\right|_{\mathcal{N}_{U}\times \left(-1,1\right)}$ is also a diffeomorphism such that $\left.\varepsilon\right|_{\mathcal{N}_{U}\times \left(-1,1\right)}\left(\gamma,\mathbf{t}\right)=\widetilde{\gamma}\left(\mathbf{t}\right)$. 
  \end{proof}

Observe that $\widetilde{\gamma}\left(s\right)$ with $s\in \left(-1,1\right)$ corresponds with a line of Jacobi fields along $\gamma$ such that they are proportional to $\gamma'$ at some point in $\gamma\subset M$, meaning that all those Jacobi fields are tangent to the sky of the respective point at $M$. 
By the expression in coordinates of $\varepsilon$ in equation (\ref{eq-varepsilon}), the curve $\widetilde{\gamma}$ can be extended smoothly by 
\[
\begin{tabular}{rcl}
$\mathbb{R}$ & $\rightarrow$ & $\mathbb{P}\left( \mathcal{H}\left(\mathcal{U}\right) \right) - \widetilde{\infty}$ \\
$s$ & $\mapsto$ & $\widetilde{\gamma}\left(s\right)$
\end{tabular}
\]
and, clearly we have
\begin{equation}\label{extension-gamma-prima-tilde}
\widetilde{\gamma}'\left(s\right)=\left( \frac{\partial}{\partial \mathbf{t}} \right)_{\widetilde{\gamma}\left(s\right)}
\end{equation} 
becoming a regular curve for all $s\in \mathbb{R}$.  

In section \ref{sec:canon-ext}, we will use the projective parameter found in the proof of Proposition \ref{prop-varepsilon} as an auxiliary tool, but any parameter $s\in \left[a,b\right]$ such that there is a diffeomorphism $h:\left[a,b\right]\rightarrow\left[-1,1\right]$ where $\mathbf{t}=h\left(s\right)$, is another admissible parameter. This notion will be introduced in Definition \ref{def-parametrizations} of Section \ref{sec:lextensions}. 
Notice that, for any admissible parameter $s\in \left[a,b\right]$, $\widetilde{\gamma}\left(s\right)\in \mathbb{P}\left(\mathcal{H}_{\gamma}\right)$ is regular and transversal to $\partial^{\pm}\widetilde{\mathcal{N}}$.

\begin{remark}
Since $\varepsilon$ is a diffeomorphism,the map $\pi^{\mathbb{PN}}_{M}\circ\sigma^{-1}\circ \varepsilon \left(\gamma,\mathbf{t}\right)=\gamma\left(\mathbf{t}\right)\in M$ is differentiable for $\left(\gamma,\mathbf{t}\right)\in \mathcal{N}_{U}\times\left(-1,1\right)$ obtaining a common parameter $\mathbf{t}\in\left(-1,1\right)$ for all $\gamma\in \mathcal{N}_{U}$.

When $M$ is globally hyperbolic, the function $\mathbf{t}$ can be smoothly defined for the whole $\mathbb{P}\left(\mathcal{H}\right)$ since $\mathbb{P}\left(\mathcal{H}\right)\simeq \mathcal{N}\times \mathbb{R}\cup\{\infty\}\simeq C\times \mathbb{S}^{1}\times \mathbb{R}\cup\{\infty\}$ where $C$ is a global Cauchy surface. 
Moreover, the map $\varepsilon$ can also be defined globally for all $\mathcal{N}\times \mathbb{R}$ and $\varepsilon:\mathcal{N}\times \left(-1,1\right)\rightarrow\widetilde{\mathcal{N}}$ is a diffeomorphism. 
In this case, $\mathbf{t}\in\left(-1,1\right)$ can define a parametrization for $\gamma\in \mathcal{N}$ by 
\[
\gamma\left(\mathbf{t}\right)=\pi^{\mathbb{PN}}_{M}\circ \sigma^{-1}\circ\varepsilon\left(\gamma,\mathbf{t}\right)\in M 
\] 
obtaining a \emph{universal projective parameter} for all maximal $\gamma\in \mathcal{N}$.
\end{remark}


\section{The boundary $\partial\widetilde{\mathcal{N}}$ of the blow up space $\widetilde{\mathcal{N}}$}

Because of Cor. \ref{corol-Ntilde} we may consider the blown up space $\widetilde{\mathcal{N}}$ of $M$ as an open submanifold of the contact Grassmannian $\mathbb{P}(\mathcal{H})$, then it has a natural topological boundary $\partial \widetilde{\mathcal{N}}$ as a subset of $\mathbb{P}(\mathcal{H})$.   It was shown in \cite{Ba17} that the closure $\overline{ \widetilde{\mathcal{N}}} =  \widetilde{\mathcal{N}} \cup \partial \widetilde{\mathcal{N}}$ is a smooth manifold with boundary, but for the sake of completeness we will sketch the proof here.

Notice that if $M$ is a 3-dimensional L-space it is possible to define the maps:
$$
\begin{tabular}{rcl}
$\ominus \colon \mathcal{N}$ & $\rightarrow$ & $\mathbb{P}\left(\mathcal{H}\right)$ \\
 $\gamma$ & $\mapsto$ & $\ominus\left(\gamma\right)=\ominus_{\gamma}$ 
\end{tabular}
\hspace{7mm} \mathrm{and} \hspace{7mm}
\begin{tabular}{rrcl}
$\oplus:$ & $\mathcal{N}$ & $\rightarrow$ & $\mathbb{P}\left(\mathcal{H}\right)$ \\
 & $\gamma$ & $\mapsto$ & $\oplus\left(\gamma\right)=\oplus_{\gamma} \, .$ 
\end{tabular}
$$
We will use these maps to identify the boundary of $\widetilde{\mathcal{N}}$ with $\mathcal{N}$ as the union of their graphs.

\begin{proposition}\label{prop-Low-boundary}
Let $M$ be a 3--dimensional L-space such that $\oplus_\gamma \neq \ominus_\gamma$ for all $\gamma \in \mathcal{N}$.  Then the closure $\overline{\widetilde{\mathcal{N}}}$ of the blow--up space $\widetilde{\mathcal{N}}$ is a smooth manifold with boundary embedded in $\mathbb{P}(\mathcal{H})$, moreover $\partial\overline{\widetilde{\mathcal{N}}} =  \{ \mathrm{graph}\oplus\} \cup\, \{ \mathrm{graph} \, \ominus\}$.
\end{proposition}

\begin{proof}
Since $\ominus_{\gamma}$ and $\oplus_{\gamma}$ are defined by the limit of $\widetilde{\gamma}\left(s\right)$ at the endpoints and $\widetilde{\gamma}$ is locally injective then $\widetilde{\gamma}$ must be a connected open set in $\mathbb{P}\left(\mathcal{H}_{\gamma}\right)\simeq \mathbb{S}^1$ with boundary $\left\{\ominus_{\gamma},\oplus_{\gamma}\right\}$.
Now, consider $P\in \mathbb{P}\left(\mathcal{H}\right)$ such that there exist $\gamma\in \mathcal{N}$ verifying $\ominus_{\gamma} = P$ and a coordinate chart $\widetilde{\varphi}=\left(x,y,\theta,\phi\right)$ at $P$ as in (\ref{carta-H}). 
Since $\ominus$ is a distribution on $\mathcal{N}$, the point $\ominus_{\gamma}\in \mathbb{P}\left(\mathcal{H}_{\gamma}\right)\subset \mathbb{P}\left(\mathcal{H}\right)$  depends smoothly on the light ray $\gamma$. 
Hence the function $\phi\circ\ominus:\mathcal{N}\rightarrow \left[0,2\pi\right)\simeq \mathbb{S}^1$ depends differentiably on the coordinates $\left(x,y,\theta\right)$.
Obviously, the same rules for $\oplus$.
Let us denote by $\phi_{\ominus}=\phi_{\ominus}\left(x,y,\theta\right)$ and $\phi_{\oplus}=\phi_{\oplus}\left(x,y,\theta\right)$ the coordinate representation of the functions $\phi\circ\ominus$ and $\phi\circ\oplus$ respectively.

Clearly $\partial\overline{\widetilde{\mathcal{N}}}\subset \{ \mathrm{graph} \oplus\} \cup\, \{ \mathrm{graph} \ominus \}  = \left\{ \ominus_{\gamma},\oplus_{\gamma}:\gamma\in \mathcal{N}\right\}$. 
 Consider now an open set $\mathcal{U}\subset \mathcal{N}$. 
Because $\ominus_{\gamma}\neq\oplus_{\gamma}$ for any $\gamma\in \mathcal{U}$ we can choose, without any lack of generality, a diffeomorphism $\left[0,2\pi\right)\simeq \mathbb{S}^1$ such that 
\[
0<\phi_{\ominus}\left(x,y,\theta\right) < \phi_{\oplus}\left(x,y,\theta\right) < 2\pi
\]
for all $\left(x,y,\theta\right)$ (restricting the domain of $\phi_{\ominus}$ and $\phi_{\oplus}$ if needed). 
Then, for all $\gamma\in \mathcal{U}$ the points in $\overline{\widetilde{\mathcal{U}}}$  (recall that $\widetilde{\mathcal{U}} = \pi_\mathcal{N}^{-1}(\mathcal{U}) = \bigcup_{\gamma \in \mathcal{U}} \widetilde{\gamma}$ is a cilindrical open subset in $\mathbb{P}(\mathcal{H})$ ),  can be written as
\[
\overline{\widetilde{\mathcal{U}}}\simeq \left\{ \left(x,y,\theta, \phi\right):\phi_{\ominus}\left(x,y,\theta\right)\leq \phi \leq \phi_{\oplus}\left(x,y,\theta\right) \right\} \, ,
\] 
describing a manifold with boundary.   Notice that using the projective parameter $\mathbf{t}$ discussed in the previous section, we get 
$$
\overline{\widetilde{\mathcal{U}}}\simeq \left\{ P \in \mathbb{P}(\mathcal{H}): \pi_\mathcal{N}(P) \in \mathcal{U}\, , -1  \leq \mathbf{t}(P) \leq 1 \right\} \, .
$$ 
Then $\left\{ \ominus_{\gamma},\oplus_{\gamma}: \gamma\in \mathcal{U} \right\} \subset \partial \overline{\widetilde{\mathcal{N}}}$ and, since $\ominus$ and $\oplus$ are regular distributions, the condition $\ominus_{\gamma}\neq \oplus_{\gamma}$ is open in $\mathcal{N}$, therefore we have that $ \{ \mathrm{graph} \oplus\} \cup\, \{ \mathrm{graph} \ominus \} = \left\{ \ominus_{\gamma},\oplus_{\gamma}: \gamma\in \mathcal{N} \right\}  \subset \partial \overline{\widetilde{\mathcal{N}}}$, which concludes the proof.
 \end{proof} 

In what follows, in order to avoid cumbersome notations, we will just write $\partial\widetilde{\mathcal{N}}$ instead of $\partial	\overline{\widetilde{\mathcal{N}}}$.    

As a consequence of the previous proposition, if the distributions $\oplus,\ominus$ are different and $\mathcal{N}$ is connected, the boundary  $\partial\widetilde{\mathcal{N}}$ has two connected components $\partial^+\widetilde{\mathcal{N}}$ and  $\partial^-\widetilde{\mathcal{N}}$ that can be identified with $\{ \mathrm{graph} \oplus\}$ and $\{ \mathrm{graph} \ominus\}$ respectively.  In what follows we will keep this notation for the boundary, then $\partial\widetilde{\mathcal{N}} = \partial^+\widetilde{\mathcal{N}} \cup \partial^-\widetilde{\mathcal{N}}$, and we will concentrate our attention on either $\partial^+\widetilde{\mathcal{N}}$ or $\partial^-\widetilde{\mathcal{N}}$ unless stated otherwise.


\section{The canonical extension of $M$}\label{sec:canon-ext}

Now, the aim of this section is to provide the analytical details of the construction of the extension of the canonical distribution $\mathcal{D}^\sim$ to the boundary of $\widetilde{\mathcal{N}}$ and to blow down the completed space $\overline{\widetilde{\mathcal{N}}}$ to obtain the seeked extension $\overline{M}$ of $M$.                             

First, we study the canonical $1$--dimensional distribution $\mathcal{D}^{\sim}$ in $\widetilde{\mathcal{N}}$ (see Sect. \ref{sec:blowing}).   Notice that the orbit of $\mathcal{D}^{\sim}$ passing through $\widetilde{\gamma}\left(\mathbf{t}\right)\in \widetilde{\mathcal{N}}$ comprises all the lines of Jacobi fields (as tangent vectors in $T\mathcal{N}$) which annihilate at $\gamma\left(\mathbf{t}\right)\in M$, that is, the tangent lines to sky $X = S(\gamma\left(\mathbf{t}\right))$. 
If we denote by $\mathcal{P}$ the distribution in $\mathbb{PN}$ whose orbits are the fibres of the bundle $\pi^{\mathbb{PN}}_{M}:\mathbb{PN}\rightarrow M$, then trivially, the map $\zeta:M\rightarrow \mathbb{PN}/\mathcal{P}$ defined by $\zeta\left(q\right)=\mathbb{PN}_q$ is a diffeomorphism.
Hence, we can define the distribution $\mathcal{D}^{\sim}$ as the one whose orbit passing by $\widetilde{\gamma}\left(\mathbf{t}\right)\in \widetilde{\mathcal{N}}$ is given by $\sigma\left(\mathbb{PN}_{\gamma\left(\mathbf{t}\right)}\right)=\{ \sigma\left(\left[v\right]\right)\in\widetilde{\mathcal{N}}:\left[v\right]\in \mathbb{PN}_{\gamma\left(\mathbf{t}\right)}  \}$.
Observe that the orbits of  $\mathcal{D}^{\sim}$ are compact, then  $\mathcal{D}^{\sim}$ is a regular distribution and therefore $\widetilde{\mathcal{N}}/\mathcal{D}^{\sim}$ is a differentiable manifold and the canonical quotient map $\widetilde{\pi}:\widetilde{\mathcal{N}} \rightarrow \widetilde{\mathcal{N}}/\mathcal{D}^{\sim}$ is a submersion. 
Now, we can define the map $\widetilde{\sigma}:\mathbb{PN}/ \mathcal{P}\rightarrow \widetilde{\mathcal{N}}/\mathcal{D}^{\sim}$ by $\widetilde{\sigma}\left(\mathbb{PN}_q\right)=\sigma\left(\mathbb{PN}_q\right)\in \widetilde{\mathcal{N}}/\mathcal{D}^{\sim}$. Then we have the following diagram

\begin{equation}\label{diagram-distrib-D}
\begin{tikzpicture}[every node/.style={midway}]
\matrix[column sep={6em,between origins},
        row sep={2em}] at (0,0)
{ ; &  \node(PN)   { $\mathbb{PN}$}  ; & \node(Nt)   { $\widetilde{\mathcal{N}}$} ; \\
 \node(M)   { $M$}; &  \node(PN-P)   { $\mathbb{PN}/ \mathcal{P}$}  ; & \node(Nt-D)   { $\widetilde{\mathcal{N}}/\mathcal{D}^{\sim}$} ;  \\} ; 
\draw[->] (PN) -- (Nt) node[anchor=south]  {$\sigma$};
\draw[->] (PN) -- (PN-P) node[anchor=east]  {$\kappa$};
\draw[->] (Nt)   -- (Nt-D) node[anchor=west] {$\widetilde{\pi}$};
\draw[->] (M)   -- (PN-P) node[anchor=north] {$\zeta$};
\draw[->] (PN-P)   -- (Nt-D) node[anchor=north] {$\widetilde{\sigma}$}; 
\end{tikzpicture}
\end{equation}
where $\kappa$ and $\widetilde{\pi}$ are submersions and $\sigma$, $\zeta$ and $\widetilde{\sigma}$ are diffeomorphisms. 
Therefore, we can observe that 
\begin{equation}\label{diffeo-S}
\widetilde{S}=\widetilde{\sigma}\circ \zeta: M \rightarrow \widetilde{\mathcal{N}}/\mathcal{D}^{\sim}
\end{equation}
is a diffeomorphism. 
This fact was previously shown in a different way in \cite[Prop.~2.6]{Ba17} and is the essence of the blowing up and down principle discussed in Sect. \ref{sec:blowing}.

The construction of the smooth extension $\overline{M}$ of $M$, (see Corollary  \ref{Corollary-ext} below) is a consequence of the following theorem, that it properly constitutes the main result of this paper as it shows that the canonical distribution $\mathcal{D}^\sim$ in the blow up space $\widetilde{\mathcal{N}}$ of a $L$-spacetime  $M$ extends smoothly to its boundary.

\begin{theorem}[Main Theorem]\label{main_thm}  Let $\widetilde{\mathcal{N}}$ be the blow up of the  L-spacetime $M$ with its canonical distribution $\mathcal{D}^{\sim}$ such that $\widetilde{\mathcal{N}}/\mathcal{D}^{\sim} \cong M$.  Let $\partial^+ \mathcal{D}^{\sim}$ be the distribution on $\partial^+ \widetilde{\mathcal{N}}$ image under the diffeomorphism $\oplus \colon \mathcal{N} \to \partial^+ \widetilde{\mathcal{N}}$ of the regular distribution $\gamma \mapsto \oplus_\gamma$ on $\mathcal{N}$.   Then $\overline{\mathcal{D}^{\sim}} =  \mathcal{D}^{\sim} \cup \partial ^+\mathcal{D}^{\sim}$ is a smooth distribution on $\overline{\widetilde{\mathcal{N}}} = \widetilde{\mathcal{N}} \cup \partial^+ \widetilde{\mathcal{N}}$.  The same result holds for $\partial^- \mathcal{D}^{\sim}$ on $\partial^-\widetilde{\mathcal{N}}$.
\end{theorem}

The key idea to prove it is to construct for each $\gamma \in \mathcal{N}$ a smooth biparametric variation $\gamma_{(\mathbf{t},s)}$, $1-\delta < \mathbf{t} \leq 1$, $|s| < \epsilon$, with $\gamma_{(\cdot, 0) } = \gamma$ and $\gamma (\mathbf{t}) \in \gamma_{(\mathbf{t},s)}$  (see Fig. \ref{diapositiva1}), in such a way that the 
curves $s \mapsto \langle J_{(\mathbf{t},s)} \rangle$ defined by the corresponding Jacobi fields $J_{(\mathbf{t},s)}$, $1-\delta < \mathbf{t} < 1$, will describe the integral curves of $\mathcal{D}^{\sim}$, and the curves $s \mapsto \langle J_{(1,s)} \rangle$ will be the integral curves of $\partial^+ \mathcal{D}^{\sim}$.

The construction of $\gamma_{(\mathbf{t},s)}$ will rely on a number of observations and definitions that we will be the subject of the following paragraphs.

First we will define, for any given light ray $\gamma_0\in \mathcal{N}$, a differentiable map $\Phi$ describing the orbits of the distribution $\mathcal{D}^{\sim}$ and then we will extend it up to $\partial^+ \widetilde{\mathcal{N}}$.

Consider an auxiliary metric $\mathbf{g}\in \mathcal{C}$ and fix some $\gamma_0\in \mathcal{N}$.
Let $\mathcal{N}_{U}\subset \mathcal{N}$ be an open neighbourhood of $\gamma_0$ as the one used in the definition of the charts (\ref{carta-N}), that is $\mathcal{N}_{U}$ is diffeomorphic to $C\times \mathbb{S}^1$ where $C\subset U$ is a local spacelike Cauchy surface 
where $U\subset M$ is a globally hyperbolic open set such that $\gamma\cap U\neq \varnothing$. 
Let us assume that all light rays $\gamma\in \mathcal{N}_{U}$ are parametrized such that $\gamma \left(0\right)\in C$.

Without any lack of generality $U$ can be assumed to be relatively compact, and since $M$ is strongly causal, then there is no imprisoned light ray in $\overline{U}$ \cite[Prop.~6.4.7]{HE}.

Moreover, consider $\{ E_1\left(c\right), E_2\left(c\right), E_3\left(c\right) \}$, $c =(x,y)$ local coordinates for points in $C$, be the orthonormal frame on the local Cauchy surface $C$ used in the definition of the charts of section \ref{sec:charts} such that $E_2, E_3$ are tangent to $C$ and $E_1$ is timelike.

For a light ray $\gamma$ with coordinates $\psi\left(\gamma\right)=\left(c,\theta\right)$, define $\left\{\mathbf{E}_i\left(\gamma,\mathbf{t}\right)\right\}_{i=1,2,3}$ as the extension of the frame $\left\{E_i\left(c\right)\right\}_{i=1,2,3}$ by parallel transport to $\gamma\left(\mathbf{t}\right)$ along $\gamma$ with respect to the metric $\mathbf{g}$. 

The smooth dependence of the frames $\left\{\mathbf{E}_i\left(\gamma,\mathbf{t}\right)\right\}_{i=1,2,3}$ on $\left(\gamma,\mathbf{t}\right)$ follows from regular dependence on parameters of solutions of initial value problems of ODEs \cite[Ch.~5]{Ha}.

Now, it is possible to define the lightlike vector
\[
V\left(\gamma,\mathbf{t},s\right) = \mathbf{E}_1\left(\gamma,\mathbf{t}\right) + \cos\left(\theta + s \right)\mathbf{E}_2\left(\gamma,\mathbf{t}\right)+\sin\left(\theta + s \right)\mathbf{E}_3\left(\gamma,\mathbf{t}\right)\in \mathbb{N}
\]
depending smoothly on $\left(\gamma,\mathbf{t}\right)$ and let us denote its corresponding line by
\[
\Lambda\left(\gamma,\mathbf{t},s\right)=\left[V\left(\gamma,\mathbf{t},s\right)\right] = \mathrm{span}\{V\left(\gamma,\mathbf{t},s\right)\}\in \mathbb{PN}  .
\]
Using the maps $\sigma$ and $\varepsilon$ and the canonical projections $p_1:\mathcal{N}\times \left(-1,1\right)\rightarrow \mathcal{N}$ and $p_2:\mathcal{N}\times \left(-1,1\right)\rightarrow \left(-1,1\right)$, we can define the differentiable maps
\begin{equation}\label{eq-def-distribution}
\begin{tabular}{c}
$\widetilde{X}\left(\gamma,\mathbf{t},s\right)=\sigma\left( \Lambda\left(\gamma,\mathbf{t},s\right)  \right) \in \widetilde{\mathcal{N}}$ \\
$X\left(\gamma,\mathbf{t},s\right)= p_1 \circ \varepsilon^{-1}\left(\widetilde{X}\left(\gamma,\mathbf{t},s\right) \right) \in \mathcal{N}$ \\
$\tau\left(\gamma,\mathbf{t},s\right)=p_2 \circ \varepsilon^{-1}\left(\widetilde{X}\left(\gamma,\mathbf{t},s\right) \right) \in \left(-1,1\right)$
\end{tabular}
\end{equation}
where, for fixed $\left(\gamma,\mathbf{t}\right)\in\mathcal{N}_{U}\times\left(-1,1\right)$, the curve $X_{\left(\gamma,\mathbf{t}\right)}\left(s\right)=X\left(\gamma,\mathbf{t},s\right)$ describes the segment of the sky of $\gamma\left(\mathbf{t}\right)$ intersecting the neighbourhood $U\subset M$ (see Fig. \ref{diapositiva1}), the function $\tau_{\left(\gamma,\mathbf{t}\right)}\left(s\right)=\tau\left(\gamma,\mathbf{t},s\right)$ corresponds to the value of the parameter at $\gamma\left(\mathbf{t}\right)$ from $C$ along the light ray $X_{\left(\gamma,\mathbf{t}\right)}\left(s\right)$; and $\widetilde{X}_{\left(\gamma,\mathbf{t}\right)}\left(s\right)=\widetilde{X}\left(\gamma,\mathbf{t},s\right)$ is a curve of lines of Jacobi fields tangent to their corresponding light ray $X_{\left(\gamma,\mathbf{t}\right)}\left(s\right)$ at the point $\gamma\left(\mathbf{t}\right)$.

Then the family of light rays $X(\gamma, \mathbf{t},s)$ is the biparametric variation we are looking for:
\[
\gamma_{\left(\mathbf{t},s\right)}= X\left(\gamma,\mathbf{t},s\right)\in \mathcal{N} \, .
\] 
Moreover, for fixed $\left(\gamma,\mathbf{t}\right)$, we define the curves $V_{\left(\gamma,\mathbf{t}\right)}:\left[0,2\pi\right)\rightarrow \mathbb{N}_{\gamma\left(\mathbf{t}\right)}$ by $V_{\left(\gamma,\mathbf{t}\right)}\left(s\right)=V\left(\gamma,\mathbf{t}, s \right)$ and $\Lambda_{\left(\gamma,\mathbf{t}\right)}:\left[0,2\pi\right)\rightarrow \mathbb{PN}_{\gamma\left(\mathbf{t}\right)}$ by $\Lambda_{\left(\gamma,\mathbf{t}\right)}\left(s\right)=\Lambda\left(\gamma,\mathbf{t}, s \right)$. 
Then we have that $\varepsilon^{-1}\circ\sigma\left(\Lambda_{\left(\gamma,\mathbf{t}\right)}\left(s\right)\right)=\left(\gamma_{\left(\mathbf{t},s\right)}, \tau_{\left(\gamma,\mathbf{t}\right)}\left(s\right)\right)$ and so $\widetilde{\gamma}_{\left(\mathbf{t},s\right)}\left(\tau_{\left(\gamma,\mathbf{t}\right)}\left(s\right)\right)=\widetilde{X}_{\left(\gamma,\mathbf{t}\right)}\left(s\right)$. 
The following diagram, Fig. \ref{diagram-distrib-D-2}, illustrates these relations.

\begin{figure}[h]
\centering

\begin{tikzpicture}[every node/.style={midway}]
\matrix[column sep={8em,between origins},
        row sep={2em}] at (0,0)
{ ; &  \node(PN)   { $\mathbb{PN}_{\gamma\left(\mathbf{t}\right)}$}  ; & ; \\
 \node(Is) {$\left[0,2\pi\right)$} ; & ; & \node(St)  { $\widetilde{S}\left(\gamma\left(\mathbf{t}\right)\right)\subset\widetilde{\mathcal{N}}$ } ; \\
  ;  & \node(It) {$\left(-1,1\right)$} ; &  ; \\
 \node(N)   { $\mathcal{N}$}; & ; &  \node(N-It)   { $\mathcal{N}\times \left(-1,1\right)$} ;  \\} ;
\draw[->] (Is) -- (PN) node[anchor=south east]  {$\Lambda_{\left(\gamma,\mathbf{t}\right)}$};
\draw[->] (PN) -- (St) node[anchor=south west]  {$\sigma$};
\draw[->] (Is) -- (N) node[anchor=east] {$X_{\left(\gamma,\mathbf{t}\right)}$};
\draw[->] (N-It)  -- (N) node[anchor=north] {$p_1$};
\draw[->] (St)   -- (N-It) node[anchor=west] {$\varepsilon^{-1}$};
\draw[->] (Is)   -- (It) node[anchor=north east] {$\tau_{\left(\gamma,\mathbf{t}\right)}$}; 
\draw[->] (Is)   -- (St) node[anchor=south] {$\widetilde{X}_{\left(\gamma,\mathbf{t}\right)}$}; 
\draw[->] (N-It)   -- (It) node[anchor=south west] {$p_2$};
\end{tikzpicture}
\caption{ Diagram summarizing the relations between the maps $X_{(\gamma,\mathbf{t})}$,  $\widetilde{X}_{(\gamma,\mathbf{t})}$ and $\Lambda_{(\gamma,\mathbf{t})}$, and the structural maps $\sigma$ and $\epsilon$.}

\label{diagram-distrib-D-2}
\end{figure}

It can be observed that, since $\widetilde{\gamma}_{\left(\mathbf{t},s\right)}\left(\tau_{\left(\gamma,\mathbf{t}\right)}\left(s\right)\right)\in \widetilde{S}\left(\gamma\left(\mathbf{t}\right)\right) = \sigma (\mathbb{PN}_{\gamma(\mathbf{t})})$ then there exists a light ray $\mu\in S\left(\gamma\left(\mathbf{t}\right)\right)$ such that 
\[
T_{\gamma_{\left(\mathbf{t},s\right)}}S\left(\gamma_{\left(\mathbf{t},s\right)}\left(\tau_{\left(\gamma,\mathbf{t}\right)}\left(s\right)\right)\right)=T_{\mu}S\left(\gamma\left(\mathbf{t}\right)\right)
\]
hence $\mu=\gamma_{\left( \mathbf{t},s \right)}$.
Due to $M$ is light non--conjugate, then
\[
S\left(\gamma_{\left(\mathbf{t},s\right)}\left(\tau_{\left(\gamma,\mathbf{t}\right)}\left(s\right)\right)\right)= S\left(\gamma\left(\mathbf{t}\right)\right)
\]
and because $M$ distinguishes skies, then we obtain the following equation
\begin{equation}\label{eq-gamma-tau}
\gamma_{\left(\mathbf{t},s\right)}\left( \tau\left(\gamma,\mathbf{t},s\right)\right)=\gamma\left(\mathbf{t}\right) \, .
\end{equation}  
Observe that for $s=0$ we have $X\left(\gamma,\mathbf{t},0\right)=\gamma_{\left(\mathbf{t},0\right)}=\gamma$ for all $\mathbf{t}\in\left(-1,1\right)$ and hence $\tau\left(\gamma,\mathbf{t},0\right)=\mathbf{t}$.

Now, we will change the parameter $s$ to a more adequate one in some neighbourhood of the previously fixed light ray $\gamma_0$.  
For the auxiliary metric $\mathbf{g}\in \mathcal{C}$ in $M$, we consider the curves $c_{\left(\gamma,\mathbf{t}\right)}\left(s\right)= X\left(\gamma,\mathbf{t},s\right)\cap C \in C\subset U$ (see Fig. \ref{diapositiva1}). 
Since $C$ is a differentiable spacelike hypersurface, then the restriction $\left.\mathbf{g}\right|_{TC\times TC}$ is a Riemannian metric on $C$ and therefore we can parametrize the curves $c_{\left(\gamma,\mathbf{t}\right)}$ with the arc length parameter $\mathbf{s}$ defined in $C$ by the restriction of $\mathbf{g}$. 
If we take a Jacobi field $J_{\left(\gamma,\mathbf{t},s\right)}\in \widetilde{X}\left(\gamma,\mathbf{t},s\right)$, because $M$ is light non--conjugate, then
\[
0\,  (\mathrm{mod} \gamma'_{\left(\mathbf{t},s\right)}\left(0\right)) \neq J_{\left(\gamma,\mathbf{t},s\right)}\left(0\right)= \frac{d c_{\left(\gamma,\mathbf{t}\right)}\left(s\right)}{d s}= c'_{\left(\gamma,\mathbf{t}\right)}\left(s\right)
\] 
and therefore, $c_{\left(\gamma,\mathbf{t}\right)}$ is a regular curve and there exist a differentiable function $s=h\left(\gamma,\mathbf{t},\mathbf{s}\right)$ which permits to change the parameter. 
Abusing of the notation, we will keep the names of the maps $\widetilde{X}\left(\gamma,\mathbf{t},\mathbf{s}\right)$, $X\left(\gamma,\mathbf{t},\mathbf{s}\right)$, $\tau\left(\gamma,\mathbf{t},\mathbf{s}\right)$ with this new variable $\mathbf{s}$.

\begin{figure}[h]
  \centering
    \includegraphics[scale=1]{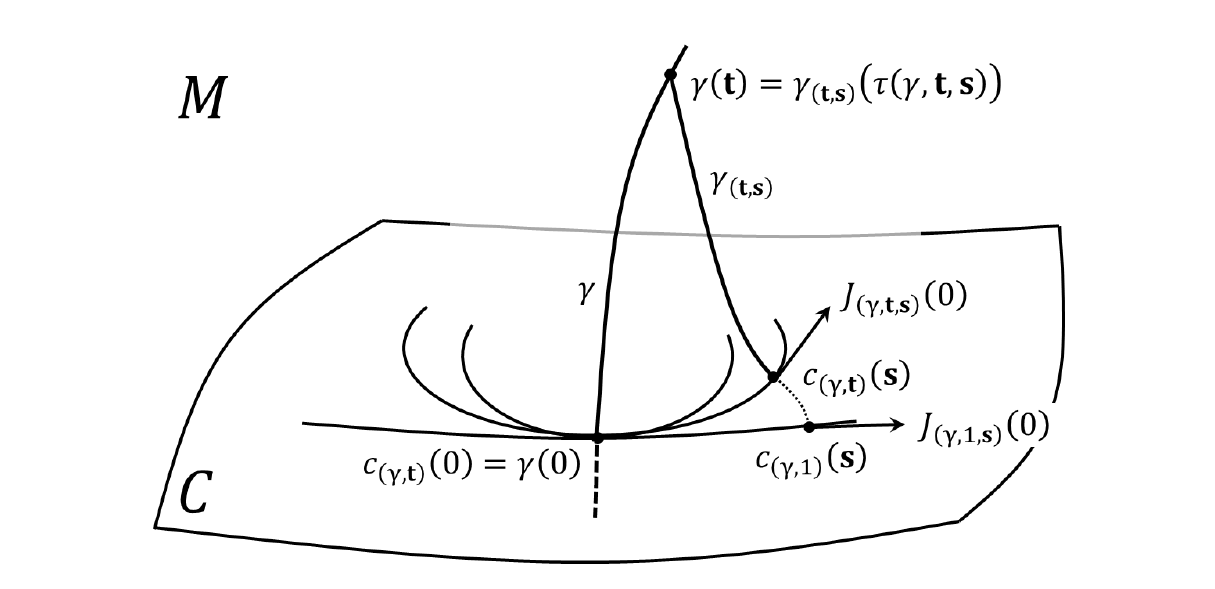}
  \caption{Illustration of the family $\gamma_{(\mathbf{t},s)}$ of light rays in the globally hyperbolic neighborhood $U$ determined by a local Cauchy surface $C$.  The trace of the sky $S(\gamma(\mathbf{t}))$ consists on the segments defined by the curves $\gamma_{(\mathbf{t},s)} = X(\gamma, \mathbf{t},s)$.  The curves $c_{\left(\gamma,\mathbf{t}\right)}$ are the traces in $C\subset M$ of $\gamma_{(\mathbf{t},s)}$.}
  \label{diapositiva1}
\end{figure}

Let us denote by $B\left(p, r\right)\subset C$ the ball centered at $p\in C$ of radius $r>0$ related to the metric $\left.\mathbf{g}\right|_{TC\times TC}$. 
Choose $\epsilon>0$ such that $B\left(\gamma_0 \cap C, 2\epsilon\right)\subset C$ and define 
\[
\mathcal{N}^{\epsilon}_{U}=\left\{ \gamma\in\mathcal{N}_{U}: \gamma \cap C \in B\left(\gamma_0  \cap C, \epsilon\right)\right\}  .
\]
Since $\mathbf{s}$ is the arc length parameter of $c_{\left(\gamma,\mathbf{t}\right)}$, then $c_{\left(\gamma,\mathbf{t}\right)}\left(\mathbf{s}\right)\in B\left(\gamma_0  \cap C, 2\epsilon\right)$ for any $\left(\gamma,\mathbf{t},\mathbf{s}\right)\in \mathcal{N}^{\epsilon}_{U}\times\left(-1,1\right)\times \left(-\epsilon,\epsilon\right)$ then, it is possible to restrict the maps $\widetilde{X}$, $X$ and $\tau$ as
\[
\begin{tabular}{l}
$\widetilde{X}:\mathcal{N}^{\epsilon}_{U}\times\left(-1,1\right)\times \left(-\epsilon,\epsilon\right)\longrightarrow \widetilde{\mathcal{N}}_{U}$ \\
$X: \mathcal{N}^{\epsilon}_{U}\times\left(-1,1\right)\times \left(-\epsilon,\epsilon\right)\longrightarrow  \mathcal{N}_{U}$ \\
$\tau: \mathcal{N}^{\epsilon}_{U}\times\left(-1,1\right)\times \left(-\epsilon,\epsilon\right)\longrightarrow \left(-1,1\right)$
\end{tabular}
\]

By construction, the maps in (\ref{eq-def-distribution}) satisfy
\begin{equation} \label{eq-Xtilde-X}
\widetilde{X}\left(\gamma,\mathbf{t},\mathbf{s}\right)=\varepsilon\left(X\left(\gamma,\mathbf{t},\mathbf{s}\right), \tau\left(\gamma,\mathbf{t},\mathbf{s}\right) \right)
\end{equation}
and the curve $\widetilde{X}_{\left(\gamma,\mathbf{t}\right)}$ describes the submanifold $\widetilde{S}\left(\gamma\left(\mathbf{t}\right)\right) \subset \widetilde{\mathcal{N}}$, in fact an orbit of the distribution $\mathcal{D}^{\sim}$. 
This implies that $\mathcal{D}^{\sim}$ is generated by the tangent vectors $\frac{\partial \widetilde{X}}{\partial \mathbf{s}}\left(\gamma,\mathbf{t},0\right) \in T\widetilde{\mathcal{N}}$, but we have
\begin{equation} \label{eq-partial-Xtilde-X}
\frac{\partial \widetilde{X}}{\partial \mathbf{s}}\left(\gamma,\mathbf{t},0\right)=\left(d\varepsilon\right)_{\left(\gamma,\mathbf{t}\right)}\left( \frac{\partial X}{\partial \mathbf{s}}\left(\gamma,\mathbf{t},0\right) , \frac{\partial \tau}{\partial \mathbf{s}}\left(\gamma,\mathbf{t},0\right) \right) .
\end{equation} 

\begin{lemma}\label{lema-f-borde}
$\displaystyle{ \frac{\partial \tau}{\partial s}\left(\gamma,\mathbf{t},0\right)=0  }$ for all $\left(\gamma,\mathbf{t}\right) \in \mathcal{N}^{\epsilon}_{U}\times\left(-1,1\right)$.
\end{lemma}

\begin{proof}
If we consider the map
\[
f\left(\gamma,\mathbf{t},\mathbf{s},\tau\right)=\gamma_{\left(\mathbf{t},\mathbf{s}\right)}\left(\tau\right)
\]
then we have that
\[
\frac{\partial f}{\partial \mathbf{s}}\left(\gamma,\mathbf{t},0,\mathbf{t}\right) = J_{\left(\gamma,\mathbf{t},0\right)}\left(\mathbf{t}\right) = 0
\]
since it is the value of the Jacobi field $J_{\left(\gamma,\mathbf{t},0\right)}\in \widetilde{X}\left(\gamma,\mathbf{t},0\right)$ along $\gamma$ at the point $\gamma\left(\mathbf{t}\right)$, and moreover
\[
\frac{\partial f}{\partial \tau}\left(\gamma,\mathbf{t},0,\mathbf{t}\right) = \gamma'\left(\mathbf{t}\right)  .
\]

Now, defining
\[
\Psi\left(\gamma,\mathbf{t},\mathbf{s}\right)=f\left(\gamma,\mathbf{t},\mathbf{s},\tau\left(\gamma,\mathbf{t},\mathbf{s}\right)\right)
\]
then we have that the equation (\ref{eq-gamma-tau}) becomes
\[
\Psi\left(\gamma,\mathbf{t},\mathbf{s}\right)=\gamma\left(\mathbf{t}\right)
\]
and hence, since $\tau\left(\gamma,\mathbf{t},0\right)=\mathbf{t}$ 
\begin{align*}
\frac{\partial \Psi}{\partial \mathbf{s}}\left(\gamma,\mathbf{t},0\right)=0 &\Rightarrow \frac{\partial f}{\partial \mathbf{s}}\left(\gamma,\mathbf{t},0,\mathbf{t}\right) + \frac{\partial f}{\partial \tau}\left(\gamma,\mathbf{t},0,\mathbf{t}\right)\cdot \frac{\partial \tau}{\partial \mathbf{s}}\left(\gamma,\mathbf{t},0\right) = 0 \Rightarrow \\
& \Rightarrow J_{\left(\gamma,\mathbf{t},0\right)}\left(\mathbf{t}\right) + \gamma'\left(\mathbf{t}\right) \cdot \frac{\partial \tau}{\partial \mathbf{s}}\left(\gamma,\mathbf{t},0\right) = 0 \Rightarrow \\
& \Rightarrow 0 + \gamma'\left(\mathbf{t}\right) \cdot \frac{\partial \tau}{\partial \mathbf{s}}\left(\gamma,\mathbf{t},0\right) = 0 \Rightarrow \\
& \Rightarrow \frac{\partial \tau}{\partial \mathbf{s}}\left(\gamma,\mathbf{t},0\right)= 0
\end{align*}
as we want to prove.
  \end{proof}

After all these preparations we are ready to prove Thm. \ref{main_thm}.

\begin{proof} (Thm.  \ref{main_thm}, Main Theorem)
Because of Lemma \ref{lema-f-borde} and equations (\ref{eq-Xtilde-X}) and (\ref{eq-partial-Xtilde-X}), we have that the distribution $\mathcal{D}^{\sim}$ in $\widetilde{\mathcal{N}}$ can be defined at any $\widetilde{\gamma}\left(\mathbf{t}\right)\in\widetilde{\mathcal{N}}$ by 
\begin{equation}\label{eq-distrib-D}
\mathcal{D}^{\sim}_{\widetilde{\gamma}\left(\mathbf{t}\right)}= \mathrm{span}\left\{ \frac{\partial \widetilde{X}}{\partial \mathbf{s}}\left(\gamma,\mathbf{t},0\right) \right\} = \mathrm{span}\left\{ \left(d\varepsilon\right)_{\left(\gamma,\mathbf{t}\right)}\left( \frac{\partial X}{\partial \mathbf{s}}\left(\gamma,\mathbf{t},0\right) , 0 \right) \right\} .
\end{equation}

On the other hand, notice that  $\Gamma\left(s\right)$ is a integral curve of $\oplus:\mathcal{N}\rightarrow \mathbb{P}\left(\mathcal{H}\right)$ if $\Gamma'\left(s\right)\in \oplus_{\Gamma\left(s\right)}$.
So, the curve $\widetilde{\Gamma}\left(s\right)=\varepsilon\left(\Gamma\left(s\right),1\right)$ is a leaf of the distribution $\partial^{+}\mathcal{D}^{\sim}$ if $\Gamma'\left(s\right)\in \oplus_{\Gamma\left(s\right)}$, that is  
\begin{equation*}
\widetilde{\Gamma}'\left(s\right)=\left(d\varepsilon\right)_{\left(\Gamma\left(s\right),1\right)}\left( \Gamma'\left(s\right) , 0 \right)\in \partial^{+}\mathcal{D}^{\sim} \Longleftrightarrow  \Gamma'\left(s\right)\in \oplus_{\Gamma\left(s\right)}  
\end{equation*}
and therefore we have 
\begin{equation}\label{eq-distrib-Dborde}
\partial^{+}\mathcal{D}^{\sim}_{\widetilde{\gamma}\left(\mathbf{t}\right)} = \mathrm{span}\left\{ \left(d\varepsilon\right)_{\left(\gamma,1\right)}\left( \langle J\rangle , 0 \right) \right\} \text{ where } \langle J\rangle\in \oplus_{\Gamma\left(s\right)}  .
\end{equation}

In order to find a vector field in $\overline{\widetilde{\mathcal{N}}}$ defining the distribution $\overline{\mathcal{D}^{\sim}}=\mathcal{D}^{\sim} \cup \partial^{+}\mathcal{D}^{\sim}$, we can take a non--zero differentiable local section $\omega:\widetilde{\mathcal{U}}\subset\mathbb{P}\left(\mathcal{H}\right)\rightarrow \mathcal{H}$ at $\widetilde{\gamma}_0\left(1\right)\in\partial^{+}\widetilde{\mathcal{N}}$ by choosing representatives $\langle J_{\left(\gamma,\mathbf{t}\right)}\rangle\in \mathcal{H}_{\gamma}$ such that $J_{\left(\gamma,\mathbf{t}\right)}\left(0\right)\in T_{\gamma\left(0\right)}C$ such that $\mathbf{g}\left(J_{\left(\gamma,\mathbf{t}\right)}\left(0\right),J_{\left(\gamma,\mathbf{t}\right)}\left(0\right)\right)=1$.  

Since $\dim\left(T_{\gamma}S\left(\gamma\left(\mathbf{t}\right)\right)\right)=1$, we can choose two different representatives, selecting the one such that $J_{\left(\gamma,\mathbf{t}\right)}\left(0\right)= c'_{\left(\gamma,\mathbf{t}\right)}\left(0\right) $.
It is important to notice that these conditions determine the section $\omega$ without any condition on $J'_{\left(\gamma,\mathbf{t}\right)}\left(0\right)$, indeed, if $Y_1 , Y_2 \in \widetilde{\gamma}\left(\mathbf{t}\right)$ then $Y_1\left(\mathbf{t}\right)= Y_2\left(\mathbf{t}\right)=0~\left(\mathrm{mod}~\gamma'\right)$, and if moreover $Y_1$ and $Y_2$ are such that $Y_1\left(0\right)= Y_2\left(0\right)\left(\mathrm{mod}~\gamma'\right)$, hence the Jacobi field $K=Y_1-Y_2$ verifies $K\left(0\right)=0~\left(\mathrm{mod}~\gamma'\right)$ and $K\left(\mathbf{t}\right)=0~\left(\mathrm{mod}~\gamma'\right)$, and since $M$ is light non--conjugate, therefore $K=0~\left(\mathrm{mod}~\gamma'\right)$.

It is possible to assume, without any lack of generality, that $\widetilde{\mathcal{N}}^{\left(\epsilon,\delta\right)}_{U}=\varepsilon\left(\mathcal{N}^{\epsilon}_{U}\times \left(1-\delta, 1+\delta\right)\right)\subset \widetilde{\mathcal{U}}$ for some small enough $\delta>0$.   Indeed, by construction of $X\left(\gamma,\mathbf{t},\mathbf{s}\right)$, we have that, since the curve has been parametrized by arc length, then $J_{\left(\gamma,\mathbf{t}\right)}\left(0\right)= c'_{\left(\gamma,\mathbf{t}\right)}\left(0\right) $, and hence the section verifies  
\begin{equation}\label{omegaX}
\omega\left(\widetilde{\gamma}\left(\mathbf{t}\right)\right)=\langle J_{\left(\gamma,\mathbf{t}\right)}\rangle =\frac{\partial X}{\partial \mathbf{s}}\left(\gamma,\mathbf{t},0\right) 
\end{equation}
for all $\left(\gamma,\mathbf{t}\right)\in \mathcal{N}_{U}\times \left(1-\delta, 1\right)$.

Using the previous constructions we can define $\overline{\widetilde{\mathcal{N}}}_{U}=\varepsilon\left(\mathcal{N}^{\epsilon}_{U}\times \left(1-\delta, 1\right]\right)$ and the map  (see Fig. \ref{diapositiva2} for a graphical representation of the map $ \overline{\Phi}$):
\[
\begin{tabular}{rccl}
$\overline{\Phi}:$ & $\overline{\widetilde{\mathcal{N}}}_{U}\subset \overline{\widetilde{\mathcal{N}}}$ & $\rightarrow$ & $T_{\widetilde{\gamma}\left(\mathbf{t}\right)}\mathbb{P}\left(\mathcal{H}\right)$ \\
 & $\widetilde{\gamma}\left(\mathbf{t}\right)$ & $\mapsto$ & $\left(d\varepsilon\right)_{(\gamma,\mathbf{t})}\left( \omega\left(\widetilde{\gamma}\left(\mathbf{t}\right)\right) , 0\right)$
\end{tabular}
\]
which is clearly differentiable by composition of differentiable maps.

Now, let us see that $\overline{\Phi}$ defines $\overline{\mathcal{D}^{\sim}}$. 
Then, by equations (\ref{eq-distrib-D}) and (\ref{omegaX}), we have that 
\[
\mathcal{D}^{\sim}_{\widetilde{\gamma}\left(\mathbf{t}\right)}=\mathrm{span}\left\{ \overline{\Phi}\left(\widetilde{\gamma}\left(\mathbf{t}\right)\right) \right\}  
\]
for all  $\left(\gamma,\mathbf{t}\right)\in \mathcal{N}_{U}\times \left(1-\delta, 1\right)$.
Moreover, since $\omega$ is a non--zero local section and recalling that $\widetilde{\gamma}\left(\tau\right)=T_{\gamma}S\left(\gamma\left(\tau\right)\right)$ and $\oplus_{\gamma}=\lim_{\mathbf{t}\mapsto 1}\widetilde{\gamma}\left(\mathbf{t}\right)$, then we have that, for $\mathbf{t}=1$ (see Fig. \ref{diapositiva2}),
\[
\omega\left(\widetilde{\gamma}\left(1\right)\right)\in \oplus_{\gamma}
\] 
whence, using equation (\ref{eq-distrib-Dborde}), we obtain
\[
\partial^{+}\mathcal{D}^{\sim}_{\widetilde{\gamma}\left(1\right)}=\mathrm{span}\left\{ \overline{\Phi}\left(\widetilde{\gamma}\left(1\right)\right) \right\}  
\]
for $\gamma\in \mathcal{N}_{U}$.  Clearly, an analogous construction can be done for $\partial^-\mathcal{D}^\sim$.
So, we have 
\[
\overline{\mathcal{D}^{\sim}}=\mathrm{span}\left\{ \overline{\Phi}\left(\widetilde{\gamma}\left(\mathbf{t}\right)\right) : \mathbf{t}\in [-1,1]  \right\}
\]
and the distribution  $\overline{\mathcal{D}^{\sim}}$ is a differentiable extension of $\mathcal{D}^{\sim}$.
 \end{proof}

\begin{figure}[h]
  \centering
    \includegraphics[scale=1]{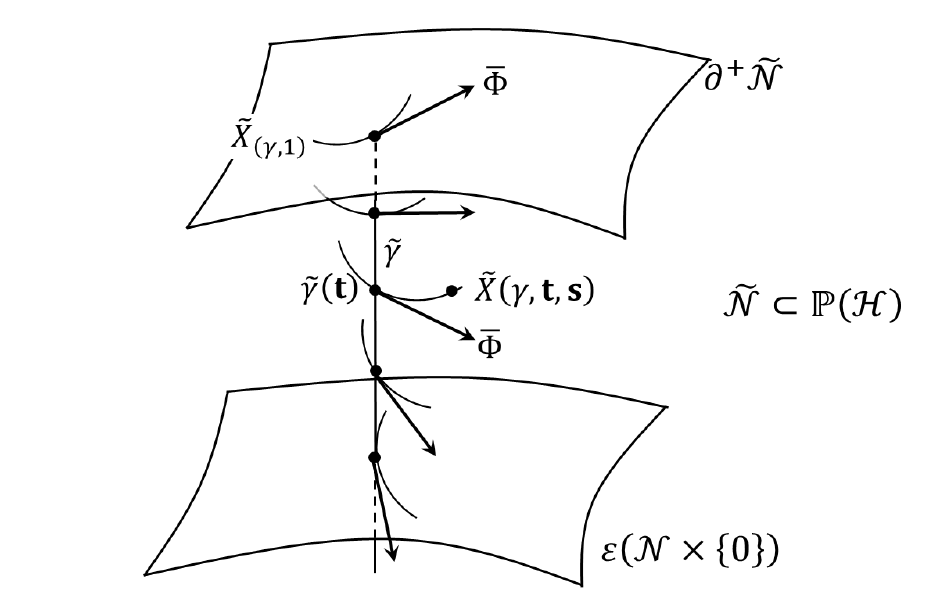}
  \caption{The vector field $\overline{\Phi}$ in $\overline{\widetilde{\mathcal{N}}}$.}
  \label{diapositiva2}
\end{figure}

\begin{corollary}\label{Corollary-ext} If the orbits of the regular distribution $\oplus$ are compact, then
the quotient $\overline{\widetilde{\mathcal{N}}} / \overline{\mathcal{D}^{\sim}}$ is a differentiable manifold with boundary  $ \partial^{+}\widetilde{\mathcal{N}} / \partial^{+}\mathcal{D}^{\sim}  \cup \partial^{-}\widetilde{\mathcal{N}} / \partial^{-}\mathcal{D}^{\sim}$.
Moreover, if $\overline{M} = M \cup \partial M $ where 
\[
\partial M =\partial^{+} M \cup \partial^{-}M = \partial^{+}\widetilde{\mathcal{N}} / \partial^{+}\mathcal{D}^{\sim}  \cup \partial^{-}\widetilde{\mathcal{N}} / \partial^{-}\mathcal{D}^{\sim} 
\]
then the diffeomorphism $\widetilde{S}:M\rightarrow \widetilde{\mathcal{N}} / \mathcal{D}^{\sim}$ of equation (\ref{diffeo-S}) can be extended by 
\[
\widetilde{S}:\overline{M}\rightarrow \overline{\widetilde{\mathcal{N}}} / \overline{\mathcal{D}^{\sim}}
\]
such that $\left.\widetilde{S}\right|_{\partial M}$ is the identity map, inducing in $\partial M$ a differentiable structure such that the extension $\widetilde{S}$ is a diffeomorphism.
\end{corollary} 

\begin{proof}
The main Theorem, Thm. \ref{main_thm}, shows that $\overline{\mathcal{D}^{\sim}}$ is a differentiable distribution. 
Now, observe that $\partial^{\pm}\mathcal{D}^{\sim}$ are regular distribution by hypothesis, and $\mathcal{D}^{\sim}$ is also a regular distribution because its orbits are compact, then because the orbits of $\partial^{\pm}\mathcal{D}^{\sim}$ are assumed to be compact too, then $\overline{\mathcal{D}^{\sim}}$ is also a regular distribution. 
Then, trivially, the quotient $\overline{\widetilde{\mathcal{N}}} / \overline{\mathcal{D}^{\sim}}$ is a differentiable manifold.  
Because $\partial^{\pm}\widetilde{\mathcal{N}}$ is the boundary of $\overline{\widetilde{\mathcal{N}}}$, then $ \partial^{\pm}\widetilde{\mathcal{N}} / \partial^{\pm}\mathcal{D}^{\sim}$ is the boundary of $\overline{\widetilde{\mathcal{N}}} / \overline{\mathcal{D}^{\sim}}$.

Moreover, since $\widetilde{S}$ restricted to $M$ is a diffeomorphism and $\overline{\widetilde{\mathcal{N}}} / \overline{\mathcal{D}^{\sim}}$ is a differentiable manifold, then there exists a differentiable structure in $\overline{M}$, compatible with the one in $M$, such that the extension $\widetilde{S}$ is a diffeomorphism. 
  \end{proof}

\begin{remark}\label{M3atob} Notice that the compactness assumption on the orbits of the boundary distributions $\oplus$ (respect. $\ominus$) is a natural one as they represent the skies at infinity.    Actually, this is exactly the situation that will happen if the spacetime $M$ would possess a compact Cauchy surface $C$ (as in the FRW cosmological models).  Notice that in such case the space of light rays $\mathcal{N}$ will be isomorphic to $C \times \mathbb{S}^1$, hence compact.  Then if $M$ is a L-spacetime, because of Thm. \ref{main_thm}, the total distribution $\overline{\mathcal{D}^\sim}$ will be regular, hence their leaves will be closed, but because the future component of the boundary $\partial^+\widetilde{\mathcal{N}}$ (respec., the past component) is diffeomorphic to $\mathcal{N}$, Prop. \ref{prop-Low-boundary}, then the leaves of the total distribution at the boundary will be compact and the conclusion of Cor. \ref{Corollary-ext} will hold.
\end{remark}

The extension $\overline{M}$ of $M$ of Corollary \ref{Corollary-ext} will be called \emph{the canonical extension of $\left(M,\mathcal{C}\right)$} and $\partial^{+} M$ and $\partial^{-} M$ are the boundaries toward the future and past of the light rays respectively.

\begin{definition}
We will say that a spacetime $M$ is a \emph{proper L-spacetime} if it is a L--spacetime such that the total smooth distribution $\overline{\mathcal{D}^\sim}$ is regular with Hausdorff space of leaves. 
\end{definition}

Notice that, the same argument used in the proof of Cor. \ref{Corollary-ext}  shows that if $M$ is a proper L-spacetime, the L-boundary $\partial M = \partial^{+} M \cup \partial^{-}M = \partial^{+}\widetilde{\mathcal{N}} / \partial^{+}\mathcal{D}^{\sim}  \cup \partial^{-}\widetilde{\mathcal{N}} / \partial^{-}\mathcal{D}^{\sim}$ defines a smooth boundary for the manifold $\overline{M} = M \cup \partial M$.  Moreover the assumption that the quotient space of the regular distribution $\overline{\mathcal{D}^\sim}$ is Hausdorff guarantess that $\overline{M}$ is Hausdorff. Then we obtain the following consequence.

\begin{corollary}\label{Corollary-ext-2}
If $M$ is a proper L--spacetime then the conclusion of Cor. \ref{Corollary-ext} holds, that is, the canonical extension $\overline{M}$ of $\left(M,\mathcal{C}\right)$ exists.
\end{corollary}

\begin{remark}\label{M3ainf}  The situation pointed out in corollary \ref{Corollary-ext-2}
is exactly what happens in the case of the Minkowski space $\mathbb{M}^3$.  It is not hard to see by repeating the computations in Sect. \ref{sec:block} for $\mathbb{M}^3(a,b)$ when $b\to +\infty$,  that the orbits of the future distribution $\oplus$ for $\mathbb{M}^3(a,+\infty)$ are straight lines in the Cauchy surface $C$, that is the future L-boundary $\partial^+\mathbb{M}^3(a,+\infty)$ is bidimensional and diffeomorphic to $\mathbb{R}\times \mathbb{S}^1$. 

In fact, using the same notations than in Sect. \ref{sec:block}, we get that given a null geodesic $\gamma = \gamma_{(x_0,y_0,\theta_0)}$:
\begin{eqnarray*}
\oplus_\gamma &=& \lim_{t \to +\infty} T_\gamma S(\gamma(t)) = \lim_{t \to +\infty} [t:t_0-t]_\gamma = \lim_{t \to +\infty} [1:t_0/t-1]_\gamma = [1:-1]_\gamma \\ &=& \mathrm{span}\left\{ \sin \theta_0 \left.\frac{\partial}{\partial x}\right|_\gamma  - \cos \theta_0 \left.\frac{\partial}{\partial y}\right|_\gamma\right\} \, .
\end{eqnarray*}
Therefore we can obtain the integral curve $c(s)  = (x(s),y(s),\theta(s))$ of $\oplus$ passing through $\gamma$, that defines the sky $S(\infty^+(\gamma))$ of $\gamma$ at $+\infty$, solving the initial value problem:
\begin{equation}\label{IVPainfty}
\frac{dx}{ds} = \sin \theta \, , \quad \frac{dy}{ds} = - \cos \theta \, , \quad \frac{d\theta}{ds} = 0 \, ,
\end{equation}
with $c(0) = (x_0,y_0,0)$.  Notice that the change in $d\theta/ds$ in Eq. (\ref{IVPainfty}) with respect to Eq. (\ref{IVPab}) is critical with respect to the analysis performed in Sect. \ref{sec:block}.   Actually, in the case of $\mathbb{M}^3(a,b)$, $d\theta/ds = 1$, see Eq. (\ref{IVPab}), which upon integration gave us circles which were precisely the skies of the points in the topological bounday of $\mathbb{M}^3(a,b)$ considered as a subset of $\mathbb{M}^3$.  However in the present situation, $d\theta/ds = 0$, and the integral curves of (\ref{IVPainfty}) are straight lines:
$$
c(s) = (x_0 + s \sin \theta_0, y_0 - s \cos \theta_0, \theta_0) \, ,
$$
that corresponds to the family of null geodesics with tangent vector $\mathbf{v} = (1,\cos \theta_0, \sin \theta_0)$ and initial value in the straight line in $C$, recall (\ref{gammaxyphi}), given by:
$$
t = 0 \, , \qquad \cos \theta_0 (x-x_0) + \sin \theta_0 (y-y_0) = 0 \, .
$$
Hence we conclude that $\mathbb{M}^3(a,+\infty)$ is a proper L-spacetime even though the orbits of the boundary distribution are not compact.

It is straightforward to check that $I^-(\gamma_1) = I^-(\gamma_2)$ for any two light rays $\gamma_1, \gamma_2 \in S(\infty^+(\gamma))$, therefore any light ray in $S(\infty^+(\gamma))$ defines the same TIP:
$$
I^-(\gamma) = \{ (t,x,y) \in \mathbb{M}^3(a,+\infty) \mid t < \cos \theta_0 (x-x_0) + \sin \theta_0 (y-y_0)  \} \, ,
$$
and the future L-boundary coincides with the future part of the c-boundary accessible by light rays.
We must point it out that there is a body of work explicating in full detail exactly
what the causal boundary is in many concrete situations that can be compared easily with the L-boundary using computations similar to those performed above (see for instance \cite{Fl07}, \cite{Ha04}).  Another example of a proper $L$-spacetime is provided for instance by the 3-dimensional de-Sitter spacetime as shown in \cite{Ba17}.
\end{remark}

\begin{remark}\label{remark-Xtilde-extension}
By Cor. \ref{Corollary-ext}, since $\widetilde{X}_{\left(\gamma,\mathbf{t}\right)}\left(s\right) = \widetilde{X}\left(\gamma,\mathbf{t},s\right)$ is an integral curve of $\overline{\Phi}$, then it is possible to extend smoothly the maps $\widetilde{X}$, $X$ and $\tau$ as
\begin{equation}\label{eq-def-distribution-extension}
\begin{tabular}{l}
$\widetilde{X}:\mathcal{N}_{U}\times\left(-1,1\right]\times\left(-\epsilon,\epsilon\right)\longrightarrow \overline{\widetilde{\mathcal{N}}}$ \\
$X:\mathcal{N}_{U}\times\left(-1,1\right]\times\left(-\epsilon,\epsilon\right)\longrightarrow \mathcal{N}$ \\
$\tau:\mathcal{N}_{U}\times\left(-1,1\right]\times\left(-\epsilon,\epsilon\right)\longrightarrow \left(-1,1\right]$
\end{tabular}
\end{equation}
where the non--zero vector field $\overline{\Phi}$ can be written by 
\begin{equation}\label{eq-def-Phi-field}
\overline{\Phi}\left(\widetilde{\gamma}\left(\mathbf{t}\right)\right)=\left( d\varepsilon \right)_{ \left(\gamma , \mathbf{t}\right) }\left( \frac{\partial X}{\partial \mathbf{s}}\left(\gamma,\mathbf{t}, 0\right), 0 \right) \in T_{\widetilde{\gamma}\left(\mathbf{t}\right)}\varepsilon\left(\mathcal{N}_U \times \{ \mathbf{t} \} \right)
\end{equation}
for $\left(\gamma,\mathbf{t}\right)\in \mathcal{N}_U \times \left(-1,1\right]$.
\end{remark}

\begin{proposition}\label{prop-gamma-transversal} Let $M$ be a proper L-spacetime and
let $\overline{M}$ be the canonical extension of $M$.
For every light ray $\gamma\in \mathcal{N}$ the extension $\overline{\gamma}\subset \overline{M}$ parametrized by a projective parameter $\overline{\gamma}:\left[-1,1\right]\rightarrow \overline{M}$ is a regular curve and transversal to $\partial M$.
Moreover, denoting $\infty^{\pm}_{\gamma}=\lim_{\mathbf{t}\mapsto \pm 1} \overline{\gamma}\left(\mathbf{t}\right)$, the maps 
$\infty^{\pm}:\mathcal{N}\rightarrow \partial^{\pm} M$ defined by $\infty^{\pm}\left(\gamma\right)=\infty^{\pm}_{\gamma}$ are surjective submersions. 
\end{proposition}

\begin{proof}
We can consider the following diagram 
\begin{equation}\label{diagram-N-tilde}
\begin{tikzpicture}[every node/.style={midway}]
\matrix[column sep={7em,between origins},
        row sep={2em}] at (0,0)
{ \node(PN1)   {$\overline{\widetilde{\mathcal{N}}}$}  ; & \node(N) {$\overline{\widetilde{\mathcal{N}}} / \overline{\mathcal{D}^{\sim}}$}; \\
; &  \node(PN2) {$\overline{M}$};                   \\};
\draw[->] (PN1) -- (N) node[anchor=south]  {$\widetilde{\pi}$};
\draw[->] (N) -- (PN2) node[anchor=west]  {$\widetilde{S}^{-1}$};
\draw[->] (PN1)   -- (PN2) node[anchor=north east] {$\mathbf{\rho}$};
\end{tikzpicture}
\end{equation}
where the quotient map $\widetilde{\pi}$ is a submersion, $\widetilde{S}^{-1}$ is a diffeomorphism and hence $\rho$ is a submersion.
Then we have that $\overline{\gamma}\left(\mathbf{t}\right)=\rho\left(\widetilde{\gamma}\left(\mathbf{t}\right)\right)$ for all $\mathbf{t}\in\left[-1,1\right]$. 
Trivially, since $\widetilde{\gamma}$ is differentiable then $\overline{\gamma}$ is so also.
Now, since the tangent space to the orbit of the distribution $\overline{\mathcal{D}^{\sim}}$ is defined by the non--zero vector field $\overline{\Phi}$, by the equation (\ref{eq-def-Phi-field}) we have that $\overline{\Phi}\left(\widetilde{\gamma}\left(\mathbf{t}\right)\right) \in T_{\widetilde{\gamma}\left(\mathbf{t}\right)}\varepsilon\left(\mathcal{N}_U \times \{ \mathbf{t} \} \right)$ and since, by equation (\ref{extension-gamma-prima-tilde}), we have $0\neq\widetilde{\gamma}'\left(s\right)\notin T_{\widetilde{\gamma}\left(\mathbf{t}\right)}\varepsilon\left(\mathcal{N}_U \times \{ \mathbf{t} \} \right)$, then 
$\widetilde{\gamma}$ is transversal to the orbits of the distribution $\overline{\mathcal{D}^{\sim}}$ in $\overline{\widetilde{\mathcal{N}}}$, then the regularity of $\overline{\gamma}$ follows. 
On the other hand, in particular for $\mathbf{t}=1$, by equation (\ref{extension-gamma-prima-tilde}), we have $\widetilde{\gamma}'\left(1\right)=\left( \frac{\partial}{\partial \mathbf{t}} \right)_{\widetilde{\gamma}\left(1\right)} \notin T_{ \widetilde{\gamma}\left(1\right)}\partial^{+}\widetilde{\mathcal{N}}$ and due to $\partial^{+}M \simeq \partial^{+}\widetilde{\mathcal{N}} / \partial^{+}\mathcal{D}^{\sim} $ and the regularity of $\overline{\gamma}$ then $0\neq \overline{\gamma}'\left(1\right)\notin T_{\overline{\gamma}\left(1\right)}\partial^{+}M$ and $\overline{\gamma}$ is transversal to $\partial^{+} M$. 

Moreover, the restriction of the diagram (\ref{diagram-N-tilde}) to $\partial^{+} \widetilde{\mathcal{N}}$ gives 
\begin{equation}\label{diagram-boundary-N-tilde}
\begin{tikzpicture}[every node/.style={midway}]
\matrix[column sep={7em,between origins},
        row sep={2em}] at (0,0)
{ \node(PN1)   {$\partial^{+} \widetilde{\mathcal{N}}$}  ; & \node(N) {$\partial^{+} \widetilde{\mathcal{N}} / \partial^{+}{\mathcal{D}^{\sim}}$}; \\
; &  \node(PN2) {$\partial^{+}M$};                   \\};
\draw[->] (PN1) -- (N) node[anchor=south ]  {$\widetilde{\pi}$};
\draw[->] (N) -- (PN2) node[anchor=west]  {$\widetilde{S}^{-1}$};
\draw[->] (PN1)   -- (PN2) node[anchor=north east] {$\mathbf{\rho}$};
\end{tikzpicture}
\end{equation}
and we have the diffeomorphisms,
\[
\begin{tabular}{rcccl}
$\mathcal{N}$ & $\overset{i}{\longrightarrow}$ & $\mathcal{N}\times \{1\}$ & $\overset{\left.\varepsilon\right|_{\mathcal{N}\times \{1\}}}{\longrightarrow}$ & $\partial^{+}\widetilde{\mathcal{N}}$ \\
$\gamma$ & $\mapsto$ & $\left(\gamma,1\right)$  & $\mapsto$  & $\widetilde{\gamma}\left(1\right)$ 
\end{tabular}
\]
such that we can write 
\[
\infty^{+}=\left.\rho\right|_{\partial^{+}\widetilde{\mathcal{N}}} \circ \left.\varepsilon\right|_{\mathcal{N}\times \{1\}} \circ i \, .
\]
Since $\left.\rho\right|_{\partial^{+}\widetilde{\mathcal{N}}}$ is still a surjective submersion onto $\partial^{+}M$, therefore $\infty^{+}$ is a surjective submersion. 
The case $\mathbf{t}=-1$ and $\infty^{-}$ is analogous.
  \end{proof}


\section{L--extensions}\label{sec:lextensions}

In this section we will characterise the differentiable structure of the canonical extension $\overline{M}$ defined in Corollary \ref{Corollary-ext}.
We will need to distinguish several types of parametrisations of light rays, so we will fix some nomenclature first.  In what follows and when referring to spacetimes $M$ in any dimension, we will be assuming that $M$ is such that its space of light rays $\mathcal{N}$ is a smooth manifold, that is, we assume for instance that $M$ is time-oriented, strongly causal, null-pseudo convex and sky-separating. 

\begin{definition}\label{def-parametrizations}
Let $\gamma:\left(a,b\right)\rightarrow M$ be an inextensible parametrization of a light ray $\gamma\in \mathcal{N}$ such that $\gamma\subset M$ is future--directed, that is, $\gamma\left(s_1\right)$ is in the causal past of $\gamma\left(s_2\right)$ for all $s_1 < s_2$. 
This parametrization is said to be
\begin{enumerate}
\item \emph{continuous} if $\gamma:\left(a,b\right)\rightarrow M$ is a continuous map,
\item \emph{regular} if $\gamma:\left(a,b\right)\rightarrow M$ is a differentiable map and $\gamma'\left(s\right)\in \mathbb{N}$ is a future--directed lightlike vector for all $s\in \left(a,b\right)$,
\item \emph{projective} if $\gamma:\left(a,b\right)\rightarrow M$ is a regular parametrization and $\widetilde{\gamma}\left(s\right)=\sigma\left(\left[\gamma'\left(s\right)\right]\right)\in \mathbb{P}\left(\mathcal{H}_{\gamma}\right)$ defines a projectivity, that is, the parameter $s$ is a function of the form (\ref{projectivity}), in the fibre $\mathbb{P}\left(\mathcal{H}_{\gamma}\right)$, and 
\item \emph{admissible} if there exists a diffeomorphism $h:\left(c,d\right]\rightarrow \left(a,b\right]$ such that $h'\left(t\right)>0$ for all $t\in \left(c,d\right]$ and $\gamma\circ h:\left(c,d\right)\rightarrow M$ is a projective parametrization.
\end{enumerate}
\end{definition} 

It can be trivially observed that any projective parametrization of $\gamma\in \mathcal{N}$ is admissible.

\begin{remark}\label{remark-parameters}
It is important to notice that for every regular parametrization of $\gamma\in\mathcal{N}$ can be reparametrized diffeomorphically to the canonical projective parameter, but this does not imply that it is an admissible parameter. 
Indeed, if $\gamma:\left(a,b\right)\rightarrow M$ is a regular parametrization, we can send it to $\widetilde{\mathcal{N}}\subset \mathbb{P}\left(\mathcal{H}_{\gamma}\right)$ via the following composition
\[
\begin{tabular}{ccccccc}
$\left(a,b\right)$ & $\overset{\gamma'}{\rightarrow}$ & $\mathbb{N}$ & $\overset{\left[~\right]}{\rightarrow}$ & $\mathbb{PN}$ & $\overset{\sigma}{\rightarrow}$ & $\widetilde{\mathcal{N}}$ \\
$s$ & $\mapsto$ & $\gamma'\left(s\right)$ & $\mapsto$ & $\left[\gamma'\left(s\right)\right]$ & $\mapsto$ & $\widetilde{\gamma}\left(s\right)$
\end{tabular}
\] 
Since, $\gamma\left(s\right)\in M$ is a regular curve, then $\gamma'\left(s\right)\in \mathbb{N}$ and $\left[\gamma'\left(s\right)\right]\in \mathbb{PN}$ are also regular curves. 
Now, the diffeomorphism $\sigma:\mathbb{PN}\rightarrow\widetilde{\mathcal{N}}$ maps regular curves into regular curves, so $\widetilde{\gamma}\left(s\right)\in \widetilde{\mathcal{N}}$ is a regular parametrization of the submanifold $\widetilde{\gamma}\subset \widetilde{\mathcal{N}}$. 
Since the canonical projective parametrization $\widetilde{\gamma}\left(\mathbf{t}\right)$ is another regular parametrization, then there exist a differentiable change of parameter $h:\left(-1,1\right)\rightarrow \left(a,b\right)$ such that $h'\left(t\right)>0$ for all $t\in \left(-1,1\right)$. 
For $s\in \left(a,b\right)$ to be an admissible parameter, condition $h'\left(1\right)>0$ remains to be satisfied.
\end{remark}

\begin{definition}\label{def-L-extension}
We define a \emph{future L--extension} of a conformal manifold $\left(M,\mathcal{C}\right)$ as a Hausdorff smooth manifold $\overline{M}=M \cup \partial^{+} M$ where $\partial^{+} M = \overline{M}-M$ is a closed hypersurface of $\overline{M}$ called the \emph{future L--boundary} such that the following properties are satisfied:
\begin{enumerate}
\item \label{L-ext-cond-0} If $\gamma:\left(a,b\right)\rightarrow M$ is a continuous parametrization of $\gamma\in \mathcal{N}$, then $\lim_{s\mapsto b^{-}}\gamma\left(s\right)=\infty^{+}_{\gamma}\in \partial^{+}M$. 
\item \label{L-ext-cond-2} The map $\infty^{+}:\mathcal{N}\rightarrow \partial^{+} M$ defined by $\infty^{+}\left(\gamma\right)=\infty^{+}_{\gamma}$ is a surjective submersion.
\item \label{L-ext-cond-1} For every $\gamma_0\in \mathcal{N}$ there exists a neighbourhood $\mathcal{U}\subset \mathcal{N}$ and a differentiable map $\Psi_{\mathcal{U}}:\mathcal{U}\times\left(a,b\right]\rightarrow\overline{M}$, where $\gamma\left(s\right)=\Psi_{\mathcal{U}}\left(\gamma, s\right)$ is an admissible parametrization of $\gamma\in \mathcal{U}$ for $s\in\left(a,b\right)$ and such that $\frac{\partial \Psi_{\mathcal{U}}}{\partial s}\left(\gamma,b\right)\notin T_{\infty^{+}\left(\gamma\right)}\partial^{+}M$.
\end{enumerate}
If there exists any L--extension of $\left(M,\mathcal{C}\right)$, then it is said that $\left(M,\mathcal{C}\right)$ is \emph{L--extensible}.
In an analogous and obvious way, we can define a \emph{past L--extension} $\overline{M}=M \cup \partial^{-} M$.
\end{definition}

Observe that since the map $\infty^{+}:\mathcal{N}\rightarrow \partial^{+}M$ of a L--extension of $\left(M,\mathcal{C}\right)$ is a surjective submersion then every of its inverse images 
\[
S\left(p\right) = \left(\infty^{+}\right)^{-1}\left(p\right)= \{ \gamma\in \mathcal{N} : p = \infty^{+}\left(\gamma\right) \} \subset \mathcal{N}  
\]
defines a leaf of a regular distribution $\boxplus:\mathcal{N}\rightarrow \mathbb{P}\left(T\mathcal{N}\right)$ given by $\boxplus\left(\gamma\right)= T_{\gamma}S\left(\infty^{+}\left(\gamma\right)\right)$, and the map 
\[
\begin{tabular}{rccl}
$S:$ & $ \partial^{+} M$ & $\rightarrow$ & $\mathcal{N}/\boxplus$  \\
 & $p$ & $\mapsto$ & $S\left(p\right)$   
\end{tabular}
\]
is a diffeomorphism.

\begin{lemma}\label{Lemma-canon-prject-param}
Let $\overline{M}$ be a L--extension of $M$ and $\Psi_{\mathcal{U}}:\mathcal{U}\times\left(a,b\right]\rightarrow\overline{M}$ the differentiable map of condition \ref{L-ext-cond-1} of definition of L--extensions, then there exist a differentiable function $h:\mathcal{U}\times\left(-1,1\right]\rightarrow \left(a,b\right]$ such that the map $\overline{\Psi}_{\mathcal{U}}\left(\gamma,\mathbf{t}\right)=\Psi_{\mathcal{U}}\left(\gamma,h\left(\gamma,\mathbf{t}\right)\right)$ also satisfies the condition \ref{L-ext-cond-1} of definition \ref{def-L-extension} and where $\mathbf{t}\in\left(-1,1\right]$ is the canonical projective parameter.
\end{lemma}

\begin{proof}
First, notice that if $\gamma:\left(a,b\right)\rightarrow M$ is an admissible parametrization of the light ray $\gamma\in \mathcal{N}$, then there exists a diffeomorphism $h:\left(c,d\right]\rightarrow \left(a,b\right]$ such that $\frac{dh}{dt}>0$.
But since $t\in\left(c,d\right)$ is a projective parameter, then there are $A,B,C,D\in \mathbb{R}$ with $AD-BC>0$ such that $t=\frac{A\mathbf{t}+B}{C\mathbf{t}+D}$ is a projective parameter diffeomorphism between the canonical projective parameter $\mathbf{t}\in \left(-1,1\right]$ and $t\in\left(c,d\right]$ verifying $\frac{dt}{d\mathbf{t}}>0$ for all $\mathbf{t}\in\left(-1,1\right]$.
Therefore, every admissible parameter is diffeomorphic to the canonical projective $\mathbf{t}$ in the sense of the definition \ref{def-parametrizations}. 

Now, let us prove the existence of $\overline{\Psi}_{\mathcal{U}}$. 
It is clear that, for any $\left(\gamma,\mathbf{t}\right)\in \mathcal{U}\times \left(-1,1\right]$, there is a unique $s\in \left(a,b\right]$ such that the equation 
\[
\Psi_{\mathcal{U}}\left(\gamma,s\right) = \gamma\left(\mathbf{t}\right)
\]
is satisfied.
Then, there exist a function $h:\mathcal{U}\times\left(-1,1\right]\rightarrow \left(a,b\right]$ such that for any $\gamma\in \mathcal{U}$,  $s=h_{\gamma}\left(\mathbf{t}\right)=h\left(\gamma,\mathbf{t}\right)$ is a reparametrization of $\gamma$. 

Let us see that $h$ is differentiable. 
Given any $\left(\gamma,\mathbf{t}\right)\in \mathcal{U}\times \left(-1,1\right]$, consider a coordinate chart $\left(W,\varphi\right)$ at $\gamma\left(\mathbf{t}\right)\in \overline{M}$. 
We construct the map 
\[
F\left(\gamma,s,\mathbf{t}\right)=\varphi\left(\Psi_{\mathcal{U}}\left(\gamma,s\right)\right) - \varphi\left(\gamma\left(\mathbf{t}\right)\right)\in \mathbb{R}^{3}
\]
and since 
\[
\frac{\partial F}{\partial s}\left(\gamma,s,\mathbf{t}\right)= d\varphi_{\gamma\left(s\right)}\left(\frac{\partial \Psi_{\mathcal{U}}}{\partial s}\left(\gamma,s\right)\right) = d\varphi_{\gamma\left(s\right)}\left(\gamma'\left(s\right)\right) \neq 0
\]
due to $\gamma'\left(s\right)\neq 0$ because $s\in \left(a,b\right]$ is an admissible parameter and $\varphi$ is a diffeomorphism, then the Implicit function Theorem assures that $h$ is differentiable in a neighbourhood of $\left(\gamma,\mathbf{t}\right)$, but this is true for all $\left(\gamma,\mathbf{t}\right)\in \mathcal{U}\times \left(-1,1\right]$, therefore $h$ is differentiable.
Moreover, $\frac{\partial h}{\partial \mathbf{t}}\left(\gamma,1\right)=\frac{d h_{\gamma}}{d \mathbf{t}}\left(1\right)>0$ according to definition of admissible parameter.

So, the map $\overline{\Psi}_{\mathcal{U}}\left(\gamma,\mathbf{t}\right)=\Psi_{\mathcal{U}}\left(\gamma,h\left(\gamma,\mathbf{t}\right)\right)$ is differentiable and the parameter $\mathbf{t}\in\left(-1,1\right]$ is admissible because it is projective. 
Finally, since 
\[
\frac{\partial \overline{\Psi}_{\mathcal{U}}}{\partial \mathbf{t}}\left(\gamma,1\right)=\frac{\partial \Psi_{\mathcal{U}}}{\partial s}\left(\gamma,b\right)\cdot \frac{\partial h}{\partial \mathbf{t}}\left(\gamma,1\right)= \frac{\partial \Psi_{\mathcal{U}}}{\partial s}\left(\gamma,b\right)\cdot \frac{d h_{\gamma}}{d \mathbf{t}}\left(1\right) \notin T_{\infty^{+}\left(\gamma\right)}\partial^{+}M
\]
as claimed.
  \end{proof}

\begin{remark}\label{remark-change-parameter}
Consider a neighbourhood $\mathcal{N}_{U}$ as the one of coordinate chart (\ref{carta-N}) and let us assume that a neighbourhood $\mathcal{V}\subset \mathcal{N}$ of condition \ref{L-ext-cond-1} in definition \ref{def-L-extension} is such that $\mathcal{V}\subset \mathcal{N}_{U}$.
By Lemma \ref{Lemma-canon-prject-param}, we can assume that the maps $\Psi_{\mathcal{V}}$ can be defined by the canonical projective parameter $\mathbf{t}$ by 
\[
\begin{tabular}{rrcl}
$\Psi_{\mathcal{V}}:$ & $\mathcal{V}\times\left(-1,1\right]$ & $\rightarrow$ & $\overline{M}$ \\
 & $\left(\gamma,\mathbf{t}\right)$ & $\mapsto$ & $\overline{\gamma}\left(\mathbf{t}\right)$  .
\end{tabular}
\]

Moreover, if $\{\mathcal{V}_{\alpha}\}_{\alpha\in I}$ is an open covering of $\mathcal{N}_{U}$ such that $\mathcal{V}_{\alpha}\subset \mathcal{N}_{U}$ for all $\alpha\in I$, since 
\[
\left.\Psi_{\mathcal{V}_{\alpha}}\right|_{\mathcal{V}_{\alpha}\cap\mathcal{V}_{\beta}\times\left(-1,1\right]} = \left.\Psi_{\mathcal{V}_{\beta}}\right|_{\mathcal{V}_{\alpha}\cap\mathcal{V}_{\beta}\times\left(-1,1\right]}
\]
then trivially, it is possible to define $\Psi_{\mathcal{N}_{U}}:\mathcal{N}_{U}\times\left(-1,1\right]\rightarrow\overline{M}$ extending all $\Psi_{\mathcal{V}_{\alpha}}$.
\end{remark}

Assuming definition \ref{def-L-extension}, the following Corollary follows automatically from the previous comments, Lemma \ref{Lemma-canon-prject-param}, and from the Prop. \ref{prop-gamma-transversal}, where $\Psi_{\mathcal{N}_{U}}=\rho\circ \varepsilon$ for the canonical extension.   Then, we can summarise most of the previous discussion stating the following Corollary that, in addition, justifies the name chosen in Sect. \ref{sec:lboundary} for the 3-dimensional manifolds satisfying the properties used in this paper.

\begin{corollary}\label{Corol-canonical-extension}
If $M$ is a 3-dimensional proper L--spacetime, the future (resp. past) canonical extension of $\left(M,\mathcal{C}\right)$ is a future (resp. past) L--extension.
\end{corollary} 

We will show in Lemma \ref{Lemma-equal-skies} that the sky of every point $\infty^{+}\left(\gamma\right)\in \partial ^{+}M$ is the same set in $\mathcal{N}$ for any future L--extension $\overline{M}$.

\begin{lemma}\label{Lemma-equal-skies}
Let $\overline{M}_1 = M \cup \partial^{+} M_1$ be the canonical future  L--extension and $\overline{M}_{2} = M \cup \partial^{+} M_2$ any other future L--extension of $\left(M,\mathcal{C}\right)$ then 
\[
S\left(\infty_{1}^{+}\left(\gamma\right)\right) = S\left(\infty_{2}^{+}\left(\gamma\right)\right)
\] 
for all $\gamma\in \mathcal{N}$, where $\infty_{i}^{+}:\mathcal{N}\rightarrow \partial^{+} M_{i}$ with $i=1,2$ are the surjective submersions of the definition \ref{def-L-extension}.
\end{lemma}

\begin{proof}
Given $\gamma\in \mathcal{N}$, consider the extended maps $\widetilde{X}$, $X$ and $\tau$ of equation (\ref{eq-def-distribution-extension}) defined in $\mathcal{N}_{U}\times\left(-1,1\right]\times\left(-\epsilon,\epsilon\right)$.
Recall that all these maps are differentiable and, by (\ref{eq-Xtilde-X}), $\widetilde{X}=\varepsilon\left(X,\tau\right)$. 
Since $\widetilde{X}\left(\gamma,1,s\right)\in \partial^{+}\widetilde{\mathcal{N}}$ for all $s\in \left(-\epsilon,\epsilon\right)$ then 
\begin{equation}\label{eq-1-equal-skies}
\tau\left(\gamma,1,s\right)=1 \quad \text{ for all } s\in \left(-\epsilon,\epsilon\right)  .
\end{equation}
Let us use the notation $\gamma_{\left(\mathbf{t},s\right)}=X\left(\gamma,\mathbf{t},s\right)\in \mathcal{N}$.
By equation (\ref{eq-gamma-tau}), we have that $\gamma_{\left(\mathbf{t},s\right)}\left(\tau\left(\gamma,\mathbf{t},s\right)\right)=\gamma\left(\mathbf{t}\right)$, hence
\begin{equation}\label{eq-2-equal-skies}
\gamma_{\left(\mathbf{t},s\right)}\in S\left(\gamma\left(\mathbf{t}\right)\right)  \quad \text{ for all } \left(\mathbf{t},s\right)\in \left(-1,1\right]\times\left(-\epsilon,\epsilon\right)  .
\end{equation}
By means of the diffeomorphism $\rho$ of diagram (\ref{diagram-N-tilde}), we have $\rho\left(\widetilde{X}\left(\gamma,\mathbf{t},s\right)\right)=\gamma_{\left(\mathbf{t},s\right)}\left(\tau\left(\gamma,\mathbf{t},s\right)\right)\in \overline{M}_1$, then by equation (\ref{eq-2-equal-skies})
\begin{equation}\label{eq-3-equal-skies}
\lim_{\mathbf{t}\mapsto 1}\gamma_{\left(\mathbf{t},s\right)}\left(\tau\left(\gamma,\mathbf{t},s\right)\right) = \rho\left(\widetilde{X}\left(\gamma,1,s\right)\right)=\overline{\gamma}\left(1\right)=\infty^{+}_{1}\left(\gamma\right) \in \partial^{+}M_1
\end{equation}
and moreover 
\begin{equation}\label{eq-4-equal-skies}
\lim_{\mathbf{t}\mapsto 1}\gamma_{\left(\mathbf{t},s\right)} = \lim_{\mathbf{t}\mapsto 1}X\left(\gamma,\mathbf{t},s\right)=X\left(\gamma,1,s\right)= \gamma_{\left(1,s\right)} \in S\left(\infty^{+}_{1}\left(\gamma\right)\right)\subset \mathcal{N}
\end{equation}
for all $s\in \left(-\epsilon,\epsilon\right)$.

Now, we want to show that $\lim_{\mathbf{t}\mapsto 1}\gamma_{\left(\mathbf{t},s\right)}\left(\tau\left(\gamma,\mathbf{t},s\right)\right)=\infty^{+}_{2}\left(\gamma\right)\in \partial^{+}M_2$. 
Since $\overline{M}_2$ is a L--extension (and according to remark \ref{remark-change-parameter}), then there is a differentiable map $\Psi_{\mathcal{N}_{U}}:\mathcal{N}_{U}\times \left(-1,1\right]\rightarrow \overline{M}_2$ such that $\overline{\gamma}_{\left(\mathbf{t},s\right)}\left(\tau\left(\gamma,\mathbf{t},s\right)\right)=\Psi_{\mathcal{N}_{U}}\left(\gamma_{\left(\mathbf{t},s\right)}, \tau\left(\gamma,\mathbf{t},s\right)\right)\in \overline{M}_2$.
But by equations (\ref{eq-1-equal-skies}) and (\ref{eq-4-equal-skies}), and since $\Psi_{\mathcal{N}_{U}}$ is continuous, then 
\[
\lim_{\mathbf{t}\mapsto 1}\overline{\gamma}_{\left(\mathbf{t},s\right)}\left(\tau\left(\gamma,\mathbf{t},s\right)\right) = \lim_{\mathbf{t}\mapsto 1}\Psi_{\mathcal{N}_{U}}\left(\gamma_{\left(\mathbf{t},s\right)}, \tau\left(\gamma,\mathbf{t},s\right)\right) = \Psi_{\mathcal{N}_{U}}\left(\gamma_{\left(1,s\right)}, 1\right)= \infty^{+}_{2}\left(\gamma_{\left(1,s\right)}\right)\in \partial^{+}M_2 
\]
for all $s\in\left(-\epsilon, \epsilon\right)$.
On the other hand 
\[
\lim_{\mathbf{t}\mapsto 1}\overline{\gamma}\left(\mathbf{t}\right) = \overline{\gamma}\left(1\right) = \infty^{+}_{2}\left(\gamma\right)\in \partial^{+}M_2 
\]
so, by equation (\ref{eq-gamma-tau}), we have $\infty^{+}_{2}\left(\gamma_{\left(1,s\right)}\right)=\infty^{+}_{2}\left(\gamma\right)$ for all $s\in\left(-\epsilon, \epsilon\right)$, hence $\gamma_{\left(1,s\right)}\in S\left(\infty^{+}_{2}\left(\gamma\right)\right)$ for all $s\in\left(-\epsilon, \epsilon\right)$. 
Because of equation (\ref{eq-4-equal-skies}), $\gamma_{\left(1,s\right)} \in S\left(\infty^{+}_{1}\left(\gamma\right)\right)$, hence the sky $S\left(\infty^{+}_{1}\left(\gamma\right)\right)$ coincides with $S\left(\infty^{+}_{2}\left(\gamma\right)\right)$ locally, then they must coincide globally. 
Therefore $S\left(\infty_{1}^{+}\left(\gamma\right)\right) = S\left(\infty_{2}^{+}\left(\gamma\right)\right)$.
  \end{proof}

Now, we introduce the Theorem characterizing all L--extensions.

\begin{theorem}\label{Theorem-diff-struct}  Let $M$ be a proper L-spacetime and $\overline{M}_1 = M \cup \partial^{+} M_1$ be the canonical future L--extension.  Let $\overline{M}_{2} = M \cup \partial^{+} M_2$ be any other future L--extension of $\left(M,\mathcal{C}\right)$, then the identity map $\mathrm{id}:M\rightarrow M$ can be extended as a diffeomorphism $\overline{\mathrm{id}}:\overline{M}_1 \rightarrow \overline{M}_2$.
\end{theorem}


\begin{proof}

By Lemma \ref{Lemma-equal-skies}, we have $S\left(\infty_{1}^{+}\left(\gamma\right)\right) = S\left(\infty_{2}^{+}\left(\gamma\right)\right)$ for all $\gamma\in \mathcal{N}$, then it is possible to define a bijection $\phi:\partial^{+}M_1 \rightarrow \partial^{+}M_2$ by $\phi\left(\infty^{+}_1\left(\gamma\right)\right)=\infty^{+}_2\left(\gamma\right)$. 
Then, we have the following diagram
\begin{equation}\label{diagram-submersions}
\begin{tikzpicture}[every node/.style={midway}]
\matrix[column sep={4em,between origins},
        row sep={2em}] at (0,0)
{ ; &  \node(N)   { $\mathcal{N}$}  ; & ; \\
  \node(M1)   { $\partial^{+}M_1$} ; &   ; & \node(M2)   { $\partial^{+}M_2$} ;       \\};
\draw[->] (N) -- (M1) node[anchor=south east]  {$\infty^{+}_{1}$};
\draw[->] (N) -- (M2) node[anchor=south west]  {$\infty^{+}_{2}$};
\draw[->] (M1)   -- (M2) node[anchor=north] {$\phi$};
\end{tikzpicture}
\end{equation}
and since $\infty^{+}_{1}$ is a submersion and $\infty^{+}_{2}=\phi\circ \infty^{+}_{1}$ is differentiable, by \cite[Prop.~6.1.2]{BC}, then $\phi$ is differentiable. Also, since $\infty^{+}_{2}$ is submersion and $\infty^{+}_{1}=\phi^{-1}\circ \infty^{+}_{2}$ is differentiable then $\phi^{-1}$ is differentiable, hence $\phi$ is a diffeomorphism.

Now, let us show that every map $\Psi_{\mathcal{U}}\left(\gamma,\mathbf{t}\right)=\overline{\gamma}\left(\mathbf{t}\right)$ with $\left(\gamma,\mathbf{t}\right)\in \mathcal{U}\times \left(-1,1\right]$ is a submersion. 
Clearly, since $\left.\Psi_{\mathcal{U}}\right|_{\mathcal{U}\times \left(-1,1\right)}=\left.\rho\circ\varepsilon\right|_{\mathcal{U}\times \left(-1,1\right)}$ where $\rho$ is the submersion of diagram (\ref{diagram-N-tilde}) and $\varepsilon$ is the diffeomorphism (\ref{eq-varepsilon}), then $\left.\Psi_{\mathcal{U}}\right|_{\mathcal{U}\times \left(-1,1\right)}$ is a submersion. 

On the other hand, observe that the restriction of $\Psi_{\mathcal{U}}$ to $\mathcal{U}\times \{1\}$ verifies that $\left.\Psi_{\mathcal{U}}\right|_{\mathcal{U}\times \{1\}}=\left.\phi\right|_{\infty^{+}_{1}\left(\mathcal{U}\right)}$ and since $\phi$ is a diffeomorphism and 
\[
\left(d\Psi_{\mathcal{U}}\right)_{\left(\gamma,1\right)}\left(  \frac{\partial}{\partial \mathbf{t}}  \right)_{\left(\gamma,1\right)}=\overline{\gamma}'\left(1\right)
\]
with $0\neq\overline{\gamma}'\left(1\right)\notin T_{\overline{\gamma}\left(1\right)} \partial^{+} M_2$, then $\left(d\Psi_{\mathcal{U}}\right)_{\left(\gamma,1\right)}$ is surjective, then we get that $\Psi_{\mathcal{U}}$ is a submersion.

Let us denote 
\[
\begin{tabular}{l}
$\overline{V}_1= \rho\circ\varepsilon\left(\mathcal{U}\times \left(-1,1\right]\right)\subset \overline{M}_1$ \\
$\overline{V}_2=\Psi_{\mathcal{U}}\left(\mathcal{U}\times \left(-1,1\right]\right)\subset \overline{M}_2$ \\
$V_1= \overline{V}_1 \cap M$ \\
$V_2=\overline{V}_2 \cap M$
\end{tabular}
\]
then, $V=V_1=V_2=\{\gamma\left(\mathbf{t}\right)\in M: \gamma\in \mathcal{U} \}$.
So, the following diagram 
\begin{equation}\label{diagram-submersions-2}
\begin{tikzpicture}[every node/.style={midway}]
\matrix[column sep={4em,between origins},
        row sep={2em}] at (0,0)
{ ; &  \node(N)   { $\mathcal{U}\times \left(-1,1\right]$}  ; & ; \\
  \node(M1)   { $\overline{M}_1\supset\overline{V}_1$} ; &   ; & \node(M2)   { $\overline{V}_2\subset\overline{M}_2$} ;       \\};
\draw[->] (N) -- (M1) node[anchor=south east]  {$\rho\circ \varepsilon$};
\draw[->] (N) -- (M2) node[anchor=south west]  {$\Psi_{\mathcal{U}}$};
\draw[->] (M1)   -- (M2) node[anchor=north] {$\overline{\mathrm{id}}$};
\end{tikzpicture}
\end{equation}
defines $\overline{\mathrm{id}}$ as a bijection such that $\left.\overline{\mathrm{id}}\right|_{V}=\mathrm{id}:V\rightarrow V$ is the identity map. Using again \cite[Prop.~6.1.2]{BC} as before, since $\rho\circ \varepsilon$ and $\Psi_{\mathcal{U}}$ are submersions, then $\overline{\mathrm{id}}$ is a diffeomorphism extending the identity map in $V\subset M$. 
Taking a covering $\{\mathcal{U}_{\alpha} \}_{\alpha\in I}\subset \mathcal{N}$ with their corresponding maps $\{\Psi_{\mathcal{U}_{\alpha}}\}_{\alpha\in I}$, we can define globally $\overline{\mathrm{id}}=\overline{M}_1\rightarrow \overline{M}_2$ as a diffeomorphism. 
This concludes the proof.
  \end{proof}

Observe that if $\psi: \left(M_1,\mathcal{C}_1\right) \rightarrow \left(M_2,\mathcal{C}_2\right)$ is a conformal diffeomorphism, then there is a diffeomorphisms $\Psi:\mathcal{N}_1 \rightarrow \mathcal{N}_2$ preserving skies, that is, for any sky $X\in \Sigma_1$ of $\mathcal{N}_1$, then $\Psi\left(X\right)\in \Sigma_2$ is a sky of $\mathcal{N}_2$. 
Then, in virtue of Theorem \ref{Theorem-diff-struct} and the way of construction of the L--extension of section \ref{sec:canon-ext}, if one of the conformal manifolds is a proper L--spacetime for any metric in $\mathcal{C}$, then the other is also a proper L--spacetime for any metric and both L--extensions are diffeomorphic by the extension of $\psi$.
This is summarized in the following corollary.

\begin{corollary}
Let $\left(M_1,\mathcal{C}_1\right)$ and $\left(M_2,\mathcal{C}_2\right)$ be conformal manifolds such that $\left(M_1,\mathcal{C}_1\right)$ is a proper L--spacetime for any metric in $\mathcal{C}_1$ and there exists a conformal diffeomorphism $\psi: \left(M_1,\mathcal{C}_1\right) \rightarrow \left(M_2,\mathcal{C}_2\right)$, then $\left(M_2,\mathcal{C}_2\right)$ is a proper L--spacetime for any metric in $\mathcal{C}_2$. 

Moreover, if $\overline{M}_i=M_i \cup \partial M_i$ for $i=1,2$ are the corresponding L--extensions, then the map $\partial\psi:\partial M_1 \rightarrow \partial M_2$ defined by $\partial \psi\left(\infty_1\left(\gamma\right)\right)=\infty_2\left(\Psi\left(\gamma\right)\right)$ is a diffeomorphism, where $\Psi:\mathcal{N}_1 \rightarrow \mathcal{N}_2$ is the diffeomorphism preserving skies between the corresponding spaces of light rays. 
In addition, the extension $\overline{\psi}:\overline{M}_1\rightarrow \overline{M}_2$, such that $\left.\overline{\psi}\right|_{M_1}=\psi$ and $\left.\overline{\psi}\right|_{\partial M_1}=\partial \psi$, is a diffeomorphism.
\end{corollary}

The transversality to the L--boundary of the extension of any light ray is a key feature of L--extensions.
The next example shows the existence of extensions, constructed in a natural way, which they are not L--extensions because the lack of transversality of light rays at the boundary.

\begin{example}
Let us consider $M=\left\{ \left(t,x,y\right)\in \mathbb{R}^3 : t>0 \right\}$ equipped with the metric $\mathbf{g} = - \frac{1}{t}dt\otimes dt + dx\otimes dx + dy \otimes dy $.
$M$ has a natural extension given by $\overline{M}= \left\{ \left(t, x,y\right)\in \mathbb{R}^3 : t\geq 0 \right\}$, where the past boundary $\partial^{-} M=\left\{ \left(t,x,y\right)\in \mathbb{R}^3 : t=0 \right\}$ has the standard differentiable structure. 

Fix the Cauchy surface $C=\{ \left(1,x,y\right)\in M\}$ to obtain a coordinate chart for its space of light rays $\mathcal{N}$ as in (\ref{carta-N}), then the null geodesic $\gamma$ such that $\gamma\left(0\right)=\left(1,x_0,y_0\right)\in C$ and $\gamma'\left(0\right)=\left(1,\cos \theta_0,\sin \theta_0\right)\in \mathbb{N}^{+}_{\gamma\left(0\right)}$ can be written by
\[
\gamma\left(s\right)=\left(\frac{1}{4}\left(s+2\right)^2 , x_0+s\cos \theta_0, y_0+s\sin \theta_0 \right) , \qquad s\in\left(-2,\infty\right)
\]
where it is possible to identify $\gamma$ by the coordinates $\left(x_0,y_0,\theta_0\right)$ in $\mathcal{N} = C \times \mathbb{S}^1$.
Repeating the computations performed in Sect. \ref{sec:block} we get that in the extension $\overline{M}$, each point $\left(0,u,v\right)\in \partial^{-}M$ corresponds to the integral curve of the distribution $\ominus$ given by the set of null geodesics:
\begin{equation}\label{example-no-L}
\gamma_{\theta}\left(s\right)=\left(\frac{1}{4}\left(s+2\right)^2, u+\left(s+2\right)\cos \theta, v+\left(s+2\right)\sin \theta  \right) , \qquad s\in\left(-2,\infty\right), \theta\in\left(-\pi, \pi\right] \, .
\end{equation}
Then the map $\infty^{-}:\mathcal{N} \rightarrow \partial^{-}M$ can be expressed in coordinates by
\[
\infty^{-}\left(x,y,\theta\right)=\left(0, x-2\cos \theta, y-2\sin \theta \right)\in \partial^{-}M \, .
\]
But observe that the extensions $\overline{\gamma}_{\theta}$ to the interval $s\in \left[-2,\infty\right)$ verify 
\[
\overline{\gamma}'_{\theta}\left(-2\right)=\lim_{s\mapsto -2^{+}}\overline{\gamma}'_{\theta}\left(s\right)=\left(0, \cos \theta, \sin \theta  \right) \, ,
\]
whence we obtain that $\overline{\gamma}'_{\theta}\left(-2\right)\in T_{\overline{\gamma}_{\theta}\left(-2\right)}C$, and therefore the transversality of the extended light rays to the boundary does not occur. This shows that $\overline{M}$ is not a L--extension.

In fact, since $\left(M,\mathbf{g}\right)$ is isometric to the hyperbolic block in Minkowski space $M_1 \cong \mathbb{M}^3(0,+\infty)$ (see the example in remark \ref{M3ainf}) by the transformation: 
\[
\left(t_1,x_1,y_1\right)= \left(2\sqrt{t},x,y\right) \, ,
\]
then the L--extension of $M$ can be obtained as the one of $\left(M_1,\mathbf{g}_1\right)$.  
Thus, we get that $\overline{M}_1=\left\{ \left(t_1,x_1,y_1\right)\in \mathbb{R}^3 : t_1\geq 0 \right\}$ is the L--extension of $M$, where its L--boundary: 
$$
\partial^{-} M_1=\left\{ \left(t_1,x_1,y_1\right)\in \mathbb{R}^3 : t_1=0 \right\} \, ,
$$ 
has the standard differentiable structure.
\end{example}

The following example shows that, in order to characterize the canonical L--extension, it is not possible to weaken condition \ref{L-ext-cond-1} of definition \ref{def-L-extension} assuming that the parametrizations is not admissible but regular.

\begin{example}
Consider the upper infinity block $\mathbb{M}^3(0,+\infty) = \left\{ \left(t,x,y\right)\in \mathbb{R}^3 : t>0 \right\}$ with the standard Minkowski metric $\mathbf{g}= -  dt\otimes dt + dx\otimes dx + dy \otimes dy$.
We obtain a coordinate chart for $\mathcal{N}$ using the Cauchy surface $C=\{  t=1 \}\subset \mathbb{M}^3(0,+\infty)$. 
The null geodesic $\gamma$ such that $\gamma\left(0\right)=\left(1,x_0,y_0\right)\in C$ and $\gamma'\left(0\right)=\left(1,\cos \theta_0,\sin \theta_0\right)\in \mathbb{N}^{+}_{\gamma\left(0\right)}$ can be written by
\[
\gamma\left(s\right)=\left(s+1,x_0+s\cos \theta_0, y_0+s\sin \theta_0  \right) , \qquad s\in\left(-1,\infty\right) \, ,
\]
then, using the standard chart in $\mathcal{N}$, we can identify $\gamma\simeq\left(x_0,y_0,\theta_0\right)$.

The parameter $s$ is admissible, hence the map 
\[
\Psi\left(\gamma, s\right)=\Psi\left(x_0,y_0,\theta_0,s\right) = \left(x_0+s\cos \theta_0, y_0+s\sin \theta_0 ,s+1 \right)\, , 
\]
with $s\in\left[-1,\infty\right)$ verifies condition \ref{L-ext-cond-1} of definition \ref{def-L-extension} and, as it was argued before, the past L--extension of $\mathbb{M}^3(0,+\infty)$ is $\left\{ \left(t,x,y\right)\in \mathbb{R}^3 : t\geq 0 \right\}$ with the standard differentiable structure.

Observe now, that $\mathbb{M}^3(0,+\infty)$ is isometric to $M=\left\{ \left(w,u,v\right)\in \mathbb{R}^3 : w>0 \right\}$ equipped with the metric $\mathbf{\overline{g}}= - w^2 dw\otimes dw + du\otimes du + dv \otimes dv $ by the isometry $\phi:\mathbb{M}^3(0,+\infty)\rightarrow M$ given by
\[
\left(w,u,v\right)=\phi\left(t,x,y\right)=\left(\sqrt{2t},x,y\right) \, ,
\]
then the past L--extension of $M$ must be $\overline{M}= \left\{ \left(w,u,v\right)\in \mathbb{R}^3 : w\geq 0 \right\}$ with the differentiable structure such that the extension of the isometry $\phi$ to $\partial^-\mathbb{M}^3(0,+\infty)$ is a diffeomorphism. 
Then, if we denote by $\overline{M}_{*}$ such differentiable manifold and by $\overline{M}_c$ the same topological manifold equipped with the standard differentiable structure, clearly the identity map $\overline{\mathrm{id}}:\overline{M}_{*}\rightarrow \overline{M}_c$ is not a diffeomorphism and therefore $\overline{M}_c$ is not the L--extension of $M$. 

Notice that 
\[
\overline{\gamma}\left(s\right)=\phi\left(\gamma\left(s\right)\right)=\left(\sqrt{2\left(s+1\right)} , x_0+s\cos \theta_0, y_0+s\sin \theta_0 \right) , \qquad s\in\left(-1,\infty\right) \, ,
\]
is an inextensible null geodesic in $M$. 
We can change the parameter by $\tau^2=s+1$, obtaining a regular parameter $\tau\in\left(0,\infty\right)$ (diffeomorphic to the canonical projective parameter $\mathbf{t}\in\left(-1,1\right)$ according to remark \ref{remark-parameters}).
The map $\overline{\Psi}:\mathcal{N}\times\left[0,\infty\right)\rightarrow \overline{M}_c$ defined in coordinates by
\[
\overline{\Psi}\left(x_0,y_0,\theta_0,s\right) = \left(x_0+\left(s^2-1\right)\cos \theta_0, y_0+\left(s^2-1\right)\sin \theta_0 ,\sqrt{2}s \right) \, ,
\]
satisfies all conditions of definition \ref{def-L-extension} except that $s\in\left[0,\infty\right)$ is not admissible but regular.
\end{example}

Definition \ref{def-L-extension} gives a characterization of L--extensions. 
Indeed, the following Proposition is a converse result of Theorem \ref{Theorem-diff-struct}.
The same result for past L--extensions can be shown in an analogous way.
 
\begin{proposition}
Let $M$ be a $3$--dimensional, strongly causal, light non--conjugate, sky-separating, conformal Lorentz manifold.
If $M$ admits a future L--extension $\overline{M}$, then the canonical field of directions $\oplus:\mathcal{N}\rightarrow \mathbb{P}\left(\mathcal{H}\right)$ defines a regular and smooth distribution. 
Moreover, the distribution $\boxplus$ defined by the L--extension verifies $\boxplus=\oplus$.
\end{proposition}

\begin{proof}
Let us denote by $\Psi:\mathcal{N}_U\times\left(-1,1\right]\rightarrow\overline{M}$ and $\infty^{+}:\mathcal{N}\rightarrow \partial^{+} M$ the parametrization and the surjective submersion involved in the definition (\ref{def-L-extension}) of L--extensions, and let $\boxplus:\mathcal{N}\rightarrow \mathbb{P}\left(\mathcal{H}\right)$ the distribution whose integral manifolds are the inverse images of $\infty^{+}$.
Since $\gamma\left(t\right)=\Psi\left(\gamma,t\right)$ runs a light ray and $t\in \left(-1,1\right]$ can be assumed to be a projective parameter in virtue of Lemma \ref{Lemma-canon-prject-param}, then 
\begin{equation}\label{eq-projectivity}
\widetilde{\gamma}\left(t\right)=\sigma\left( \left[ \frac{\partial \Psi}{\partial t}\left(\gamma,t\right) \right] \right) \, ,
\end{equation}
defines a projectivity in each fibre $\mathbb{P}\left(\mathcal{H}_{\gamma}\right)$ such that the map
\begin{equation}\label{eq-epsilon-psi-map}
\begin{tabular}{rccl}
$\varepsilon^{\Psi}\colon $ & $\mathcal{N}_U \times \left(-1,1\right)$ & $\rightarrow$ & $\mathbb{P}\left(\mathcal{H}\right)$ \\
& $\left(\gamma,t\right)$ & $\mapsto$ & $\widetilde{\gamma}\left(t\right)$ \, ,
\end{tabular}
\end{equation}
is differentiable by composition, because $\Psi$ and $\sigma$ are smooth and $t$ is regular.
But observe that any projectivity in the fibre $\mathbb{P}\left(\mathcal{H}_{\gamma}\right)$ is completely determined when three values are given and, in fact, we have already defined the projectivity of $\gamma\left(t\right)$ for $t\in\left(-1,1\right)$ depending smoothly on $\gamma$ according to the differentiability of the map (\ref{eq-epsilon-psi-map}). 
Then for any $\gamma\in \mathcal{N}_{U}$ the projectivity $\widetilde{\gamma}:\mathbb{R}\rightarrow \mathbb{P}\left(\mathcal{H}_{\gamma}\right)$ is automatically defined and it permits to extend the map $\varepsilon^{\Psi}$ in a smooth way as
\begin{equation}\label{eq-epsilon-psi-map-ext}
\begin{tabular}{rccl}
$\varepsilon^{\Psi} \colon$ & $\mathcal{N}_U \times \mathbb{R}$ & $\rightarrow$ & $\mathbb{P}\left(\mathcal{H}\right)$ \\
& $\left(\gamma,t\right)$ & $\mapsto$ & $\widetilde{\gamma}\left(t\right)$ \, ,
\end{tabular}
\end{equation}
Since the map $\varepsilon^{\Psi}_{t_0}:\mathcal{N}_{U}\rightarrow \mathbb{P}\left(\mathcal{H}\right)$ given by $\varepsilon^{\Psi}_{t_0}\left(\gamma\right)=\varepsilon^{\Psi}\left(\gamma,t_0\right)=\widetilde{\gamma}\left(t_0\right)$ is a local smooth section of the fibre bundle $\pi_{\mathcal{N}}^{\mathbb{P}\left(\mathcal{H}\right)}:\mathbb{P}\left(\mathcal{H}\right)\rightarrow \mathcal{N}$ due to $\pi_{\mathcal{N}}^{\mathbb{P}\left(\mathcal{H}\right)}\circ \varepsilon^{\Psi}_{t_0}\left(\gamma\right) = \pi_{\mathcal{N}}^{\mathbb{P}\left(\mathcal{H}\right)}\left(\widetilde{\gamma}\left(t_0\right)\right)= \gamma$, then 
\[
\varepsilon^{\Psi}_{t_0}:\mathcal{N}_{U}\rightarrow \varepsilon^{\Psi}_{t_0}\left(\mathcal{N}_{U}\right)\subset\mathbb{P}\left(\mathcal{H}\right) \, ,
\]
is a diffeomorphism.
So, taking $t=1$, the map $\varepsilon^{\Psi}_{1}$ is the diffeomorphism onto its image such that $\varepsilon^{\Psi}_{1}\left(\gamma\right)=\widetilde{\gamma}\left(1\right)$. 
By continuity of $\Psi$ and equation (\ref{eq-projectivity}), we have that $\widetilde{\gamma}\left(1\right)=\boxplus\left(\gamma\right)$, but since $\varepsilon^{\Psi}$ coincides with the canonical $\varepsilon$ map for $t\in \left(-1,1\right)$, then by continuity 
\[
\boxplus\left(\gamma\right)=\varepsilon^{\Psi}\left(\gamma,1\right)= \varepsilon\left(\gamma,1\right)=\oplus\left(\gamma\right) \, ,
\]
and since $\boxplus$ is regular and smooth, therefore so $\oplus$ is.
  \end{proof}

\begin{remark}
We have studied L--extensions in the case $\oplus_{\gamma} \neq \ominus_{\gamma}$ for all $\gamma\in \mathcal{N}$, but there are simple examples (such as $3$--dimensional Minkowski spacetime) in which $\oplus = \ominus$ (see \cite[Sec.~IV.C]{Ba17}). 
In these cases, the projective parameter of the canonical future L--extension can be obtained by choosing two (local) smooth spacelike Cauchy surfaces $C_0, C_{-1}\subset M$ such that every point $q\in C_{-1}$ is in the chronological past of $C_0$. 
Then, the projective parameter $\mathbf{t}$ verifies that $\gamma\left(-1\right)\in C_{-1}$, $\gamma\left(0\right)\in C_{0}$ and $\gamma\left(1\right)\in \partial^{+}M$. 
With this parameter, the construction of the canonical future L--extension of $M$ is done in the same way as in section \ref{sec:canon-ext}. 
The canonical past L--extension can be built in an analogous way.
\end{remark}


\section{Discussion and conclusions}\label{sec:discussion}

It has been shown that for a class of 3-dimensional spacetimes $M$ a new causal boundary, called L-bounday, can be constructed that defines a smooth extension of the original spacetime.  The construction of the new boundary is explicit and intrinsically conformal invariant.    It uses in a direct way the space of light rays $\mathcal{N}$ with its contact structure $\mathcal{H}$ and the natural projective bundle $\mathbb{P}(\mathcal{H})$ over it.      

Moreover a class of extensions of spacetimes, called L-extensions has been introduced, their properties defined exclusively in terms of local properties of the corresponding boundary points, that encode the transversality properties of the light rays accessing to them.   It has been shown that such L-extensions, if they exist, are essentially unique and that the canonical extension defined by the L-boundary is an L-extension.    It remains to analyse the relation of L-extensions and conformal envelopments, a key notion to investigate the relation of the L-boundary with the conformal boundary of a given spacetime.  This problem will be discussed in a forthcoming work.

Even if some of the constructions has been done in the realm of 3-dimensional spacetimes, the results can be extended naturally to higher dimensions.    All basic ingredients needed in the detailed proofs are available in higher dimensional spacetimes.   Most conspicuous is the projective parameter used to prove the smoothness of the canonical extension.   Actually, one of the main reasons to restrict ourselves in this presentation to three dimensionas was that in such case, the projective parameter is naturally defined because it is the natural projective parameter on the fibres of the projective bundle $\mathbb{P}(\mathcal{H})$ (which are projective circles).  In higher dimensions a technical construction is needed to obtain such projective parameter that requires using an adapted Fermi-Walker connection and using the associated affine parameter.  The details of such constructions and new significative examples will be discussed elsewhere.

Finally we would like to comment that a discussion on the detailed relation between the proposed L-boundary and conformal boundaries is still missing.   While the relation between the $L$-boundary and the causal boundary was explored in \cite{Ba17}, the relation between the $L$-boundary and conformal extensions of the given space-time has not been addressed yet. We believe that the L-boundary provides an intermediate step in between the conformal boundary, an ad-hoc construction but immediately available, and the causal boundary, the fundamental abstract construction of ideal points at infinity which is hard to describe, if not impossible, in concrete situations.  The clarification of such issues will be the subject of forthcoming work.


\section*{Acknowledgements}
The authors wish to acknowledge the referee for the careful revision of the manuscript, the many suggestions that have helped to improve the paper and by pointing out a difficulty with the regularity of the total distribution on the blow up space that has been considered.  Financial support from the Spanish Ministry of Economy and Competitiveness, through the Severo Ochoa Programme for Centres of Excellence in RD (SEV-2015/0554) is acknowledged.
AI would like to thank partial support provided by the MINECO research project MTM2017-84098-P and QUITEMAD+, S2013/ICE-2801.


\phantomsection

\end{document}